\documentclass[reqno]{amsart}
\usepackage{amssymb,latexsym,upref,enumerate,fouridx}
\usepackage{bbm}
\usepackage[normalem]{ulem}
\usepackage{comment}
\usepackage{mathrsfs,color}
\usepackage{centernot}
\usepackage[colorlinks,linkcolor=blue,citecolor=blue]{hyperref}
\usepackage{graphicx}
\usepackage{cite}
\newcommand{\added}{}

\usepackage[displaymath,mathlines]{lineno}
\usepackage{tikz} 

\allowdisplaybreaks

\numberwithin{equation}{section}

\theoremstyle{plain}
\newtheorem{thm}{Theorem}[section]
\newtheorem{lem}[thm]{Lemma}
\newtheorem{prop}[thm]{Proposition}

\theoremstyle{definition}

\theoremstyle{remark}
\newtheorem{rem}[thm]{Remark}

\newcommand{\ep}{\epsilon}

\newcommand{\tauac}{\acute{\tau}{}}
\newcommand{\gac}{\acute{g}{}}

\newcommand{\Qft}{\tilde{\mathfrak{Q}}{}}

\newcommand{\Jcch}{\check{\mathcal{J}}{}}

\newcommand{\chihu}{\underline{\hat{\chi}}{}}
\newcommand{\ghu}{\underline{\gh}{}}
\newcommand{\whu}{\underline{\wh}{}}
\newcommand{\vhu}{\underline{\vh}{}}

\newcommand{\zhu}{\underline{\zh}{}}
\newcommand{\hhu}{\underline{\hh}{}}
\newcommand{\Bhu}{\underline{\Bh}{}}
\newcommand{\Qhu}{\underline{\Qh}{}}
\newcommand{\Gammahu}{\underline{\hat{\Gamma}}{}}

\newcommand{\Gammah}{\hat{\Gamma}{}}

\newcommand{\gchat}{\hat{\gc}{}}

\newcommand{\chih}{\hat{\chi}{}}

\newcommand{\betah}{\hat\beta}

\newcommand{\taur}{\mathring{\tau}{}}
\newcommand{\chir}{\mathring{\chi}{}}

\newcommand{\taug}{\grave{\tau}{}}
\newcommand{\er}{\mathring{e}{}}
\newcommand{\gr}{\mathring{g}{}}
\newcommand{\lr}{\mathring{l}{}}
\newcommand{\Jcr}{\mathring{\mathcal{J}}{}}
\newcommand{\vr}{\mathring{v}{}}

\newcommand{\rhor}{\mathring{\rho}{}}

\newcommand{\Kttt}{\tilde{\Ktt}{}}
\newcommand{\gttt}{\tilde{\gtt}{}}
\newcommand{\Nttt}{\tilde{\Ntt}{}}
\newcommand{\bttt}{\tilde{\btt}{}}

\newcommand{\ggr}{\grave{g}{}}
\newcommand{\taugr}{\grave{\tau}{}}

\newcommand{\udim}{{\mathcal n}}%{{n^4+2 n^3-3 n^2+n+1}}

\newcommand{\Bcv}{\boldsymbol{\Bc}{}}

\newcommand{\cnormal}{n}
\newcommand{\cspatmetr}{\gtt}
\newcommand{\cextcurv}{\Ktt}
\newcommand{\cacc}{\dot n}
\newcommand{\cspnabla}{D}
\newcommand{\cexp}{\Htt}
\newcommand{\cshear}{\sigma}
\newcommand{\cshearresc}{\Sigma}
\newcommand{\cspatcurv}{\Rtt}

\newcommand{\pnormal}{\bar{\cnormal}}
\newcommand{\pspatmetr}{\bar{\cspatmetr}}
\newcommand{\pextcurv}{\bar\cextcurv}
\newcommand{\pacc}{\bar{\cacc}}
\newcommand{\pspnabla}{\bar{\cspnabla}}
\newcommand{\pexp}{\bar\cexp}
\newcommand{\pshear}{\bar\cshear}
\newcommand{\pshearresc}{\bar\cshearresc}
\newcommand{\pspatcurv}{\bar\cspatcurv}
\newcommand{\pWeyl}{\bar{\mathtt{W}}}

\input Tdef.def

\renewcommand{\emph}[1]{\textit{#1}}

\newcounter{mnotecount}[section]
\let\oldmarginpar\marginpar
\renewcommand\marginpar[1]{\-\oldmarginpar[\raggedleft\footnotesize #1]%
 {\raggedright\footnotesize #1}}

\begin{document}

\title[Big bang stability and isotropisation in the ekpyrotic regime]{Big bang stability and isotropisation for the Einstein-scalar field equations in the ekpyrotic regime}

\author[F. Beyer]{Florian Beyer}
\address{Dept of Mathematics and Statistics\\
730 Cumberland St\\
University of Otago, Dunedin 9016\\ New Zealand}
\email{florian.beyer@otago.ac.nz }

\author[D. Garfinkle]{David Garfinkle}
\address{Deptartment of Physics, Oakland University, Rochester, MI 48309, USA}
\email{garfinkl@oakland.edu}

\author[J. Isenberg]{James Isenberg}
\address{Department of Mathematics and Institute for Fundamental Science, University of Oregon, Eugene,OR 97403, USA}
\email{isenberg@uoregon.edu}

\author[T.A. Oliynyk]{Todd A. Oliynyk}
\address{School of Mathematical Sciences\\
9 Rainforest Walk\\
Monash University, VIC 3800\\ Australia}
\email{todd.oliynyk@monash.edu}

\begin{abstract}
It has been shown that, in spacetime dimensions $n\geq 3$, that the Kasner-scalar field solutions to the Einstein-scalar fields equations with potential $V_0 e^{-s \phi}$, where $s<s_c=\sqrt{\frac{8(n-1)}{n-2}}$ and $V_0\in \Rbb$, are nonlinearly stable to the past and terminate at a quiescent big bang singularity over the full range of sub-critical Kasner exponents. In particular, the spatially homogeneous and isotropic solutions, the Friedman-Lemaitre-Robertson-Walker (FLRW) spacetimes, to the Einstein-scalar field equations are stable in this sense for $s<s_c$ and $V_0\in\Rbb$.

While perturbations of the sub-critical Kasner-scalar field family of solutions, including the FLRW solutions, are asymptotically velocity term dominated (AVTD) near the big bang, they do not in general isotropise near the big bang singularity. Rather, they remain highly anisotropic, even for small perturbations of the isotropic FLRW solutions.

In this article, we establish the stability of FLRW solutions to the Einstein-scalar field equations for the potential parameter values $s>s_c$ and $V_0<0$. Such scalar field potentials are known in the literature as \textit{ekpyrotic}. In particular, we prove that the FLRW solutions to the Einstein-scalar field equations are nonlinearly stable to the past and terminate at a quiescent, crushing AVTD big bang singularity. A distinguishing property of these perturbed spacetimes is that they isotropise towards the big bang. 
\end{abstract}

\maketitle

\section{Introduction}

\subsection{Background and motivation}
\label{sec:introduction}
Cosmological solutions of the Einstein--matter system have been studied for decades,
yet the structure of their singularities remains far from fully understood. In this
paper, we explore the behavior of the gravitational field near such singularities in the \emph{cosmological setting} of globally hyperbolic spacetimes
with compact Cauchy surfaces \cite{bartnik1988}: Our aim is to understand the
backward-in-time asymptotics as one approaches ``the beginning'', that is, the behaviour of the gravitational field in a neighbourhood of a \emph{big bang singularity}\footnote{Without loss of generality, we orient time so
that the contracting direction corresponds to the past.} in the finite past.

Consistent with Hawking's
singularity theorem \cite{hawkingLargeScaleStructure1973} (See also
\cite[Chapter~14, Theorem~55A]{oneill1983a}.) and with the influential heuristics of Belinskii-Khalatnikov-Lifshitz (BKL)
\cite{belinskii1970,lifshitz1963} regarding {generic singularities}, one expects
big bang formation to persist for broad classes of solutions to the Einstein--matter
equations. These heuristics are commonly summarised in terms of the \emph{BKL conjecture},
which predicts that cosmological singularities should be generically spacelike, local, and oscillatory.
Support for the BKL picture comes from several directions, including rigorous work
in the spatially homogeneous setting
\cite{BeguinDutilleul:2023,beguin:2010,liebscher_et_al:2011,ringstrom2001a} and
numerical studies beyond spatial homogeneity
\cite{anderssonAsymptoticSilenceGeneric2005,curtis2005,garfinkle2002,garfinkle2002a,garfinkle2004,garfinkle2007,berger2001,marshall2025a}.
At the same time, developments involving spikes
\cite{berger1993,coley2015,lim2008,lim2009,rendall2001} and weak null singularities
\cite{dafermos2017,luk2017} indicate that the BKL picture is incomplete. Nevertheless,
it remains a useful heuristic guide for many investigations in mathematical cosmology.

It is generally believed that a minimally coupled scalar field (or, more generally, a \emph{stiff} fluid)  suppresses all oscillations suggested by the BKL conjecture and leads to \emph{convergent}, \emph{asymptotically velocity term dominated} (AVTD) asymptotics \cite{eardley1972,isenberg1990} near the big bang. This property is of physical interest because it leads to monotonic behaviour of the gravitational and the matter fields near the big bang. As well, if AVTD is present in a family of cosmological solutions, obtaining a proof of strong cosmic censorship \cite{penrose1969,Penrose:1980ge} is simplified \cite{CIM1990}. In addition, scalar fields are widely believed to play a significant role in the dynamics of the early universe \cite{mukhanov2005}. These considerations provide strong motivation for a detailed analysis of these equations and the qualitative behaviour of their solutions.

The Einstein-scalar field equations take the form
\begin{align}
  \Rb_{ij}&=2\nablab_i\phi\nablab_j \phi+\frac{4}{n-2}V(\phi)\gb_{ij}, 
  \label{ESF.1}\\
  \Box_{\gb} \phi &=V'(\phi), \label{ESF.2}
\end{align}
where $\gb_{ij}$ is the physical $n$-dimensional spacetime metric ($n>2$),
$\nablab_i$, $\Rb_{ij}$, and $\Box_{\gb} =\gb^{ij}\nablab_i\nablab_j$ are, respectively, the Levi-Civita connection, the Ricci tensor, and the wave operator associated to $\gb_{ij}$, 
$\phi$ is the scalar field, and $V(\phi)$ is the scalar field potential. In our convention, the energy-momentum tensor of the scalar field is 
\begin{equation}
  \label{eq:sfEMT}
  \bar T_{ij}=2\Bigl(\nablab_i\phi\nablab_j \phi-\frac 12\nablab_k\phi\nablab^k \phi\gb_{ij}-V(\phi)\gb_{ij}\Bigr).
\end{equation}

In this article, we restrict our attention to scalar field potentials  $V(\phi)$ of the form
\begin{equation}
  \label{eq:exponentialpotential}
  V(\phi)=V_0 e^{-s\phi},
\end{equation}
where $s$ and $V_0$ are here arbitrary real constant parameters. Without loss of  generality, we assume $V_0$ is normalised by 
\begin{equation}
  \label{eq:V0normalisation}
  V_0\in\{-1,0,1\}.
\end{equation} 
To see why this is the case, if we suppose that $\{\gb_{ij},\phit\}$ is a solution of \eqref{ESF.1}--\eqref{eq:exponentialpotential} with the potential $V(\phit)=V_0 e^{-s\phit}$, then it is straightforward to check that $\{\gb_{ij},\phi:=\phit-s^{-1}\ln(|V_0|)\}$ solves
\eqref{ESF.1}--\eqref{eq:exponentialpotential}
with the potential $V(\phi)=\frac{V_0}{|V_0|} e^{-s\phi}$. Moreover, we note that
$(\gb_{ij},\phi:=-\phit)$ is a solution of \eqref{ESF.1}--\eqref{eq:exponentialpotential} with the potential $V(\phi)=V_0e^{-(-s)\phi}$. 

It is common in the cosmology literature (e.g., \cite{mukhanov2005}) to appeal to the correspondence between  scalar fields $\phi$ and irrotational perfect fluids, which holds as long as
$\bar\nabla_i\phi$ is timelike. The correspondence follows from comparing the scalar field energy momentum tensor \eqref{eq:sfEMT} with that of a perfect fluid, and this correspondence determines the fluid velocity vector $\utt_i$, the density $\mu$ and the pressure $P$ according to
\begin{equation} \label{utt-def}
  \utt_i=\frac{\nablab_i\phi}{\sqrt{-\nablab_k\phi\nablab^k \phi}},\quad
  \mu=-\nablab_k\phi\nablab^k \phi+2V(\phi) \AND 
  P=-\nablab_k\phi\nablab^k \phi-2V(\phi).
\end{equation}
While the type of fluid determined by this correspondence can have unusual properties if $V(\phi)$ is not zero \cite{faraoni2023,piattella2014}, it is standard to introduce
the \emph{effective equation of state parameter} $w$, which is given by
\begin{equation}
  \label{eq:effectiveequationofstate}
  w=\frac{P}{\mu}=\frac{-\frac 12\nablab_k\phi\nablab^k \phi-V(\phi)}{-\frac 12\nablab_k\phi\nablab^k \phi+V(\phi)}.
\end{equation}
This quantity is especially useful in asymptotic regimes if $w$ is approximately constant and the dynamical PDEs are dominated by effective ODEs, which define the AVTD regime.

Recent work has led to substantial progress regarding the nonlinear past stability of FLRW solutions to the Einstein--scalar field equations and their AVTD big bang asymptotics, beginning with \cite{rodnianski2014, rodnianski2018}. The analysis has been subsequently extended in \cite{fournodavlos2020b}, where AVTD big bang stability had been established for nonlinear perturbations of the subcritical Kasner family. Background-independent AVTD big bang formation for a broad class of Einstein--scalar field models was then obtained in the series of works \cite{ringstrom2021,ringstrom2021a,ringström2022,ringström2022a,franco-grisales2026,franco-grisales2024}, especially in \cite{groeniger2023} making use of the ideas in \cite{fournodavlos2020b}.
In a complementary line of research, spatial localisation techniques were developed for the big bang stability problem in \cite{BeyerOliynyk:2021,BeyerOliynykZheng:2025,zheng2026,franco-grisales2026,athanasiou2024}. 
For completeness, we note that AVTD behaviour is not restricted to scalar-field models. Even in the absence of a scalar field, it has been found to occur, particularly within certain symmetry classes of vacuum spacetimes \cite{ames2021, ames2021a,chruscielStrongCosmicCensorship1990,isenberg1990,fournodavlos2020b,ringstrom2009a, choquet-bruhat2004,ames2019,ames2013a,andersson2001,BeyerLeLoch:2017,choquet-bruhat2006,klinger2015,clausen2007,damour2002,fournodavlos2020,heinzle2012,isenberg1999,isenberg2002,kichenassamy1998,stahl2002,athanasiou2024,franco-grisales2024, dong2026}, and it can also occur in the presence of alternative matter fields \cite{an2025a, beyerStabilityFLRWSolutions2023, fajman2025a}.

The works cited above concerning the Einstein--scalar field equations treat either the \emph{free massless scalar field} case, that is, $V_0=0$ in \eqref{eq:exponentialpotential}, or the case $V_0\not=0$ in which the (steepness) parameter $s$ in the potential \eqref{eq:exponentialpotential} satisfies 
\begin{equation*}
s<s_c:=\sqrt{\frac{8(n-1)}{n-2}}.
\end{equation*}
In this regime, scalar field dynamics near the big bang singularity are characterised by the dominance of the kinetic energy over the potential energy. As a consequence, the effective equation of state parameter $w$ in \eqref{eq:effectiveequationofstate} approaches $w=1$ near the singularity, and the scalar field asymptotically behaves like an irrotational stiff fluid. 

Heuristically, in evolving backward in time for this regime, the scalar field ``rolls up'' the exponentially growing potential barrier\footnote{For definiteness, we assume that $V_0$ is positive, that $0<s<s_c$, and that $\phi$ tends to negative infinity.}. More specifically, the potential is not ``too steep,'' and the gravitational effects are sufficiently strong to drive the scalar field uphill while its kinetic energy continues to dominate the potential energy. We refer to the parameter range {$s<s_c$} as the \emph{Kasner regime}. As discussed in detail in Section~\ref{sec:FLRWEinsteinSF}, this terminology reflects the fact that, in this regime, the spacetime geometry is generically anisotropic and exhibits pointwise Kasner-like behaviour near the big bang.

On the other hand, if $s>s_c$, then the potential is too steep for the characteristic Kasner dynamics described above to persist, and one enters the \emph{ekpyrotic regime}. In this article, we investigate the behaviour of solutions to the Einstein--scalar field system near big bang singularities in this regime, beginning in Section~\ref{sec:FLRWEinsteinSF} with spatially isotropic and homogeneous FLRW solutions. Historically, ekpyrotic spatially homogeneous Einstein--scalar field solutions were first proposed in \cite{khoury2001} as an alternative to the inflationary paradigm and as a mechanism to address fundamental problems in standard cosmology, including the flatness and horizon problems. A striking feature of ekpyrotic dynamics was identified in \cite{erickson2004} and was referred to as the ``cosmic no-hair for contracting universes,'' in analogy with the classical cosmic no-hair results for expanding spacetimes \cite{waldAsymptoticBehaviorHomogeneous1983}. In contrast to the Kasner regime, the ekpyrotic setting provides a dynamical mechanism that suppresses anisotropies.

The analysis of \cite{erickson2004} was subsequently placed within a dynamical-systems framework in \cite{alho2022,lidsey2006}, and ekpyrotic solutions without symmetry assumptions were investigated numerically in \cite{garfinkle2008}. These studies indicate that ekpyrotic solutions typically enter an AVTD regime near the big bang, as in the Kasner case, but with the effective equation of state parameter $w$ approaching  a constant strictly greater than one, rather than tending to $w=1$ as in the Kasner regime. In light of \eqref{eq:effectiveequationofstate}, this behaviour reflects the fact that the scalar field potential and kinetic energies remain comparable in magnitude near the singularity in these ekpyrotic solutions.

From this perspective, ekpyrotic solutions are expected to resemble, near the big bang, barotropic irrotational fluids with a speed of sound exceeding the speed of light and the speed of gravitational radiation. A broad class of analytic solutions to the Einstein–fluid equations with superluminal sound speeds was constructed in \cite{heinzle2012} by solving a singular initial value problem where asymptotic data is prescribed at the big bang singularity. Phenomena analogous to those observed in numerical studies of ekpyrotic solutions were found. The fluids in such solutions are referred to as \textit{ultrastiff fluids}.

The primary aim of this article is to mathematically establish the past nonlinear stability of spatially flat
FLRW Einstein-scalar field spacetimes and their big bang singularities 
over the full ekpyrotic parameter range  $s>s_c$. 
Our main results are presented in Theorem \ref{glob-stab-thm}.

\subsection{Spatially flat FLRW Einstein-scalar field spacetimes}
\label{sec:FLRWEinsteinSF}
The class of the FLRW Einstein-scalar field spacetimes consists of spatially homogeneous and isotropic solutions to the Einstein--scalar field equations \eqref{ESF.1}--\eqref{ESF.2}. In the present work, we restrict our attention to the FLRW models with spatially flat hypersurfaces and we analyse both the Kasner and the ekpyrotic regimes, emphasising the qualitative differences between them. As noted above, the primary objective of this article is to prove the nonlinear past stability of the spatially flat ekpyrotic-FLRW spacetimes.

\subsubsection{Hubble normalised variables}
To analyse the spatially flat FLRW spacetimes, we adopt the dynamical-systems approach of \cite{beyer2013}; see also \cite{alho2022} for a more general dynamical-systems treatment. We employ the  \emph{Hubble-normalised}, symmetry-adapted physical orthonormal frame variables and  we define the \emph{Hubble time} $\tilde t$ in terms of the Gaussian time coordinate $\tb$ by\footnote{We use the standard notation $\dot{}=d/d\tb$ and $'=d/d{\tilde t}$.}
\begin{equation}
  \label{eq:defHubbletime}
  \frac{d{\tilde t}}{d\tb}={\bar\Htt},
\end{equation}
where ${\bar\Htt}$ denotes the physical mean curvature of the Gaussian time slices with respect to the future-pointing timelike unit normal $\partial_{\tb}$ (cf.\ Section~\ref{spacetimedecomp}). Since we are focusing on spacetimes that collapse toward the past, we restrict attention to the case of strictly positive ${\bar\Htt}$.

Following\footnote{The definitions of the variables $\xtt$ and $\ytt$ in \eqref{eq:defxy} and their evolution equations \eqref{eq:FLRWevoleq.11}--\eqref{eq:FLRWq} differ from those in \cite{beyer2013} for three reasons: first, we allow the potential parameter $V_0$ to take positive, zero, or negative values (whereas \cite{beyer2013} considered only positive potentials); second, the energy--momentum tensor of the scalar field derived from \eqref{ESF.1} differs by a factor of $2$ due to differing conventions; and third, we work in arbitrary spatial dimension $n>2$, rather than restricting to the case $n=4$.} \cite{beyer2013}, we find it useful to introduce the dimensionless variables $\xtt$ and $\ytt$ via
\begin{equation}
  \label{eq:defxy}
  \xtt:=\sqrt{\frac{2}{(n-1)(n-2)}}\,\frac{\dot\phi}{{\bar\Htt}},\quad 
  \ytt:=\frac{4}{(n-1)(n-2)}\frac{V(\phi)}{{\bar\Htt}^2}. 
\end{equation}
The variables $\xtt$ and $\ytt$ can be interpreted as a scale invariant scalar field kinetic energy and a scale invariant scalar field potential energy, respectively. In terms of these variables, the effective equation of state parameter \eqref{eq:effectiveequationofstate} can be expressed as
\begin{equation}
  \label{eq:effectivew}
  w=\frac{\frac 12\dot\phi^2-V(\phi)}{\frac 12\dot\phi^2+V(\phi)}
  =\frac{\xtt^2 -\ytt}{\xtt^2 +\ytt}.
\end{equation}

It is straightforward to verify that the spatially flat FLRW Einstein--scalar field equations with the potential \eqref{eq:exponentialpotential} can be separated into the following \textit{evolution equations}
\begin{align}
  \label{eq:FLRWevoleq.11}
  \xtt'&=\xtt(q-n+2)+\frac{\sqrt{(n-1)(n-2)}}{2\sqrt2}s \ytt,\\
  \label{eq:FLRWevoleq.12}
  \ytt'&=-\sqrt{\frac{(n-1)(n-2)}{2}}s \xtt \ytt+2(1+q)\ytt,\\ 
  \label{eq:FLRWevoleq.13}
  {\bar\Htt}'&=-(1+q) {\bar\Htt},
\end{align}
and the \textit{Hamiltonian constraint equation}
\begin{equation}
  \label{eq:FLRWconstr}
  \xtt^2+\ytt=1,
\end{equation}
where
\begin{equation}
  \label{eq:FLRWq}
  q=(n-2)\xtt^2-\ytt
\end{equation}
defines the \textit{deceleration scalar}. The momentum constraint is satisfied trivially. In deriving these equations, we have used the fact that the scalar field potential enters only through the ratio $V'(\phi)/V(\phi)$, which, for exponential potentials \eqref{eq:exponentialpotential} with $V_0\not=0$, is equal to the constant $-s$. The system above remains valid in the case $V_0=0$, since in that case $\ytt$ vanishes identically as a consequence of its definition \eqref{eq:defxy}.

We observe that \eqref{eq:FLRWevoleq.13} decouples from the remaining evolution equations, and that the constraint \eqref{eq:FLRWconstr} allows us to determine the variable $\ytt$ algebraically from the variable $\xtt$, thereby eliminating the evolution equation \eqref{eq:FLRWevoleq.12}. The essential dynamics of the spatially flat FLRW Einstein--scalar field system is thus captured by the one-dimensional autonomous equation
\begin{equation}
  \label{eq:FLRWevoleq.21}
  \xtt'=(n-1)\Bigl(\frac{s}{s_c}-\xtt\Bigr)\Bigl(1-\xtt^2\Bigr),
\end{equation}
where the constant
\begin{equation}
  \label{sc-def}
  s_{c} = \sqrt{\frac{8(n-1)}{n-2}}
\end{equation}
plays a central role in the present work.

Once a solution $\xtt$ of \eqref{eq:FLRWevoleq.21} has been determined,  $\ytt$ is determined by
\begin{equation}
  \label{eq:FLRWevoleq.22}
  \ytt=1-\xtt^2,
\end{equation}
and ${\bar\Htt}$ is obtained by integrating \eqref{eq:FLRWevoleq.13}. From \eqref{eq:FLRWq}, we note that the deceleration scalar is given by
\begin{equation}
  \label{eq:FLRWq.2}
  q=(n-1)\xtt^2-1,
\end{equation}
while the effective equation of state parameter satisfies
\begin{equation}
  \label{eq:effectivew.2}
  w=2\xtt^2-1,
\end{equation}
cf.\ \eqref{eq:effectivew}.

\subsubsection{Fixed point solutions}
\label{sec:FLRWfpsol}
The fixed points of the ODE \eqref{eq:FLRWevoleq.21} play a crucial role in determining the behaviour of the FLRW solutions. In general, there are three fixed points:
\begin{equation}
  \label{eq:FLRWFP}
  \xtt_1=\frac{s}{s_c},\quad \xtt_2=1,\quad \xtt_3=-1.
\end{equation}
Each of these fixed points corresponds to a distinct physical regime and each determines the asymptotic behaviour of the FLRW solutions.
From \eqref{eq:FLRWevoleq.22}, it follows that
\begin{equation}
  \label{eq:yfp}
  \ytt_1=1-\frac{s^2}{s_c^2},\quad \ytt_2=0,\quad \ytt_3=0.
\end{equation}
In view of \eqref{eq:defxy} and \eqref{eq:exponentialpotential} together with \eqref{eq:V0normalisation}, the fixed points $\xtt_2$ and $\xtt_3$ correspond to the case $V_0=0$, whereas $\xtt_1$ corresponds to $V_0\not=0$ (unless $s^2=s_c^2$). As we see below, the fixed points $\xtt_2$ and $\xtt_3$ describe the asymptotic behaviour of solutions in the \emph{Kasner regime}\footnote{The class of Kasner solutions to the Einstein--scalar field equations has been studied extensively; see, for example, \cite{rodnianski2018, rodnianski2014}. In general, these solutions are anisotropic, while $\xtt_2$ and $\xtt_3$ correspond to the special isotropic cases.} in which the scalar field potential becomes asymptotically negligible compared to the kinetic energy (so that $w$ approaches $1$ in view of \eqref{eq:effectivew}). In contrast, for solutions converging to $\xtt_1$, the potential energy remains dynamically significant.

As we show below, the fixed point $\xtt_1$ is past stable for \eqref{eq:FLRWevoleq.21} if and only if $s^2>s_c^2$. In this case, $\ytt_1<0$, and hence $V_0<0$. For solutions asymptotic to $\xtt_1$, the effective equation of state parameter $w$ in \eqref{eq:effectivew} approaches the value $(2s^2-s_c^2)/s_c^2$, cf.\ \eqref{eq:effectivew.2}, which exceeds $1$ if $s^2>s_c^2$. Such solutions are sometimes referred to as \emph{ultrastiff} and fall within the \textit{ekpyrotic regime}.

We now derive the physical FLRW metric $\gb_{ij}$ and the scalar field $\phi$ corresponding to the three fixed point solutions  $\xtt_i$ listed in \eqref{eq:FLRWFP}. In the following, we assume that $\xtt_i\not=0$ (that is, we exclude the special case in which $i=1$ and $s=0$). We have already seen that \eqref{eq:yfp} is a consequence of \eqref{eq:FLRWevoleq.22}. 
Exploiting \eqref{eq:defHubbletime}, \eqref{eq:FLRWevoleq.13}, \eqref{eq:FLRWconstr} and \eqref{eq:FLRWq},  we find that 
\begin{equation}
  {\bar\Htt}^{-2}\dot {\bar\Htt}=-(n-1)\xtt_i^2,
\end{equation}
for each $i=1,2,3$. The particular solution for which ${\bar\Htt}$ diverges at $\tb=0$ is given by
\begin{equation}
  \label{eq:FLRWFPsol.01}
  {\bar\Htt}=\frac1{(n-1)\xtt_i^2\tb}.
\end{equation}
From this, we find that the corresponding spatially flat FLRW metrics can by expressed in local coordinates as
\begin{equation}
  \label{eq:FLRWFPsol.11}
  \bar g=-d\tb\otimes d\tb+{\tb}^{2/((n-1)\xtt_i^2)}\sum_{\Lambda=1}^{n-1}dx^\Lambda\otimes dx^\Lambda.
\end{equation}

Next, we consider the definition  \eqref{eq:defxy} of $\xtt$ as a differential equation for $\phi(\tb)$ for any of the three fixed point solutions $\xtt_i$:
\begin{equation*}
  \dot\phi=\sqrt{\frac{(n-1)(n-2)}{2}} \xtt_i {\bar\Htt} =  \frac2{s_c \xtt_i\tb}.
\end{equation*}
Integrating this equation yields the general scalar field solution
\begin{equation}
  \label{eq:FLRWFPsol.21}
  \phi=\frac2{s_c \xtt_i}\log\tb+\phi_*,
\end{equation}
in which $\phi_*$ is an arbitrary constant.

The constraint equation \eqref{eq:FLRWconstr} takes the form
\begin{equation}
  \label{eq:FLRWFPConstr}
  1=\xtt_i^2+\frac 12e^{-s\phi_*}\xtt_i^4s_c^2 V_0 {\tb}^{2(1-\frac{s}{s_c \xtt_i})},
\end{equation}
which is required to hold for all $\tb>0$. Clearly, for $i=2,3$, see \eqref{eq:FLRWFP}, this requires that $V_0=0$. For $i=1$, however, the $\tb$-dependence in \eqref{eq:FLRWFPConstr} drops out and we find, provided $s_c^2\not= s^2$, that
\begin{equation}
  \label{eq:FLRWFPConstr.ekp}
  {e^{s\phi_*}}=-\frac 12V_0 \frac{s^4}{s^2-s_c^2}.
\end{equation}
Hence, given \eqref{eq:V0normalisation}, it follows that $V_0=-1$ if $s^2>s_c^2$ and $V_0=1$ if $s^2<s_c^2$. In the special case $s^2=s_c^2$ and $i=1$, it is straightforward to see that \eqref{eq:FLRWFPConstr} implies $V_0=0$.

In the following, we refer to the particular solutions \eqref{eq:FLRWFPsol.11} and \eqref{eq:FLRWFPsol.21} corresponding to $\xtt_{2,3}=\pm 1$ and $V_0=0$ of the spatially flat FLRW Einstein--scalar field system with the potential \eqref{eq:exponentialpotential} as the \emph{positive Kasner--FLRW} and the \emph{negative Kasner--FLRW} solution, respectively. 
On the other hand, we refer to the particular solution $\xtt_1=s/s_c$ with $V_0=-1$ and $\phi_*$ determined by \eqref{eq:FLRWFPConstr.ekp} as the \emph{positive ekpyrotic--FLRW} or the \emph{negative ekpyrotic--FLRW} solution, according to whether $s$ is positive or negative\footnote{As discussed in Section~\ref{sec:FLRWstabilityanalsys}, the case $s=0$ is always past unstable and is therefore be excluded from consideration.}; see also Table~\ref{tbl:FLRW}. Here, the qualifier \emph{positive/negative} corresponds to whether $\xtt$ is positive/negative. In the positive case, the scalar field \emph{decreases monotonically toward the past}, whereas in the negative case it \emph{increases monotonically toward the past}, cf.\ \eqref{eq:FLRWFPsol.21}. For the remainder of this article, we restrict attention, without loss of generality, to the positive case. Indeed, if $x$ is a solution to \eqref{eq:FLRWevoleq.21} with the parameters $s$ and $V_0$, then $-x$ solves \eqref{eq:FLRWevoleq.21}  with parameters $-s$ and the same $V_0$.

We conclude this subsection by highlighting several interesting properties.
First, it follows from \eqref{eq:effectivew.2} that the effective equation of state parameter $w$ satisfies $w=1$ in the Kasner setting, whereas in the ekpyrotic setting, it is given by
\begin{equation}
  \label{eq:FLRWepyrotic.w}
  w=\frac{2s^2-s_c^2}{s_c^2}>1.
\end{equation}

Second, we emphasise that the \emph{Kasner exponents} (that is, the eigenvalues of the scale-invariant physical Weingarten map) are $p_I=1/(n-1)$ for all $I=1,\ldots,n-1$, as a consequence of isotropy, and this holds in \emph{both} the Kasner and the ekpyrotic regimes. In contrast, the deceleration scalar $q$ differs in the two regimes: in the Kasner regime one has $q=n-2$, while in the ekpyrotic regime,
\begin{equation*}
  q=(n-1)\frac{s^2}{s_c^2}-1
  =n-2+(n-1)\frac{w-1}{2}>n-2,
\end{equation*}
for $w>1$, in view of \eqref{eq:FLRWq.2} and \eqref{eq:FLRWepyrotic.w}.

Third, we consider the scalar field quantities $\phi_0$ and $\phi_1$, originally introduced in \cite[Definition~7 and Example~8]{groeniger2023} and adapted here to our more general setting:
\begin{equation}
  \label{eq:Ringstromlimits}
  \phi_1=\frac{\dot\phi}{(n-1){\bar\Htt}},\quad 
  \phi_0=\phi+\xtt_i^{-2}\phi_1\log((n-1){\bar\Htt}).
\end{equation}
In the Kasner setting where $\xtt_i=\pm 1$ and our definition coincides with that of \cite{groeniger2023}, it was shown in \cite{franco-grisales2024, groeniger2023, ringström2022, ringström2022a, ringström2025} that these quantities converge as $t\searrow 0$ to big bang \emph{asymptotic data} that characterises the leading order behaviour of the scalar field near the big bang.

In the present setting, it follows from \eqref{eq:defxy} and \eqref{sc-def} that
\begin{equation*}
  \phi_1=\frac{2}{s_c}\xtt_i,
\end{equation*}
which yields $\phi_1=\pm \frac{2}{s_c}$ at the Kasner-FLRW fixed points $\xtt_2$ and $\xtt_3$ (recall \eqref{sc-def}), and
\begin{equation}
  \label{eq:FLRWPhiOekpy}
  \phi_1=\frac{2 s}{s_c^2}
\end{equation}
at the ekpyrotic-FLRW fixed point $\xtt_1$. Using \eqref{eq:FLRWFPsol.01} and \eqref{eq:FLRWFPsol.21}, we further obtain
\begin{equation}
  \label{eq:FLRWphi0}
  \phi_0=\phi_*-\frac{2}{s_c\xtt_i}\log(\xtt_i^2).
\end{equation}

In summary, both $\phi_0$ and $\phi_1$ converge to finite limits as $t\searrow 0$ in both the Kasner and the ekpyrotic regimes, and they may therefore be interpreted as asymptotic invariants that characterise the leading order behaviour of the scalar field $\phi$ near the big bang.

\subsubsection{Stability analysis for the class of FLRW Einstein-scalar field solutions}
\label{sec:FLRWstabilityanalsys}
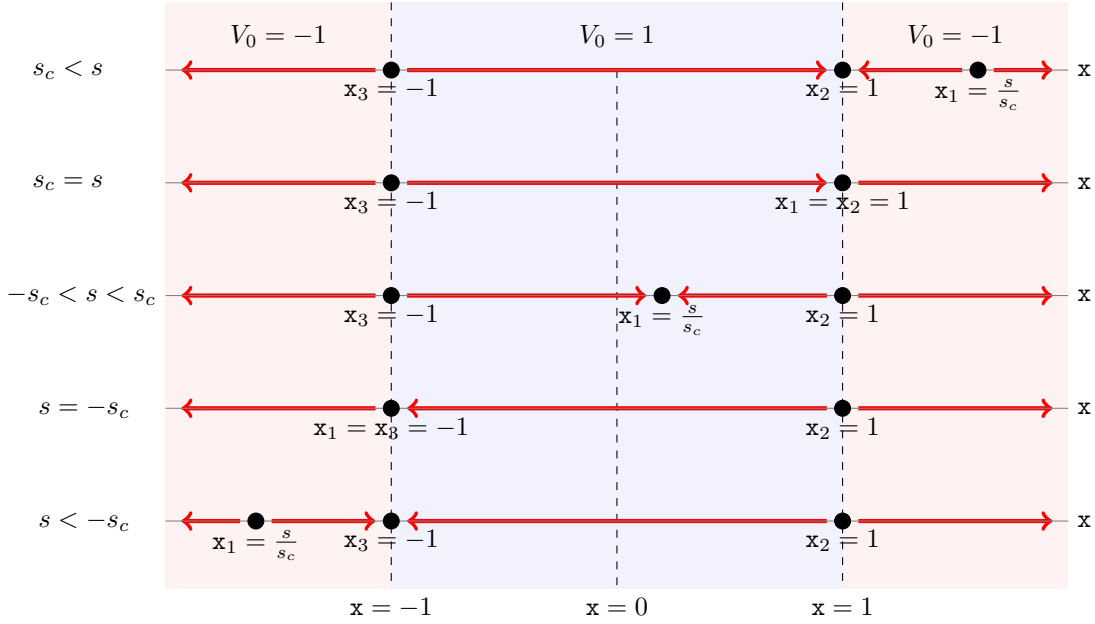
\begin{figure}[t]
  \centering
  \begin{tikzpicture}[scale=3]
    \def\h{1.3}
    \def\th{1.15}
    \draw [fill=red!05,draw=none] (-2,-\h) rectangle (-1,\h);
    \node[draw=none] at (-1.5,\th) {$V_0=-1$};
    \draw [fill=red!05,draw=none] (1,-\h) rectangle (2,\h);
    \node[draw=none] at (1.5,\th) {$V_0=-1$};
    \draw [fill=blue!05,draw=none] (-1,-\h) rectangle (1,\h);
    \node[draw=none] at (0,\th) {$V_0=1$};
    \def\ha{1.0}
    \def\hab{0.5}
    \def\hb{0.0}
    \def\hbc{-0.5}
    \def\hc{-1.0}
    \draw [dashed] (-1,\h) --(-1,-\h) node[below] {$\xtt=-1$};
    \draw [dashed] (0,\ha) --(0,-\h) node[below] {$\xtt=0$};
    \draw [dashed] (1,\h) --(1,-\h) node[below] {$\xtt=1$};

    \def\fpA{1.6}
    \def\fpB{0.2}
    \def\fpC{-1.6}
    \def\gap{0.07}
    \draw[<-] [red] [ultra thick] (-2+\gap,\ha) --(-1-\gap,\ha);
    \draw[->] [red] [ultra thick] (-1+\gap,\ha) --(1-\gap,\ha);
    \draw[<-] [red] [ultra thick] (1+\gap,\ha) --(\fpA-\gap,\ha);
    \draw[->] [red] [ultra thick] (\fpA+\gap,\ha) --(2-\gap,\ha);
    \draw[<-] [red] [ultra thick] (-2+\gap,\hab) --(-1-\gap,\hab);
    \draw[->] [red] [ultra thick] (-1+\gap,\hab) --(1-\gap,\hab);
    %\draw[<-] [red] [ultra thick] (1+\gap,\hab) --(\fpA-\gap,\hab);
    \draw[->] [red] [ultra thick] (1+\gap,\hab) --(2-\gap,\hab);
    \draw[<-] [red] [ultra thick] (-2+\gap,\hb) --(-1-\gap,\hb);
    \draw[->] [red] [ultra thick] (-1+\gap,\hb) --(\fpB-\gap,\hb);
    \draw[<-] [red] [ultra thick] (\fpB+\gap,\hb) --(1-\gap,\hb);
    \draw[->] [red] [ultra thick] (1+\gap,\hb) --(2-\gap,\hb);
    \draw[<-] [red] [ultra thick] (-2+\gap,\hbc) --(-1-\gap,\hbc);
    %\draw[->] [red] [ultra thick] (\fpC+\gap,\hbc) --(-1-\gap,\hbc);
    \draw[<-] [red] [ultra thick] (-1+\gap,\hbc) --(1-\gap,\hbc);
    \draw[->] [red] [ultra thick] (1+\gap,\hbc) --(2-\gap,\hbc);
    \draw[<-] [red] [ultra thick] (-2+\gap,\hc) --(\fpC-\gap,\hc);
    \draw[->] [red] [ultra thick] (\fpC+\gap,\hc) --(-1-\gap,\hc);
    \draw[<-] [red] [ultra thick] (-1+\gap,\hc) --(1-\gap,\hc);
    \draw[->] [red] [ultra thick] (1+\gap,\hc) --(2-\gap,\hc);
    \draw [gray] (-2,\ha) node [left] [black] {$s_c<s\qquad$}--(2,\ha) node [right] [black] {$\xtt$};
    \draw [gray] (-2,\hab) node [left] [black] {$s_c=s\qquad$}--(2,\hab) node [right] [black] {$\xtt$};
    \draw [gray] (-2,\hb) node [left] [black] {$-s_c<s<s_c$}--(2,\hb) node [right] [black] {$\xtt$};
    \draw [gray] (-2,\hbc) node [left] [black] {$s=-s_c\quad$}--(2,\hbc) node [right] [black] {$\xtt$};
    \draw [gray] (-2,\hc) node [left] [black] {$s<-s_c\quad$}--(2,\hc) node [right] [black] {$\xtt$};
    \def\fpth{1.0}
    \filldraw [black] (\fpA,\ha) circle (\fpth pt) node[below] {$\xtt_1=\frac{s}{s_c}$};
    \filldraw [black] (1,\ha) circle (\fpth pt) node[below] {$\xtt_2=1$};
    \filldraw [black] (-1,\ha) circle (\fpth pt) node[below] {$\xtt_3=-1$};
    \filldraw [black] (1,\hab) circle (\fpth pt) node[below] {$\xtt_1=\xtt_2=1$};
    \filldraw [black] (-1,\hab) circle (\fpth pt) node[below] {$\xtt_3=-1$};
    \filldraw [black] (\fpB,\hb) circle (\fpth pt) node[below] {$\xtt_1=\frac{s}{s_c}$};
    \filldraw [black] (1,\hb) circle (\fpth pt) node[below] {$\xtt_2=1$};
    \filldraw [black] (-1,\hb) circle (\fpth pt) node[below] {$\xtt_3=-1$};
    \filldraw [black] (1,\hbc) circle (\fpth pt) node[below] {$\xtt_2=1$};
    \filldraw [black] (-1,\hbc) circle (\fpth pt) node[below] {$\xtt_1=\xtt_3=-1$};
    \filldraw [black] (\fpC,\hc) circle (\fpth pt) node[below] {$\xtt_1=\frac{s}{s_c}$};
    \filldraw [black] (1,\hc) circle (\fpth pt) node[below] {$\xtt_2=1$};
    \filldraw [black] (-1,\hc) circle (\fpth pt) node[below] {$\xtt_3=-1$};
\end{tikzpicture}  
  \caption{The spatially flat FLRW solution state space of the Einstein-scalar field equation with the potential \eqref{eq:exponentialpotential} (assuming the normalisation \eqref{eq:V0normalisation}) represented by the variable $\xtt$ defined in \eqref{eq:defxy} with the fixed points $\xtt_1$, $\xtt_2$ and $\xtt_3$ given by \eqref{eq:FLRWFP} where the constant $s_c$ is given by \eqref{sc-def}. The arrows in this diagram indicate the direction of the \emph{future} flow (that is, the arrow points to the right (left, respectively) if the right-hand side of \eqref{eq:FLRWevoleq.21} is positive (negative, respectively)).}
  \label{fig:FlatFLRWStatespace}
\end{figure} 
To prepare for our general stability analysis,
we provide a brief past stability analysis of the one-dimensional dynamical system \eqref{eq:FLRWevoleq.21}. The elementary conclusions are summarised in Figure~\ref{fig:FlatFLRWStatespace} and in Table~\ref{tbl:FLRW}. 

Figure~\ref{fig:FlatFLRWStatespace} depicts the one-dimensional state space of the FLRW solutions of the Einstein-scalar field equations for the  different ranges of the scalar field potential parameters $s$ and $V_0$, and the corresponding positions of the three fixed points $\xtt_1$, $\xtt_2$ and $\xtt_3$. 
We recall from \eqref{eq:defxy} and \eqref{eq:exponentialpotential} that the sign of $\ytt$ must agree with the sign of $V_0$. 
If $V_0\ge 0$ the constraint \eqref{eq:FLRWconstr} implies the bound $x^2\le 1$. Consequently, the state space of \eqref{eq:FLRWevoleq.21} is compact and coincides with the blue shaded part of Figure~\ref{fig:FlatFLRWStatespace}. On the other hand, the constraint \eqref{eq:FLRWconstr} implies that $x^2>1$ if $V_0<0$, which is the red shaded part of Figure~\ref{fig:FlatFLRWStatespace}. It is evident from Figure~\ref{fig:FlatFLRWStatespace} that the value of $s$ and the sign of $V_0$ together determine which of the fixed points $\xtt_1$, $\xtt_2$ or $\xtt_3$ are past attractors for the class of FLRW solutions; cf.\ also Table~\ref{tbl:FLRW}.
\begin{table}[h]
  \centering
  \begin{tabular}{|l|c|c|c|c|c|}\hline
    Regime & past stable attractor & $V_0$ & $s$ & limit of $w$ & limit of $q$\\\hline\hline
    positive ekpyrotic & $\xtt_1$ & $-1$ & $(s_c,\infty)$ & $(2s^2-s_c^2)/s_c^2\in (1,\infty)$ & $(n-1)\frac{s^2}{s_c^2}-1$\\
    positive critical & $\xtt_1=\xtt_2$ & $-1,0$ & $s_c$ & $1$ & $n-2$ \\
    positive Kasner & $\xtt_2$ & $-1,0,1$ & $(-\infty,s_c)$ & $1$ &$n-2$\\
    negative Kasner & $\xtt_3$ & $-1,0,1$ & $(-s_c,\infty)$ & $1$ & $n-2$\\
    negative critical & $\xtt_1=\xtt_3$ & $-1,0$ & $-s_c$ & $1$ &  $n-2$\\
    negative ekpyrotic & $\xtt_1$ & $-1$ & $(-\infty,-s_c)$ &  $(2s^2-s_c^2)/s_c^2\in (1,\infty)$& $(n-1)\frac{s^2}{s_c^2}-1$\\\hline
  \end{tabular}
  \caption{Past dynamics for the class of spatially flat FLRW solutions of the Einstein-scalar field equations with exponential potential assuming the normalisation \eqref{eq:V0normalisation}.}
  \label{tbl:FLRW}
\end{table}

For the remainder of this article, we restrict our attention to analysing small perturbations of the  positive ekpyrotic-FLRW solution. No generality is loss in doing so because perturbations of the negative ekpyrotic-FLRW solution can be obtained from perturbations of the positive  ekpyrotic-FLRW solution via the transformation $(s,\phi)\mapsto (-s,-\phi)$ where the metric $\gb$ and the potential parameter $V_0$ remain fixed. 

The past nonlinear stability of the positive and the negative Kasner-FLRW solutions and their big bang singularities follow directly from the stability results established in \cite{groeniger2023}. Finally, it is worth mentioning that the past stability of the positive and the negative critical-FLRW solutions, cf.\ Table~\ref{tbl:FLRW}, remains unresolved.

\subsection{An informal statement of the main results}
%\label{sec:informal}
The main result of this article, Theorem~\ref{glob-stab-thm}, establishes nonlinear stability in the contracting time direction for perturbations of the positive ekpyrotic-FLRW solution \eqref{eq:FLRWFPsol.11}--\eqref{eq:FLRWFPsol.21} with $\xtt_i=s/s_c$ and $s>s_c$ to the Einstein--scalar field equations \eqref{ESF.1}--\eqref{ESF.2} with the potential \eqref{eq:exponentialpotential} and $V_0=-1$, in spacetime dimensions $n\geq 3$. The proof of Theorem~\ref{glob-stab-thm} involves solving the Cauchy problem arising from \emph{synchronised} initial data; that is, initial data prescribed on the hypersurface 
\begin{equation*}
    \Sigma_{t_0}=\{t_0\}\times \Tbb^{n-1},
\end{equation*} for some $t_0>0$, where the scalar field is initially constant along $\Sigma_{t_0}$. It follows from the results of Section~\ref{temp-synch} that this restriction entails no loss of generality. The synchronisation condition permits the introduction of a gauge that preserves this property and allows us to use the scalar field $\phi$ as a time function
\[
\tau=e^{\frac{(n-1)s^2-s_c^2}{2(n-1)s}\phi}.
\]
The purpose of this synchronisation is to ensure that the big bang singularity occurs simultaneously at $\tau=0$.

We now present an informal version of our main result, which is Theorem~\ref{glob-stab-thm} below.

\begin{thm}[Past global stability of the positive ekpyrotic-FLRW spacetime] \label{thm:informal}
The positive ekpyrotic-FLRW solution of the Einstein-scalar field equations is nonlinearly stable in the contracting direction in the sense that solutions $\{\gb_{ij},\phi\}$ of the Einstein-scalar field equations that are 
evolved from sufficiently differentiable, synchronized initial data imposed on $\Sigma_{t_0}$ that is suitably close to initial data of the positive ekpyrotic-FLRW solution exist on the spacetime region $M \cong \bigcup_{t\in (0,t_0]}\tau^{-1}(\{t\})\cong (0,t_0]\times \Tbb^{n-1}$. The perturbed solutions are $C^2$-inextendible through the $\tau\!=\!0$-boundary of $M$, past timelike geodesically incomplete and they terminate at quiescent, crushing AVTD big bang singularities, which are
located at $\tau\!=\!0$ and are characterised by curvature blow up.
 
\smallskip
 
\noindent Moreover, using\footnote{The induced spatial metric $\bar{\gtt}_{ij}$, acceleration $\pacc_i$, and extrinsic curvature $\bar{\Ktt}_{ij}$ associated to the $\tau\!=\!\text{constant}$ foliation are defined in Appendix \ref{spacetimedecomp}; see in particular \eqref{eq:trafounitnormal.a}, \eqref{eq:pspatmetrdef} and \eqref{defphysextr}.} $\bar{\gtt}_{ij}$, $\nb_i$, $\pacc_i$, and $\bar{\Ktt}_{ij}$ to denote the induced spatial metric, normal vector, acceleration and extrinsic curvature associated to the $\tau\!=\!\text{constant}$ foliation, and using  $\bar{\Ktt}_{ij}=\sigmab_{ij}+\bar{\Htt}\bar{\gtt}_{ij}$ to denote the decomposition of the extrinsic curvature into the shear $\sigmab_{ij}$ and the Hubble scalar (mean curvature) $\bar{\Htt}$, these solutions
\begin{enumerate}[(a)]
\item isotropise in the sense that 
\begin{equation*}
\lim_{\tau\searrow 0} \frac{\sigmab_{ij}\sigmab^{ij}}{\bar{\Htt}^2}=0, \quad \lim_{\tau\searrow 0}\frac{\pacc^i\pacc_i}{\bar{\Htt}^2} = 0, \quad
\lim_{\tau\searrow 0} \frac{\bar{\Upsilon}_i\bar{\Upsilon}^i}{\bar{\Htt}^4} =0
\AND \lim_{\tau\searrow 0} \frac{\bar{\Xi}_i{}^j\bar{\Xi}^i{}_j}{\bar{\Htt}^4} =0
\end{equation*}
where $\bar{\Upsilon}_i=\Rb_j{}^k\bar{\gtt}_i{}^j\nb_k$, $\bar{\Xi}_i{}^j=\Rb_k{}^l\bar{\gtt}_i{}^k\bar{\gtt}_k{}^j-\frac{1}{n-1}\Rb_k{}^l\bar{\gtt}_k{}^l\bar{\gtt}_i{}^j$, and
$\Rb_{ij}$ is the Ricci curvature tensor of $\gb_{ij}$, 
\item and their spacetime curvature is Ricci dominated in the sense that
\begin{equation*}
\lim_{\tau\searrow 0} \frac{\Wb_{ijkl}\Wb^{ijkl}}{\Rb_{ij}\Rb^{ij}} =0 \AND 
\lim_{\tau\searrow 0}\Rb_{ij}\Rb^{ij}=\infty,
\end{equation*}
where $\Wb_{ijkl}$ is the Weyl curvature tensor of $\gb_{ij}$. 
\end{enumerate}
\end{thm}

Motivated by the concept of \textit{isotropic singularities}  \cite{anguige1999,tod1999,goode1985,tod1999a,tod2002}, we find that it is natural to ask whether the \textit{conformal} geometry defined by
\begin{equation*}
  \gb_{ij} = e^{2\Phi}g_{ij}, \quad 
  e^{2\Phi}=e^{2\Phi_0}\tau^{\frac{2(1-\Rtt)}{n-2}},\quad 
  \Phi_0\in\Rbb,
\end{equation*}
which is defined naturally in our framework, admits a meaningful extension through the big bang singularity at $\tau\!=\!0$. A spacetime is defined to possess an isotropic singularity at $\tau\!=\!0$ in the sense of \cite[Def.~1]{goode1985} if the physical geometry is $C^2$-inextendible across the hypersurface $\tau\!=\!0$, whereas the conformal geometry extends through the singularity with sufficiently high regularity. Typically, the regularity of this extension is at least $C^2$, which guarantees that the conformal Weyl tensor extends continuously across the singular hypersurface. In such a situation, the singular behaviour is encoded in the conformal factor. In particular, Ricci dominance, that is,
\begin{equation} \label{Ricci-dominance}
\lim_{\tau\searrow 0} \frac{\Wb_{ijkl}\Wb^{ijkl}}{\Rb_{ij}\Rb^{ij}} =0 \AND
\lim_{\tau\searrow 0}\Rb_{ij}\Rb^{ij}=\infty,
\end{equation}
follows directly from the singular conformal structure \cite[Thm.~3.1]{goode1985}. Furthermore, if the timelike congruence is taken to be the unit normal\footnote{In the terminology of \cite{goode1985}, this ensures that the timelike congurence is regular and orthogonal at $\tau=0$. Also in our setting this choice coincides with taking the timelike congruence to be the one determined by the fluid velocity $u_i$; see \eqref{utt-def}.} $\nb_i$ of the $\tau\!=\!\text{constant}$ foliation, then such solutions isotoprise is the kinematic sense 
\cite[Thm.~3.3]{goode1985}   
\begin{equation}\label{kinematic-isotropisation}
\lim_{\tau\searrow 0} \frac{\sigmab_{ij}\sigmab^{ij}}{\bar{\Htt}^2}=0, \quad \lim_{\tau\searrow 0}\frac{\pacc^i\pacc_i}{\bar{\Htt}^2} = 0, 
\end{equation}
and the Ricci sense \cite[Thm.~3.2]{goode1985}
\begin{equation} \label{Ricci-isotropisation}
\lim_{\tau\searrow 0} \frac{\bar{\Upsilon}_i\bar{\Upsilon}^i}{\bar{\Htt}^4} =0,\quad \lim_{\tau\searrow 0} \frac{\bar{\Xi}_i{}^j\bar{\Xi}^i{}_j}{\bar{\Htt}^4} =0.
\end{equation}

While Theorem~\ref{glob-stab-thm} ensures that the conformal geometry of the perturbed positive ekyprotic-FLRW solution extends at least continuously to a neighbourhood of the big bang singularity at $\tau\!=\!0$, this degree of regularity is insufficient to guarantee that the conformal Weyl tensor $W_{ijkl}$ extends continuously to $\tau\!=\!0$. Consequently, Theorem~\ref{glob-stab-thm} does not yield a definitive answer concerning the behaviour of the conformal Weyl tensor at the singularity. To clarify this point, we use \eqref{eq:Weylconfinv} together with $\eb_{\ib}^\mu=e^{-\Phi} e_i^\mu \delta_{\ib}^i$ to obtain
\begin{equation}
  \label{eq:Weylconf.rel}
  W_{ijkl}
  =e^{2\Phi} \Wb_{\ib\jb{\bar k}{\bar l}}\delta_i^{\ib}\delta_j^{\jb}\delta_k^{\bar k}\delta_l^{\bar l}
  =\tau^{-2}\left(\tau e^{\Phi}\bar\Htt\right)^2
  \frac{\Wb_{\ib\jb{\bar k}{\bar l}}}{\bar\Htt^2}
  \delta_i^{\ib}\delta_j^{\jb}\delta_k^{\bar k}\delta_l^{\bar l}.
\end{equation}
Although Theorem~\ref{glob-stab-thm} implies that $\tau e^{\Phi}\bar\Htt$ remains bounded  and that $\Wb_{\ib\jb{\bar k}{\bar l}}/\bar\Htt^2$ decays at a certain rate, our estimates are not sufficiently sharp to conclude that every component decays rapidly enough to compensate for the $\tau^{-2}$ singular factor in \eqref{eq:Weylconf.rel}. 
Accordingly, the resulting big bang singularities in the perturbed solutions need not be isotropic in the strong sense of \cite{goode1985}. Nevertheless, since these singularities satisfy the characteristic features of Ricci dominance \eqref{Ricci-dominance} and isotropisation \eqref{kinematic-isotropisation}-\eqref{Ricci-isotropisation}, they may reasonably be regarded as isotropic.

Our results also show that the scalar field quantities $\phi_0$ and $\phi_1$, cf.\ Theorem~\ref{glob-stab-thm}.(h), converge as $\tau\searrow 0$. Closely related quantities have been analysed in detail in \cite{franco-grisales2024, groeniger2023, ringström2022, ringström2022a, ringström2025,franco-grisales2026}, where they were shown to converge and to determine asymptotic data for the scalar field in the Kasner regime. Theorem~\ref{glob-stab-thm} establishes an analogous convergence result in the ekpyrotic regime, in the sense that these quantities also admit limits as $t\searrow 0$. However, in contrast to the Kasner regime, where the limiting values are generically spatially dependent and encode asymptotic data, the limits of $\phi_0$ and $\phi_1$ in the ekpyrotic regime are necessarily \textit{spatially constant}. It is therefore unclear whether they capture the full asymptotic behaviour of the scalar field in this setting.

It is also worth mentioning that the effective equation-of-state quantity $w$, defined by \eqref{eq:effectiveequationofstate}, is well-defined for the perturbed positive ekyprotic-FLRW solutions  and converges to the positive ekpyrotic-FLRW  constant value \eqref{eq:FLRWepyrotic.w} as $\tau \searrow 0$. In particular, the asymptotic value of $w$ for these solutions is larger than one, which suggests that the perturbed solutions behave asymptotically like irrotational ultra-stiff fluids. Indeed, the authors of \cite{heinzle2012} found similar behaviour in solutions with big bang singularities for a large class of ultra-stiff Einstein-fluid systems. Interestingly, while our solutions are also asymptotically velocity term dominated (AVTD), the scalar field is not dominated by its kinetic energy. 

The proof makes essential use of the variable and gauge choices for the Einstein-scalar field equations with vanishing potential introduced in \cite{BeyerOliynyk:2021}, and  employed again in \cite{beyerStabilityFLRWSolutions2023}. In \cite{BeyerOliynyk:2021}, two of the authors developed a conformal representation of the geometry in which the scalar field plays a central role, together with a novel wave gauge approach where the target manifold carries the flat Minkowski metric which takes the canonical form with respect to Lagrangian coordinates adapted to the scalar field gradient. As a consequence of these choices, the reduced conformal Einstein–scalar field system is symmetric hyperbolic and Fuchsian in the sense of \cite{oliynyk2016,BeyerOliynyk:2020, BOOS:2021}. Moreover, for all solutions whose initial data are sufficiently close to FLRW data and are synchronised in the sense that the scalar field is constant on the initial hypersurface, the function $\tau=e^{-\alpha\phi}$ (see \eqref{taudef}) can be used as a time coordinate where the big bang singularity occurs at $\tau=0$. The formulation in \cite{BeyerOliynyk:2021} also uses Fermi-Walker propagated frames that are orthonormal with respect to the conformal metric.

Most of this framework carries over to the non-vanishing potential case studied here. In particular, the local existence, uniqueness, and extension result of Section~\ref{sec:locexist_cont} is essentially the same as in \cite{BeyerOliynyk:2021}.
Beyond this, the main differences required to capture the ekpyrotic dynamics (see Section~\ref{sec:confEinstSF}) are as follows: the conformal gauge must be adjusted, the reduction to Fuchsian form, while following the same general strategy, requires a modification of the orthonormal-frame variables (see Section~\ref{sec:tnormalisedequations}), and the potential terms need more refined nonlinear decompositions (see Section~\ref{sec:NonlinDecomp}).

These differences lead to a more complicated leading-order coefficient matrix in the Fuchsian system, for which we are unable to verify that its symmetrised version is positive semi-definite, a property required to apply Fuchsian techniques. We resolve this difficulty using the method of successive spatial differentiation introduced in \cite{BeyerOliynykZheng:2025}, as explained in Section~\ref{sec:differentiated_system}. Once global existence toward the big bang and the basic energy, decay, and convergence estimates are obtained from the Fuchsian theory of \cite{oliynyk2016,BeyerOliynyk:2020, BOOS:2021}; see Proposition~\ref{prop:globalstability}, the remainder of the proof of Theorem~\ref{glob-stab-thm} proceeds along more standard lines.

As a final remark, we note that the stability results established in this article can be \emph{localised in space}. This can be achieved by following the same strategy as in \cite[Theorem~11.1]{BeyerOliynyk:2021}, which is a spatially localised version of the global-in-space stability theorem \cite[Theorem~10.1]{BeyerOliynyk:2021}.

\subsection{Overview}
This paper is organised as follows: In Section~\ref{sec:FLRWEinsteinSF}, we analyse the family of spatially flat FLRW solutions to the Einstein scalar field equations and we define the Kasner-FLRW and the ekpyrotic-FLRW solutions. In particular, we analyse the stability of these solutions within the class of spatially homogeneous and  isotropic solutions in Section~\ref{sec:FLRWstabilityanalsys}, which provides intuition for the general setting. Our conventions, notation, and relevant function spaces are collected together in Section~\ref{prelim}. In Section~\ref{presec:confEinstSF}, we develop a conformal reformulation of the Einstein-scalar field equations as well as a wave gauge reduced version, while the conformal representation of the FLRW-ekpyrotic solution is presented in Section~\ref{sec:confrep_ekpyrotic}. Subsequently, local existence and continuation,  constraint propagation and synchronisation of the initial data for the conformal Einstein-scalar field equations is addressed in Section~\ref{sec:locexist_cont}. In Section~\ref{Fermi}, we introduce a Fermi-Walker transported orthonormal frame. To prepare for the Fuchsian analysis, we employ the orthornormal frame and use it to express the gauge-reduced conformal Einstein-scalar field equations in first-order form in Section~\ref{sec:FirstOrderForm} and we derive estimates for the nonlinear terms in Section~\ref{sec:NonlinDecomp}. 
The first-order evolution equations evolution equations are transformed in Section~\ref{sec:FuchsianForm} into a Fuchsian form, which plays an important role in our derivation of uniform bounds near the big bang singularity (See Subsections~\ref{sec:Fuch-form} and \ref{sec:differentiated_system}). Finally, in Section~\ref{sec:main_theorem}, we state and prove our main result, Theorem~\ref{glob-stab-thm}.

\section{Preliminaries\label{prelim}}

%\subsection{Data availability statement}

%This article has no associated data.

\subsection{Coordinates, frames and indexing conventions\label{indexing}}
In the article, we consider $n$-dimensional spacetime manifolds of the form
\begin{equation} \label{Mt1t0-def}
    M_{t_1,t_0}= (t_1,t_0]\times \Tbb^{n-1},
\end{equation}
where $t_0>0$, $t_1<t_0$, and $\Tbb^{n-1}$ is
the $(n-1)$-torus defined by
\begin{equation} \label{Tbb-def}
    \Tbb^{n-1} = [-L,L]^{n-1}/\sim
\end{equation} 
with $\sim$ denoting the equivalence relation obtained
from identifying the sides of the box $[-L,L]^{n-1}\subset \Rbb^{n-1}$. On $M_{t_1,t_0}$,
we always employ coordinates $(x^\mu)=(x^0,x^\Lambda)$
where $(x^\Lambda)$ are periodic spatial coordinates
on $\Tbb^{n-1}$ and $x^0$ is a time coordinate
on the interval $(t_1,t_0]$. Lower case Greek letters, e.g., $\mu,\nu,\gamma$, run from $0$ to $n-1$ and are used to label spacetime coordinate indices while upper case Greek letters, e.g., $\Lambda,\Omega,\Gamma$, run from $1$ to $n-1$ and label spatial coordinate indices. Partial derivative with respect to the coordinates $(x^\mu)$ are denoted by $\del{\mu} = \frac{\del{}\;}{\del{}x^\mu}$.
We often denote the time coordinate $x^0$ by $t$, that is, $t = x^0$, and write $\del{t} = \del{0}$ for the partial derivative with respect to $x^0$. The time function $t$ induces a foliation of the spacetime manifold $M_{t_1,t_0}$ into spatial slices $\Sigma_{t}$ as follows:
\begin{equation*}
M_{t_1,t_0}= \bigcup_{t_1 < t \leq t_0} \Sigma_{t}
\quad \text{where} \quad
\Sigma_{t}=\{t\}\times \Tbb^{n-1}.
\end{equation*}

We use two types of frames in this article, each serving a distinct purpose. The first frame, $e_j = e_j^\mu \del{\mu}$, is orthonormal with respect to the conformal metric $g_{\mu\nu}$ and is employed for most computations. The second frame, $\eb_{\jb} = \eb_{\jb}^\mu \del{\mu}$, is orthonormal with respect to the physical metric $\gb_{\mu\nu}$ and only appears in the statement of our main results.  To distinguish indices, we adopt the following conventions: lowercase Latin letters, e.g.,\ $i,j,k$ and $\bar i,\jb,\bar k$, label frame indices running from $0$ to $n-1$, while uppercase Latin letters, e.g.,\ $I,J,K$ and $\bar I,\bar J,\bar K$, label spatial frame indices running from $1$ to $n-1$.

\subsection{Inner-products and matrices}
We use $\ipe{\xi}{\zeta} = \xi^{\tr} \zeta$ and $|\xi| = \sqrt{\ipe{\xi}{\xi}}$,  where $\xi,\zeta \in \Rbb^N$, to denote the Euclidean inner-product and norm, respectively. The set of all $N \times N$ matrices is written as $\Mbb{N}$, and $\Sbb{N}$ denotes the subspace of symmetric $N \times N$ matrices. For $A \in \Mbb{N}$, the operator norm $|A|_{\op}$ is defined by
\begin{equation*}
   |A|_{\op} = \sup_{\xi \in \Rbb^N_\times} \frac{|A\xi|}{|\xi|},
\end{equation*}
where $\Rbb^N_\times = \Rbb^N \setminus \{0\}$. For $A,B \in \Mbb{N}$, we also use the notation
\begin{equation*}
    A \leq B \quad \Longleftrightarrow \quad \xi^{\tr} A \xi \leq \xi^{\tr} B \xi, \quad \forall \; \xi \in \Rbb^N.
\end{equation*}

\subsection{Sobolev spaces}
For $k \in \Zbb_{\geq 0}$, the $W^{k,p}$ norm of a map $u \in C^\infty(U,\Rbb^N)$, where $U \subset \Tbb^{n-1}$ is open, is defined by
\begin{equation*}
\norm{u}_{W^{k,p}(U)} = \begin{cases} 
\begin{displaystyle}
\biggl( \sum_{|\bc|\leq k} \int_U |\del{}^{\bc} u|^p \, d^{n-1} x \biggl)^{\frac{1}{p}}  
\end{displaystyle} & \text{if $1 \leq p < \infty$}, \\[1em]
\begin{displaystyle} 
\max_{|\bc| \leq k} \sup_{x \in U} |\del{}^{\bc} u(x)|  
\end{displaystyle} & \text{if $p = \infty$},
\end{cases}
\end{equation*}
where $\bc = (\bc_1,\ldots,\bc_{n-1}) \in \Nbb_{0}^{n-1}$ is a multi-index, $|\bc| = \sum_{\ell=1}^{n-1} \bc_\ell$, and
$\del{}^\bc = \del{1}^{\bc_1} \del{2}^{\bc_2} \cdots \del{n-1}^{\bc_{n-1}}$. The Sobolev space $W^{k,p}(U,\Rbb^N)$ is defined as the completion of $C^\infty(U,\Rbb^N)$ with respect to this norm. If $N = 1$ or the dimension $N$ is clear from the context, we simplify notation and write $W^{k,p}(U)$ instead of $W^{k,p}(U,\Rbb^N)$. We also adopt the standard notation $H^k(U,\Rbb^N) = W^{k,2}(U,\Rbb^N)$ throughout this article.

\subsection{Constants and inequalities}
We use the notation $a \lesssim b$ to represent inequalities of the form $a \leq Cb$ in situations where the precise value or dependence of the constant $C$ is not important. If the dependence of $C$ on other quantities must be specified, for example if $C$ depends on the norm $\norm{u}_{L^\infty}$, we write $C = C(\norm{u}_{L^\infty})$. All such constants are assumed to be non-negative, non-decreasing, and continuous functions of their arguments.

\subsection{Curvature conventions\label{sec:curv_conv}}
The \textit{Riemann (curvature) tensor} $\Rc_{ijk}{}^l$ associated with a metric $\gc_{ij}$ is defined by
\begin{equation*}%\label{comm-Rc}
  [\Dc_i,\Dc_j]\omega_k = \Rc_{ijk}{}^l \omega_l
\end{equation*}
where $\omega_l$ is any $1$-form and $\Dc_i$ denotes the Levi-Civita connection of $\gc_{ij}$.  The \textit{Ricci tensor} is then defined by
\begin{equation*}
  %\label{Ricciconv}
  \Rc_{ik} = \Rc_{ijk}{}^j.
\end{equation*}
These definitions fix the curvature conventions employed in this article.

\section{Conformal field equations}
\label{presec:confEinstSF}

\subsection{Conformal Einstein-scalar field equations}
\label{sec:confEinstSF}
In this section, we reformulate the Einstein-scalar field equations, closely following the approach in \cite{BeyerOliynyk:2021}. The reformulation begins by replacing the physical metric $\gb_{ij}$ with a \textit{conformal metric} $g_{ij}$ defined by
\begin{equation}\label{confmet}
  \gb_{ij} = e^{2\Phi}g_{ij},
\end{equation}
where, here and throughout this article, we assume that the spacetime dimension $n$ satisfies $n\geq 3$, and $\Phi$ is here an arbitrary sufficiently smooth function.

Under the conformal transformation \eqref{confmet}, the Ricci tensor transforms as
\begin{equation}\label{confRicci}
\Rb_{ij}=R_{ij}-(n-2)\nabla_i \nabla_j \Phi + (n-2)\nabla_i \Phi \nabla_j \Phi -(\Box_g \Phi + (n-2)|\nabla\Phi|^2_g)g_{ij},
\end{equation}
where $\nabla_i$ denotes the Levi-Civita connection associated with $g_{ij}$,
$\Box_g = g^{ij}\nabla_i\nabla_j$ is the wave operator, and
$|\nabla\Phi|_g^2 = g^{ij}\nabla_i\Phi \nabla_j\Phi$. 
The connection coefficients of $\gb_{ij}$ and $g_{ij}$ satisfy the relation
\begin{equation}\label{confChrist}
\Gammab_{i}{}^k{}_j-\Gamma_{i}{}^k{}_j=g^{kl}(g^{il}\nabla_j \Phi + g_{jl}\nabla_i\Phi -g_{ij}\nabla_l\Phi).
\end{equation}
Substituting \eqref{confRicci} into the Einstein equations \eqref{ESF.1}, we obtain
\begin{equation}\label{confESFA}
-2R_{ij}=-2(n-2)\nabla_i \nabla_j \Phi + 2(n-2)\nabla_i \Phi \nabla_j \Phi -2(\Box_g \Phi + (n-2)|\nabla\Phi|^2_g)g_{ij}-4\nabla_{i}\phi\nabla_{j} \phi-\frac{8}{n-2}V(\phi)\gb_{ij}.
\end{equation}

We now introduce an arbitrary Lorentzian background metric $\gc_{ij}$, not required to satisfy any field equations, which is later chosen to simplify our analysis.
We denote by $\Dc_i$  the Levi-Civita connection corresponding to the background metric $\gc_{ij}$, and by $\gamma_{i}{}^k{}_{j}$ the connection coefficients with respect to a chosen frame. The scalar field equation \eqref{ESF.2} then takes the form
\begin{equation*}
  \Box_{\gb} \phi=\gb^{ij}\Dc_{i}\Dc_{j}\phi-\gb^{ij}(\Gammab_{i}{}^k{}_j-\Gamma_i{}^k{}_j+\Cc_{i}{}^k{}_{j})\Dc_{k}\phi = V'(\phi)
\end{equation*}
where
\begin{equation}\label{Ccdef}
\Cc_{i}{}^k{}_j := \Gamma_{i}{}^k{}_j-\gamma_{i}{}^k{}_j=\frac{1}{2}g^{k l}\bigl(\Dc_i g_{j l}+\Dc_j g_{i l}-\Dc_l g_{ij}\bigr)
\end{equation}
is a tensor field.
Using \eqref{confmet} and \eqref{confChrist}, we can express the scalar field equation as
\begin{equation*} %\label{confESFB}
  g^{ij}\Dc_{i}\Dc_{j}\phi= X^k\Dc_{k}\phi -(n-2)g^{ij}\Dc_{i}\Phi\Dc_{j}\phi+e^{2\Phi}V'(\phi),
\end{equation*}
where
\begin{equation} \label{Xdef}
  X^k := g^{ij}\Cc_{i}{}^k{}_j=\frac{1}{2}g^{ij}g^{k l}\bigl(2\Dc_i g_{j l}-\Dc_l g_{ij}\bigr),
\end{equation}
or equivalently as
\begin{equation} \label{confESFC}
  \Box_g \phi = -(n-2)\nabla^i\Phi\nabla_i\phi +e^{2\Phi}V'(\phi).
\end{equation}
We note that here and below that all indices are raised and lowered using the conformal metric, e.g.,  $\nabla^k\Phi = g^{kl}\nabla_l\Phi$. 

Until now, the function $\Phi$ in the conformal transformation \eqref{confmet} has been arbitrary. In order to fix some irrelevant degrees of freedom, we now specify the form\footnote{The constant $\Phi_0$, while unnecessary in the Kasner setting \cite{BeyerOliynyk:2021}, is essential in the analysis of the ekpyrotic regime.}
\begin{equation}\label{gaugefixA}
\Phi =\lambda \phi+\Phi_0
\end{equation}
of $\Phi$
where $\lambda$ and $\Phi_0$ are free real constants.
Next, we introduce the scalar field $\tau$ via
\begin{equation} \label{taudef}
\tau=e^{-\alpha\phi}  \quad \Longleftrightarrow \quad \phi = -\frac{1}{\alpha}\ln(\tau)
\end{equation}
where $\alpha$ is a free non-zero constant. Using \eqref{taudef}, the scalar field potential \eqref{eq:exponentialpotential} can be expressed as a function of $\tau$ as follows:
\begin{equation}
\label{eq:DefVtt}
\Vtt(\tau):=V\biggl(-\frac{1}{\alpha}\ln(\tau)\biggr)=V_0\tau^{s/\alpha}.
\end{equation}
Moreover, the derivative of $V(\phi)$ satisfies
\begin{equation} \label{V'(phi)}
V'(\phi)=\Vtt'(\tau)\frac{d\tau}{d\phi}
=-\alpha\tau \Vtt'(\tau)=-s \Vtt(\tau)=-sV_0\tau^{s/\alpha},
\end{equation}
and we note that
\begin{equation}\label{exp(2Phi)}
e^{2\Phi}=e^{2\Phi_0}\tau^{-2\lambda/\alpha}.
\end{equation}

In terms of $\tau$, the scalar field equation \eqref{confESFC} takes the form
\begin{equation}\label{box-tau}
   \Box_g\tau=g^{ij}\nabla_i \nabla_j \tau
  = \Bigl((n-2)\frac{\lambda}{\alpha}+1\Bigr) \tau^{-1}\nabla^i\tau\nabla_i\tau 
  +s \alpha V_0 e^{2\Phi_0}\tau^{(s-2\lambda)/\alpha+1},
\end{equation}
while the Einstein equations \eqref{confESFA} become
\begin{align}
  R_{ij}
  =&-\frac{(n-2)\lambda}{\alpha}\tau^{-1}\nabla_i \nabla_j\tau
  +\frac{(n-2)\lambda \alpha+2-(n-2)\lambda^2}{\alpha^2}  
  \tau^{-2}\nabla_i \tau \nabla_j \tau \notag \\
  &
  +\Bigl(-\lambda s+\frac{4}{n-2}\Bigr) V_0 e^{2\Phi_0}\tau^{(s-2\lambda)/\alpha}g_{ij}. \label{Rij-tau}
\end{align}
Assuming that $\lambda^2\neq 2/(n-2)$ and that $\lambda\neq 0$, and following \cite{BeyerOliynyk:2021}, we set
\begin{equation}\label{alphafix}
\alpha= -\frac{2-\lambda^2(n-2)}{\lambda (n-2)},
\end{equation}
which eliminates the second term on the right-hand side of \eqref{Rij-tau}. We also introduce a constant $\Rtt$ to represent the constant coefficient in the first term on the right-hand side of\footnote{In the Kasner setting \cite{BeyerOliynyk:2021}, choosing $\lambda$ so that $\Rtt=0$ is sufficient. For the ekpyrotic regime, however, it is essential to allow $\Rtt\neq 0$. This introduces a major complication absent in \cite{BeyerOliynyk:2021} due to the presence of the term $\tau^{-1}\nabla^i\tau\nabla_i\tau$ on the right-hand side of \eqref{box-tau} in the ekpyrotic regime.}  \eqref{box-tau}:
\begin{equation}
\label{eq:RttDef}
\Rtt=\frac{(n-2)\lambda}{\alpha}+1=\frac{2-\lambda^2(n-2)(n-1)}{2-\lambda^2(n-2)}
\quad\Longleftrightarrow\quad
\lambda^2={\frac{2 (1 - \Rtt)}{(n - 2) (n - 1 - \Rtt)}}.
\end{equation}
To ensure that $\Phi\to -\infty$ as $\tau\searrow 0$, we impose $\lambda/\alpha<0$. Using \eqref{alphafix} and $n>2$, this translates to $0<\lambda^2<\frac{2}{n-2}$, corresponding to $\Rtt\in(-\infty,1)$ in \eqref{eq:RttDef}. Choosing $\lambda>0$, we obtain
\begin{equation}
\label{eq:lambdaByR}
\lambda=\sqrt{{\frac{2 (1 - \Rtt)}{(n - 2) (n - 1 - \Rtt)}}}.
\end{equation}
Thus
\begin{equation}
\label{eq:alphaByR}
\alpha\overset{\eqref{alphafix}}{=}-\sqrt{\frac{2(n-2)}{(1-\Rtt)(n-1-\Rtt)}}<0,
\end{equation}
implying that $\lim_{\tau\searrow 0}\phi=-\infty$ as a consequence of \eqref{taudef}. This is consistent with restricting our analysis to \emph{positive} ekpyrotic/Kasner regimes; cf.\ Section~\ref{sec:FLRWfpsol}.
Summarising, we have shown that
\begin{equation} \label{confESFAa.1}
R_{ij}=\frac{1-\Rtt}{\tau}\nabla_i \nabla_j \tau
+\Ttt_{ij},
\end{equation}
where
\begin{equation} \label{eq:genTtt}
\Ttt_{ij}
=\frac{2V_0 e^{2\Phi_0}}{n-1-\Rtt}\Ltt \tau^{\Ltt-2}g_{ij},\quad
\Ltt=2+\frac{s-2\lambda}{\alpha}
=\frac{n-1-\Rtt}2\Bigl(\frac{4}{n-2}-\lambda s\Bigr),
\end{equation}
and
\begin{equation} \label{confESFAa.2}
\Box_g \tau = \Wtt
\end{equation}
with 
\begin{equation} \label{Wtt-def}
\Wtt = \frac{\Rtt}{\tau} |\nabla\tau|_g^2
+s \alpha V_0 e^{2\Phi_0}\tau^{\Ltt-1}.
\end{equation}
Collecting \eqref{confESFAa.1} and \eqref{confESFAa.2} together, we have
\begin{align} \label{confESFAaN}
R_{ij}&=(1-\Rtt){\tau^{-1}}\nabla_i \nabla_j \tau
+\Ttt_{ij},\\
\label {confESFCbN}
  \Box_g \tau &= \Wtt.
\end{align}
We refer to these equations as the \textit{conformal Einstein-scalar field equations}. From \eqref{confmet} and \eqref{gaugefixA}--\eqref{alphafix}, each solution $\{g_{ij},\tau\}$ of these equations yields the physical solution
\begin{equation}
  \label{eq:conf2phys}
  \biggl\{\gb_{ij}=e^{2\Phi_0}\tau^{\frac{2(1-\Rtt)}{n-2}} g_{ij},\,\phi=\sqrt{\frac{(1-\Rtt)(n-1-\Rtt)}{2(n-2)}}\ln(\tau)\biggr\}
\end{equation}
of the \textit{physical Einstein-scalar field equations} \eqref{ESF.1}--\eqref{eq:exponentialpotential}. By setting $\Rtt=0$ and $\Phi_0=0$, the above equations become identical to those from \cite[\S1.1]{BeyerOliynyk:2021}.

\subsection{Reduced conformal Einstein-scalar field equations}
To establish the existence of solutions to the conformal Einstein-scalar field equations, we replace \eqref{confESFAaN} with a gauge-reduced version. We employ a conformal wave gauge defined by
\begin{equation}\label{wave-gauge}
X^k =0,
\end{equation}
where $X^k$ is defined in \eqref{Xdef}. We then consider the corresponding wave-gauge-reduced equations
\begin{align}
  -2R_{ij}+2\nabla_{(i} X_{j)}
  &=-\frac{2(1\added{-\Rtt})}{\tau}\nabla_i \nabla_j \tau
  \added{-2\Ttt_{ij}} \notag\\
  &\oset{\eqref{Ccdef}}{=}-\frac{2(1\added{-\Rtt})}{\tau}\bigl(
  \Dc_i \Dc_j \tau - \Cc_i{}^k{}_j\Dc_k \tau  \bigr)
  \added{-2\Ttt_{ij}}, \label{confESFF}
\end{align}
which we refer to as the \textit{reduced conformal Einstein equations}.

To justify the spacetime imposition of the wave gauge constraint \eqref{wave-gauge}, we cite Proposition \ref{lag-exist-prop}.(e), which guarantees that this wave gauge constraint propagates.
Using \eqref{Ccdef}, \eqref{Xdef}, and \eqref{wave-gauge}, the conformal scalar field equation \eqref{confESFAa.2} can be written as
\begin{equation} \label{confESFG}
g^{ij}\Dc_{i}\Dc_{j}\tau=\Wtt.
\end{equation}
Combining \eqref{confESFF} and \eqref{confESFG}, we obtain the system
\begin{align}
-2R_{ij}+2\nabla_{(i} X_{j)}& =-\frac{2(1\added{-\Rtt})}{\tau}\bigl(
  \Dc_i \Dc_j \tau - \Cc_i{}^k{}_j\Dc_k \tau  \bigr)
  \added{-2\Ttt_{ij}}, 
\label{confESFFa}\\
 g^{ij}\Dc_{i}\Dc_{j}\tau    &=\added{\Wtt}, \label{confESFGa}
\end{align}
which we refer to as the \textit{reduced conformal Einstein-scalar field equations}.

For use below, we recall that the reduced Ricci tensor is given by
\begin{equation} \label{red-Ricci}
-2R_{ij}+2\nabla_{(i} X_{j)}=g^{kl}\Dc_k \Dc_l g_{ij} + Q_{ij}
+2g^{kl}g_{m(i}\Rc_{j)kl}{}^m
\end{equation}
where 
\begin{equation}\label{Q-def}
Q_{ij} = \frac{1}{2}g^{kl}g^{mn}\Bigl(\Dc_i g_{mk} \Dc_j g_{n l}
+2 \Dc_{n}g_{il}\Dc_{k}g_{jm} - 2\Dc_{l}g_{in} \Dc_{k}g_{jm}
-2 \Dc_{l}g_{in}\Dc_j g_{mk} -2 \Dc_{i}g_{mk}\Dc_{l}g_{jn}\Bigr)
\end{equation}
and, as above, $\Rc_{ijk}{}^l$ denotes the Riemann tensor of the background metric $\gc_{ij}$. Differentiating \eqref{confESFGa} and using the commutator formula
$\Dc_k\Dc_i \Dc_j\tau-\Dc_i\Dc_j \Dc_k \tau
=\Rc_{kij}{}^l\Dc_l\tau$,
we observe that
\begin{align} \label{confESFI} 
  g^{ij}\Dc_i\Dc_{j} \Dc_{k}\tau
  =& g^{il}g^{jm}\Dc_kg_{lm}  \Dc_{i} \Dc_{j}\tau
     -g^{ij}\Rc_{kij}{}^l \Dc_l \tau 
   \added{
     +\Dc_k\Wtt}.
\end{align}

\subsection{Fixing the background metric}
So far, the background metric $\gc_{ij}$ has been arbitrary. Because the FLRW metric \eqref{eq:FLRWFPsol.11} is conformally flat and our focus is on nonlinear perturbations of this solution, we restrict attention to flat background metrics. Consequently the curvature tensor vanishes
\begin{equation}\label{curvature}
\Rc_{ijk}{}^l = 0
\end{equation}
and therefore the background covariant derivatives $\Dc_i$ commute
\begin{equation}\label{commutator}
[\Dc_i,\Dc_j] =0.
\end{equation}

\subsection{Conformal representation of the positive Kasner and ekpyrotic FLRW solutions}
\label{sec:confrep_ekpyrotic}
Next, we transform the FLRW solutions \eqref{eq:FLRWFPsol.11} and \eqref{eq:FLRWFPsol.21}, with $\xtt_i$ from \eqref{eq:FLRWFP} and $s_c$ from \eqref{sc-def}, into the conformal representation \eqref{eq:conf2phys}. In this section we only present the \emph{positive} Kasner and ekpyrotic classes: $\xtt_i=1$ for the positive Kasner case, and $\xtt_i=s/s_c$ with $s>0$ for the positive ekpyrotic case. For the latter case, we further assume that $s>s_c$, which is expected to be stable to the past according to Section~\ref{sec:FLRWstabilityanalsys}.

Under these restrictions and assuming that $n\geq 3$, the inequality
\begin{equation}
\label{eq:FLRWFPrestr}
(n-1)x_i^2>1
\end{equation}
holds. 
Using \eqref{eq:FLRWFPrestr}, we perform a coordinate transformation from the Gaussian time coordinate $\tb$ in \eqref{eq:FLRWFPsol.11} and \eqref{eq:FLRWFPsol.21} to a new time coordinate $t$ defined by
\begin{equation}
  \label{eq:FLRWtimetrafo}
  \tb=\Theta\, t^{\frac{(n-1)x_i^2}{(n-1)x_i^2-1}},
\end{equation}
where $\Theta>0$ is an arbitrary constant. After suitable constant rescaling of the spatial coordinates, \eqref{eq:FLRWFPsol.11} and \eqref{eq:FLRWFPsol.21} become
\begin{equation}
  \label{eq:FLRWFPsol.2}
  \begin{aligned}
  \bar g&=\Theta^2 \left(\frac{(n-1)x_i^2}{(n-1)x_i^2-1}\right)^{2} t^{\frac{2}{(n-1)x_i^2-1}}\Bigl[-dt\otimes dt+\sum_{\Lambda=1}^{n-1}dx^\Lambda\otimes dx^\Lambda\Bigr], \\
  \phi&=
  \frac2{s_c}\frac{(n-1)x_i}{(n-1)x_i^2-1}\log t,
  \end{aligned}
\end{equation}
where we choose the constant $\Theta$ in \eqref{eq:FLRWtimetrafo} as
\begin{equation}
  \label{eq:FLRWphiStar}
  \Theta^2=e^{-s_c x_i\phi_*},
\end{equation}
in order to cancel the constant term in \eqref{eq:FLRWFPsol.21}.
For the Kasner case $\xtt_i=1$, the constraint \eqref{eq:FLRWFPConstr} is satisfied when $V_0=0$ irrespective of the value of $e^{-s\phi_*}$. We may therefore pick the value of this last constant so that $\Theta$ in \eqref{eq:FLRWphiStar} is $\Theta=(n-2)/(n-1)$. Using \eqref{eq:FLRWphiStar} and \eqref{sc-def}, \eqref{eq:FLRWFPsol.2} then becomes
\begin{equation*}
  %\label{eq:FLRWFPsol.2.Kasner}
  \bar g= t^{\frac{2}{n-2}}\Bigl[-dt\otimes dt+\sum_{\Lambda=1}^{n-1}dx^\Lambda\otimes dx^\Lambda\Bigr], \quad 
  \phi=\sqrt{\frac{n-1}{2(n-2)}}\log t.
\end{equation*}
The form of these expressions agree with \eqref{eq:conf2phys} with $\tau=t$, $g_{ij}=\eta_{ij}$ ($\eta_{ij}$ is the Minkowski metric) and $\Phi_0=\Rtt=0$. Together with \eqref{eq:lambdaByR}, \eqref{eq:alphaByR}, and \eqref{eq:genTtt}, this implies that
\begin{equation*}
\lambda=\sqrt{{\frac{2}{(n - 2) (n - 1 )}}},\quad
\alpha=-\sqrt{\frac{2(n-2)}{n-1}}.
\end{equation*}
We note that $\Ltt$ is irrelevant in the Kasner case since it appears multiplied by $V_0$ in \eqref{confESFAaN} and \eqref{confESFCbN}, and $V_0$ vanishes in this case. Although we do not pursue it here, the values $\Phi_0=\Rtt=0$ essentially reduce the positive-Kasner stability problem to the zero-potential setting treated in \cite{BeyerOliynyk:2021}. For any $s<s_c$ and $V_0\in\Rbb$, the argument of \cite{BeyerOliynyk:2021} can be adapted to obtain a spatially localised nonlinear big-bang stability result, complementing \cite{groeniger2023}.

For the positive ekpyrotic case with $s>s_c$ and $V_0=-1$,
the constraint \eqref{eq:FLRWFPConstr.ekp} implies
\begin{equation}
  \label{eq:FLRWFPConstr.ekp.2}
  {e^{s\phi_*}}=\frac 12 \frac{s^4}{s^2-s_c^2},
\end{equation}
and by \eqref{eq:FLRWphiStar}, this yields
\begin{equation*}
  %\label{eq:FLRWphiStar}
  \Theta^2=e^{-s\phi_*}=2 \frac{s^2-s_c^2}{s^4}.
\end{equation*}
Thus, \eqref{eq:FLRWFPsol.2} becomes
\begin{align*}
  %\label{eq:FLRWFPsol.2.ekp}
  \bar g&=2 \frac{s^2-s_c^2}{s^4}\left(\frac{(n-1)s^2}{(n-1)s^2-s_c^2}\right)^{2} t^{\frac{2(1-\Rtt)}{n-2}}\biggl(-dt\otimes dt+\sum_{\Lambda=1}^{n-1}dx^\Lambda\otimes dx^\Lambda\biggr), \\
  \phi&=\frac{2(n-1)s}{(n-1)s^2-s_c^2}\log t.
\end{align*}
The form of these expressions agree with \eqref{eq:conf2phys} with $\tau=t$, $g_{ij}=\eta_{ij}$  and (see~\eqref{eq:lambdaByR}, \eqref{eq:alphaByR}, and \eqref{eq:genTtt})
\begin{equation}
  \label{eq:ekpyroticparameterconstants.N}
  \Rtt=\frac{(n-1)(s^2-s_c^2)}{(n-1)s^2-s_c^2},\,\, 
  e^{2\Phi_0}=2\frac{(n-1)^2({s^2}-s_c^2)}{((n-1)s^2-s_c^2)^2},\,\,
  \Ltt=0,\,\,
  \lambda=\frac {s_c^2 s}{2(n - 1)},\,\, \alpha=-\frac{(n-1)s^2-s_c^2}{2(n-1)s}.
\end{equation} 
The fact that $\Rtt$ does not vanish in the ekpyrotic regime turns out to be of crucial importance for the results in this paper.

\begin{rem}
For the remainder of the article, we assume that \begin{equation*}
s>s_c=\sqrt{\frac{8(n-1)}{n-2}}, \quad V_0=-1,
\end{equation*}
and the constants $\Rtt$, $e^{2\Phi_0}$, $\Ltt$, $\lambda$, and $\alpha$ are determined by \eqref{eq:ekpyroticparameterconstants.N}. With these choices, we note that 
\begin{equation*}
    0<\Rtt<1,
\end{equation*}  
\begin{equation*}
 % \label{eq:ekpyroticRtt}
 s\alpha V_0e^{2\Phi_0}= \frac{2 e^{2\Phi_0}}{ n - 1 - \Rtt} \frac{s}{\lambda}=\Rtt,
\end{equation*}
and that
\begin{equation}\label{Ttt-Wtt-fix}
\Ttt_{ij}=0 \AND \Wtt = \frac{\Rtt}{\tau}(|\nabla\tau|_g^2+1).
\end{equation}
\end{rem}

\section{Local existence and continuation in Lagrangian coordinates}
\label{sec:locexist_cont}

The local-in-time existence in Lagrangian coordinates of solutions to the conformal Einstein-scalar field equations with a vanishing potential along with a continuation principle for these solutions was established in Proposition~5.6 of  \cite{BeyerOliynyk:2021}. Since the addition of a potential only modifies the sub-principal terms, it is straightforward to adapt the proof of  \cite[Prop.~5.6]{BeyerOliynyk:2021} so that it applies to the conformal Einstein-scalar field equations with a non-vanishing potential. Due to the simplicity of the modifications required, we omit the proof and only present the final result. However, before we can state it, we first need to setup the conformal Einstein-scalar field equations in Lagrangian coordinates and describe the initial data. 

Thus far, the frame in which the reduced conformal Einstein-scalar field equations \eqref{confESFFa}-\eqref{confESFGa} are formulated is arbitrary. To proceed, we need to fix the frame, which we do by introducing coordinates 
$(\xh^\mu)=(\xh^0,\xh^\Lambda)$ on a spacetime $M_{t_1,t_0}$ of the form \ref{Mt1t0-def}, as described in Section \ref{indexing}, and then choose the frame to be the coordinate frame determined by this choice of coordinates. It is important to note here that the coordinate frame is only useful for establishing the local-in-time existence of solutions. In order to establish global estimates, we need to use an orthonormal frame as described in Section \ref{Fermi}.

We proceed by fixing the background metric $\gc$ so that its coordinate components (relative to the coordinates $(\xh^\mu)$) take the form
\begin{equation}\label{gct-def}
    \gchat_{\mu\nu}=\eta_{\mu\nu}:= -\delta_\mu^0\delta_\nu^0 + \delta_\mu^\Lambda \delta_\nu^\Gamma \delta_{\Gamma\Lambda}.
\end{equation}
With this choice, the covariant derivative associated to the background metric reduces to partial differentiation, i.e.,~
$\hat{\Dc}_\mu = \delh{\mu}$. 
Then it follows from \eqref{Ccdef}, \eqref{red-Ricci}, \eqref{curvature} and \eqref{Ttt-Wtt-fix} that we can express the reduced conformal Einstein-scalar field equations \eqref{confESFFa}-\eqref{confESFGa} as
\begin{align}
\gh^{\alpha\beta}\delh{\alpha}\delh{\beta}\gh_{\mu\nu} + \Qh_{\mu\nu} &= -\frac{2(1-\Rtt)}{\tauh}\bigl( \delh{\mu}\delh{\nu}\tauh - \Gammah^\gamma_{\mu\nu} \delh{\nu}\tauh \bigr)%-2\hat{\Ttt}_{\mu\nu}
, \label{tconfESF-A.1}\\
\gh^{\alpha\beta}\delh{\alpha}\delh{\beta}\tauh &= \hat{\Wtt},  \label{tconfESF-A.2}
\end{align} 
where $\tauh$ denotes the scalar field $\tau$ viewed as a function of the coordinates $(\xh^\mu)$, $\gh_{\mu\nu}$ are the components of the conformal metric $g$ with respect to the coordinates $(\xh^\mu)$, $\Gammah^\gamma_{\mu\nu} =
     \frac{1}{2}\gh^{\gamma\lambda}\bigl(\delh{\mu}\gh_{\nu\lambda} +\delh{\nu}\gh_{\mu\lambda}-\delh{\lambda}\gh_{\mu\nu}\bigr)$
are the Christoffel symbols of $\gh_{\mu\nu}$,
\begin{equation}
  \label{eq:defQh}
\Qh_{\mu \nu} = \frac{1}{2}\gh^{\alpha \beta}\gh^{\sigma \delta}\bigl(\delh{\mu} \gh_{\sigma \alpha} \delh{\mu} \gh_{\delta \beta}
+2 \delh{\delta}\gh_{\mu \beta}\delh{\alpha}\gh_{\nu \sigma} - 2\delh{\beta}\gh_{\mu \delta} \delh{\alpha}\gh_{\nu \sigma}
-2 \delh{\beta}\gh_{\mu \delta}\delh{\nu} \gh_{\sigma\alpha} -2 \delh{\mu }\gh_{\sigma\alpha}\delh{\beta}\gh_{\nu \delta}\bigr),
\end{equation}
and
\begin{equation} \label{hatWtt-def}
\hat{\Wtt} 
=\frac{\Rtt}{\tauh}\bigl( \gh^{\mu\nu}\delh{\mu}\tauh\delh{\nu}\tauh
+1\bigr).
\end{equation}
We additionally note from \eqref{Ccdef} and \eqref{Xdef} that the coordinate components of the wave gauge vector field $X$ are determined by
\begin{equation} \label{Xt-rep}
    \Xh^\gamma = \gh^{\mu\nu}\Gammah^\gamma_{\mu\nu}.
\end{equation}

\subsection{Initial data}
In the remainder of this article, we specify initial data for the Einstein-scalar field equations \eqref{tconfESF-A.1}-\eqref{tconfESF-A.2} on hypersurfaces
\begin{equation*}
\Sigma_{t_0}=\{t_0\}\times\Tbb^{n-1},
\end{equation*}
and we label the initial data as follows:
\begin{align}
    \gh_{\mu\nu}\bigl|_{\Sigma_{t_0}} &= \gr_{\mu\nu}, \label{gt-idata} \\
    \delh{0}\gh_{\mu\nu}\bigl|_{\Sigma_{t_0}} &= \grave{g}_{\mu\nu}, \label{dt-gt-idata} \\
    \tauh\bigl|_{\Sigma_{t_0}} &= \taur, \label{taut-idata}\\
    \delh{0}\tauh \bigl|_{\Sigma_{t_0}} &= \taug.\label{dt-taut-idata}
\end{align}
In order to guarantee that solutions of the reduced conformal Einstein-scalar field equations \eqref{tconfESF-A.1}-\eqref{tconfESF-A.2} also satisfy the conformal Einstein-scalar field equation \eqref{confESFAaN}-\eqref{confESFCbN}, we always assume that the initial data \eqref{gt-idata}-\eqref{dt-taut-idata} satisfies the constraint equations
\begin{align}
 \biggr( \Gh^{0\nu} - \frac{1-\Rtt}{\tauh}\nablah^0\nablah^\nu\tauh+\frac{1-\Rtt}{2\tauh}\hat{\Wtt} \gh^{0\nu}\biggr)\biggl|_{\Sigma_{t_0}} &=0, \qquad \text{(gravitational constraints)}
  \label{grav-constr}\\
  \Xh^\mu\bigl|_{\Sigma_{t_0}} &=0, \qquad \text{(wave gauge constraints)}
  \label{wave-constr}
\end{align}
where
$\Gh^{\mu\nu}$ and $\nablah_\mu$ are the Einstein tensor and the Levi-Civita connection of the conformal metric $\gh_{\mu\nu}$, respectively. 

\begin{rem} \label{idata-rem}
The \emph{geometric initial data} on the hypersurface $\Sigma_{t_0}$ is a subset of the initial data \eqref{gt-idata}-\eqref{dt-taut-idata} consisting of the fields $\{\gtt,\Ktt,\taur,\taugr\}$
where $\gtt=\gtt_{\Lambda\Omega}d\xh^\Lambda \otimes d\xh^\Omega$ is the \textit{spatial metric} and $\Ktt=\Ktt_{\Lambda\Omega}d\xh^\Lambda \otimes d\xh^\Omega$ is the \textit{second fundamental form}, which are determined from the initial data $\{\gr_{\mu\nu}, \grave{g}_{\mu\nu}, \taur, \taug\}$ via
\begin{equation}
  \label{gtt-def1}
  \gtt_{\Lambda\Omega} = \gr_{\Lambda\Omega} \AND \Ktt_{\Lambda\Omega} = \frac{1}{2\Ntt}(\ggr_{\Lambda\Omega}-2\Dtt_{(\Lambda} \btt_{\Omega)}),
\end{equation}
respectively. Here, 
$\btt_\Lambda = \gr_{0\Lambda}$ and $\Ntt^2 = -\gr_{00}+ \btt^\Lambda\btt_\Lambda$
define the \textit{shift}
$\btt=\btt_\Lambda d\xh^\Lambda$ and  \textit{lapse} $\Ntt$, respectively,  $\Dtt_\Lambda$
denotes the Levi-Civita connection of the spatial metric $\gtt_{\Lambda\Omega}$, and $\btt^\Lambda = \gtt^{\Lambda\Omega}\btt_\Omega$ where $(\gtt^{\Lambda\Omega})=( \gtt_{\Lambda\Omega})^{-1}$ is the inverse spatial metric. The importance of the geometric initial data is that it determines the physical part (i.e., non-gauge) of the initial data. In particular, the gravitational constraint equations \eqref{grav-constr} can be formulated entirely in terms of the geometric initial data, and for a given choice of geometric initial data, the remaining initial data can always be chosen so that the wave gauge constraints \eqref{wave-constr} are satisfied; see \cite[Rem.~5.1]{BeyerOliynyk:2021} for details.
\end{rem}

\subsection{First-order formulation\label{fof}}
As in \cite[\S5.3]{BeyerOliynyk:2021}, we introduce first-order variables
\begin{gather}
    \hh_{\beta\mu\nu} = \delh{\beta}\gh_{\mu\nu}, \quad
    \vh_{\mu} = \delh{\mu}\tauh \AND
    \wh_{\mu\nu} = \delh{\mu}\delh{\nu}\tauh, \label{tvars}
\end{gather}
and we define the vector field $\chih^\mu$ via
\begin{equation}
  \chih^\mu = \frac{1}{|\vh|_{\gh}^2} \vh^{\mu}, \label{chit-def}
\end{equation}
which we note satisfies $\chih(\tauh)=1$.
We always assume that $\vh^\mu$ is timelike (i.e.,\ $|\vh|_{\gh}^2 <0$) in order to ensure that $\hat{\chi}^\mu$ remains well defined and timelike. Then proceeding in the same fashion as in \cite[\S5.3]{BeyerOliynyk:2021}, we obtain the following first-order formulation of the reduced conformal Einstein-scalar field equations \eqref{tconfESF-A.1}-\eqref{tconfESF-A.2}:
\begin{align}
   \Bh^{\lambda\beta\alpha} \delh{\alpha}\hh_{\beta\mu\nu}&= \chih^\lambda \Bigr(\Qh_{\mu\nu}+\frac{2(1-\Rtt)}{\tauh}\bigl( \wh_{(\mu \nu)} - \Gammah^\gamma_{\mu\nu} \vh_\gamma\bigl)
   \Bigr), \label{tconf-ford-B.1}\\
     \Bh^{\lambda\beta\alpha}\delh{\alpha}\wh_{\beta\mu}
     & =-\chih^\lambda\bigl(\gh^{\alpha \sigma}\gh^{\beta\delta}\hh_{\mu\sigma\delta}\wh_{\alpha\beta}+\hat{\Wtt}_\mu\bigr), \label{tconf-ford-B.2}\\
      \Bh^{\lambda\beta\alpha}\delh{\alpha}\zh_{\beta} &=-\chih^\lambda \hat{\Wtt}, \label{tconf-ford-B.3} \\
      \chih^\alpha \delh{\alpha}\gh_{\mu\nu}&= \chih^\alpha \hh_{\alpha\mu\nu}, \label{tconf-ford-B.4}\\
     \chih^\alpha \delh{\alpha} \vh_{\mu}&= \chih^\alpha \wh_{\alpha\mu},  \label{tconf-ford-B.5}\\
     \chih^\alpha \delh{\alpha} \tauh&= \chih^\alpha \zh_{\alpha},  \label{tconf-ford-B.6}
\end{align}
where 
\begin{equation}\label{Bt-def}
 \Bh^{\lambda\beta\alpha}=    -\chih^\lambda\gh^{\beta\alpha} -\chih^\beta \gh^{\lambda\alpha} 
  + \gh^{\lambda\beta}\chih^\alpha,
\end{equation}
\begin{equation}\label{hatWtt-mu-def}
\hat{\Wtt}_\mu 
= -\frac{\Rtt}{\tauh}\gh^{\alpha\sigma}\gh^{\beta\delta}\hh_{\mu\sigma\delta}\zh_\alpha\zh_\beta-\frac{\Rtt}{\tauh^2}\gh^{\alpha\beta}\zh_{\alpha}\zh_{\beta}\zh_\mu + \frac{2\Rtt}{\tauh}\gh^{\alpha \beta}\zh_\alpha \wh_{\beta\mu}-\frac{\Rtt}{\tauh^{2}}\zh_\mu,
\end{equation}
and we interpret $\zh_\mu$ as the derivative of $\tauh$ as follows: 
\begin{equation} \label{zh-def}
    \zh_\mu = \delh{\mu}\tauh.
\end{equation}

\subsection{Lagrangian coordinates\label{Lag-coordinates}}
Following \cite{BeyerOliynyk:2021}, we now introduce Lagrangian coordinates $(x^\mu)$ that are adapted to the vector field $\chih^\alpha$. The reason for introducing these coordinates is that they allow us to use the scalar field $\tau$ as a time coordinate. The advantage of using $\tau$ as a time coordinate is that it synchronizes the singularity; see Section \ref{temp-synch} below for details.

The Lagrangian coordinates $(x^\mu)$ are defined via the map
\begin{equation} \label{lemBa.1a}
\xh^\mu = l^\mu(x) := \Gc_{x^0-t_0}^\mu(t_0,x^\Lambda), \quad \forall\, (x^0,x^\Lambda)\in M_{t_1,t_0},
\end{equation}
where $\Gc_s(\xh^\lambda) = (\Gc^\mu_s(\xh^\lambda))$
denotes the flow map of $\chih^\mu$, i.e., 
\begin{equation*}
\frac{d\;}{ds} \Gc^\mu_s(\xh^\lambda) = \chih^\mu(\Gc_s(\xh^\lambda)) \AND
\Gc^\mu_{0}(\xh^\lambda) = \xh^\mu.
\end{equation*}
We note that the map $\ell^\mu$ defines a diffeomorphism
\begin{equation*}
    l \: : \: M_{t_1,t_0} \longrightarrow l(M_{t_1,t_0})\subset M_{-\infty,t_0}
\end{equation*}
that satisfies $l(\Sigma_{t_0})=\Sigma_{t_0}$ so long as the vector field $\chih^\mu$ does not vanish and remains sufficiently regular. 
In agreement with our coordinate conventions, we often use 
$t=x^0$
to denote the Lagrangian time coordinate.

It follows from \eqref{lemBa.1a} that $l^\mu$ solves the initial value problem 
\begin{align}
    \del{0}l^\mu &= \chihu^\mu, \label{l-ev.1} \\
    l^\mu(t_0,x^\Lambda) &= \delta^\mu_0 t_0 +\delta^\mu_\Lambda x^\Lambda, \label{l-ev.2} 
\end{align}
where we employ the notation
\begin{equation}\label{fu-def}
    \underline{f}=f\circ l
\end{equation}
to denote the pull-back of scalars by the Lagrangian map $l$. In what follows, we need to distinguish the geometric pull-back of a field using $l$ from pulling-back the components of the field as scalars. In particular, the notation \eqref{fu-def} is used to denote the pull-back of components as scalars while the removal of the ``hat'' denotes the geometric pull-back. For example,
\begin{equation}\label{chi-lag}
\chi^\mu = \Jcch_\nu^\mu \chihu^\nu \oset{\eqref{l-ev.1}}{=} \Jcch_\nu^\mu \Jc_0^\nu =\delta^\mu_0,
\end{equation}
\begin{equation} \label{tau-lag}
\tau = \underline{\tauh},
\end{equation}
and
\begin{equation} \label{g-lag}
    g_{\mu\nu} = \Jc^\alpha_\nu \Jc^\beta_\mu\ghu_{\alpha\beta},
\end{equation}
where
\begin{equation} \label{Jc-def}
\Jc^\mu_\nu = \del{\nu}l^\mu
\end{equation}
is the Jacobian matrix of the map $l^\mu$,
$(\Jcch^\mu_\nu):= (\Jc^\mu_\nu)^{-1}$    
is its inverse, and we recall from \cite[\S5.4]{BeyerOliynyk:2021} that 
$\Jc^\mu_\nu$ satisfies the
initial value problem \begin{align}
    \del{0}\Jc^\mu_\nu &= \Jc^\lambda_\nu \Jsc_\lambda^\mu, \label{J-ev.1} \\
    \Jc^\mu_\nu(t_0,x^\Lambda) &= \delta^0_\nu \chih^\mu(t_0,x^\Lambda)+\delta_\nu^\Lambda\delta^\mu_\Lambda, \label{J-ev.2} 
\end{align}
with
\begin{equation} \label{Jsc-def}
  \begin{split}
    \Jsc_\lambda^\mu =\frac{1}{|\vhu|^2_{\ghu}} \biggl(&
   \ghu^{\mu\sigma}\whu_{\lambda\sigma}-\ghu^{\mu\tau}\ghu^{\sigma\omega} \hhu_{\lambda\tau\omega}\vhu_\sigma \\
   &-\frac{1}{|\vhu|^2_{\ghu}}\bigl(-\ghu^{\alpha\tau}\ghu^{\beta\omega} \hhu_{\lambda\tau\omega}\vhu_\alpha\vhu_\beta+2\ghu^{\alpha\beta}\vhu_\alpha \wh_{\lambda\beta}\bigr)\ghu^{\sigma\mu}\vhu_\sigma \biggr).
  \end{split}
\end{equation}

Expressing \eqref{tconf-ford-B.1}-\eqref{tconf-ford-B.6} in Lagrangian coordinates by pulling-back those equations as scalars, and then combining the resulting equations with
\eqref{l-ev.1} and \eqref{J-ev.1} yields the following first-order Lagrangian formulation of the reduced Einstein-scalar field equations:
\begin{align}
   \Bhu^{\lambda\beta\alpha} \Jcch_\alpha^\gamma \del{\gamma}\hhu_{\beta\mu\nu}&= \Jc_0^\lambda \Bigr(\Qhu_{\mu\nu}+\frac{2(1-\Rtt)}{\tau}\bigl( \whu_{(\mu \nu)} - \Gammahu^\gamma_{\mu\nu} \vhu_\gamma\bigl)
   \Bigr), \label{tconf-ford-C.1}\\
     \Bhu^{\lambda\beta\alpha} \Jcch_\alpha^\gamma \del{\gamma}\whu_{\beta\mu},
     &=-\Jc^\lambda_0\bigl(\ghu^{\alpha \sigma}\ghu^{\beta\delta}\hhu_{\mu\sigma\delta}\whu_{\alpha\beta}+\underline{\hat{\Wtt}}{}_\mu\bigr), \label{tconf-ford-C.2}\\
      \Bhu^{\lambda\beta\alpha} \Jcch_\alpha^\gamma \del{\gamma}\zhu_{\beta} &=-\Jc^\lambda_0\underline{\hat{\Wtt}}, \label{tconf-ford-C.3} \\
       \del{0}\ghu_{\mu\nu}&= \Jc_0^\alpha \hhu_{\alpha\mu\nu}, \label{tconf-ford-C.4}\\
     \del{0} \vhu_{\mu}&= \Jc_0^\alpha \whu_{\alpha\mu},  \label{tconf-ford-C.5}\\
     \del{0} \tau&= \Jc^\alpha_0\zhu_{\alpha},  \label{tconf-ford-C.6}\\
     \del{0} \Jc^\mu_\nu &=  \Jc^\lambda_\nu \Jsc^\mu_\lambda,\label{tconf-ford-C.7}\\
    \del{0}l^\mu &= \chihu^\mu, \label{tconf-ford-C.8}\end{align}
where the fields to be solved for are 
$\{\hhu_{\beta\mu\nu},\whu_{\beta\mu},\zhu_{\beta},\ghu_{\mu\nu},\vhu_\mu,\tau,\Jc^\mu_\nu,l^\mu\}.$

\subsection{Lagrangian initial data}
The initial data for the system \eqref{tconf-ford-C.1}-\eqref{tconf-ford-C.6} is determined from the reduced conformal Einstein-scalar field initial data \eqref{gt-idata}-\eqref{dt-taut-idata} in the same way as it is described in  
\cite[\S 5.5]{BeyerOliynyk:2021}. Specifically, we have
\begin{align}
    l^\mu \bigl|_{\Sigma_{t_0}} &= \lr^\mu, \label{l-idata}\\
    \Jc^\mu_\nu \bigl|_{\Sigma_{t_0}} &= \Jcr^\mu_\nu, \label{Jc-idata}\\
    \tau\bigl|_{\Sigma_{t_0}} &= \taur, \label{tauh-idata} \\ 
    \vhu_\mu\bigl|_{\Sigma_{t_0}} &= \delta^0_\mu \taug +\delta_\mu^\Lambda \del{\Lambda} \taur, \label{vhu-idata} \\
    \whu_{\Lambda\Omega}\bigl|_{\Sigma_{t_0}} &= \del{\Lambda}\del{\Omega} \taur, \label{whu-idata-1}\\
    \whu_{0\Omega}\bigl|_{\Sigma_{t_0}} &= \del{\Omega} \taug, \label{whu-idata-2}\\
     \whu_{\Lambda 0}\bigl|_{\Sigma_{t_0}} &= \del{\Lambda} \taug, \label{whu-idata-3}\\
      \whu_{0 0}\bigl|_{\Sigma_{t_0}} &= \frac{1}{\gr^{00}}\bigl(-\gr^{0\Lambda}\del{\Lambda}\taug+
      -\gr^{\Lambda\Omega}\del{\Lambda}\del{\Omega}\taur+ \mathring{\Wtt} \bigr), \label{whu-idata-4}\\
    \zhu_\mu\bigl|_{\Sigma_{t_0}} &= \delta^0_\mu \taug +\delta_\mu^\Lambda \del{\Lambda} \taur, \label{zhu-idata} \\
    \ghu_{\mu\nu}\bigl|_{\Sigma_{t_0}} &= \gr_{\mu\nu}, \label{ghu-idata} \\
    \hhu_{\alpha\mu\nu}\bigl|_{\Sigma_{t_0}} &= \delta_\alpha^0 \grave{g}_{\mu\nu}+\delta_\alpha^\Lambda\del{\Lambda}\gr_{\mu\nu}, \label{hhu-idata}
\intertext{where}
    \lr^\mu &= \delta^\mu_0 t_0 + \delta^\mu_\Lambda x^\Lambda, \label{lr-def} \\
    \vr^\mu &= \gr^{\mu\nu}(\delta^0_\nu \taug +\delta_\nu^\Lambda\del{\Lambda}\taur),\label{vr-def}\\
    \chir^\mu &= \frac{1}{|\vr|_{\gr}^2}\vr^{\mu}, \label{chir-def}\\
    \Jcr^\mu_\nu &= \delta^0_\nu \chir^\mu + \delta^\Lambda_\nu \delta^\mu_\Lambda, \label{Jcr-def}
    \intertext{and}
    \mathring{\Wtt} &
    =  \frac{\Rtt}{\taur}\bigl(\gr^{00}\taugr^2+2\gr^{0\Omega}\taugr\del{\Omega}\taur+\gr^{\Lambda\Omega}\del{\Lambda}\taur\del{\Omega}\taur+1\bigr). \label{Wtt-ring-def}
\end{align}

\begin{rem}\label{rem-Lag-idata}
On the initial hypersurface $\Sigma_{t_0}=\{t_0\}\times \Tbb^{n-1}$,  our choice of initial data
implies that 
\begin{equation} \label{dt-tau-idata}
    \del{0}\tau\bigl|_{\Sigma_{t_0}} = \Jc^\mu_0 \zh_\mu\bigl|_{\Sigma_{t_0}} = \chir^\mu \vr_\mu =1, 
\end{equation}
\begin{equation}\label{g-idata}
 g_{\mu\nu} \bigl|_{\Sigma_{t_0}} = \Jc^\alpha_\nu \Jc^\beta_\nu \gr_{\alpha\beta} \bigl|_{\Sigma_{t_0}} 
 = \Jcr^\alpha_\mu \gr_{\alpha\beta}\Jcr^\beta_\nu,
\end{equation}
\begin{equation} \label{dt-g-idata}
    \del{0}g_{\mu\nu} \bigl|_{\Sigma_{t_0}} = 
     \Jcr^\gamma_0\Jcr^\mu_\alpha \Jcr^\beta_\nu \bigl(\delta_\gamma^0 \grave{g}_{\alpha\beta}+\delta_\gamma^\Lambda\del{\Lambda}\gr_{\alpha\beta}\bigr) +  \gr_{\alpha\beta}\bigl(
     \Jcr_\mu^\lambda \mathring{\Jsc}{}^\alpha_\lambda+\Jcr^\beta_\nu + \Jcr^\alpha_\mu \Jcr_\nu^\lambda \mathring{\Jsc}{}^\beta_\lambda\bigr)
\end{equation}
and
\begin{equation} \label{chi-idata}
    \chi^\mu \bigl|_{\Sigma_{t_0}} = \Jcch^\mu_\nu \chih^\nu \bigl|_{\Sigma_{t_0}} = (\Jcr^{-1})^\mu_\nu \chir^\nu = \delta^\mu_0,
\end{equation}
where $\mathring{\Jsc}{}^\mu_\nu= \Jsc^\mu_\nu|_{\Sigma_{t_0}}$.
\end{rem}

\subsection{Local-in-time existence}
We are now in a position to state our local-in-time existence and uniqueness results and continuation principle for the system \eqref{tconf-ford-C.1}-\eqref{tconf-ford-C.8}. Since the proof only requires a trivial modification of the proof of Proposition 5.5 from  \cite{BeyerOliynyk:2021}, we state the result without proof. 

\begin{prop} \label{lag-exist-prop}
Suppose $k>(n-1)/2+1$, $t_0>0$, and that the initial data $\taur\in H^{k+2}(\Tbb^{n-1})$, $\taugr\in H^{k+1}(\Tbb^{n-1})$, $\gr_{\mu\nu}\in H^{k+1}(\Tbb^{n-1},\Sbb{n})$ and $\ggr_{\mu\nu}\in H^{k}(\Tbb^{n-1},\Sbb{n})$ is chosen so that 
the inequalities $\det(\gr_{\mu\nu})<0$ and $|\vr|_{\gr}^2 <0$
are satisfied where $\vr^\mu$ is defined by \eqref{vr-def}.
Then there exists a $t_1<t_0$ and a unique solution
\begin{equation}
  \label{eq:Wreg}
\Wsc \in \bigcap_{j=0}^{k}C^j\bigl((t_1,t_0], H^{k-j}(\Tbb^{n-1})\bigr),
\end{equation}
where
\begin{equation}
  \label{eq:Wdef}
    \Wsc=(\hhu_{\beta\mu\nu},\whu_{\beta\nu}, \zhu_\beta,\ghu_{\mu\nu},\vhu_\mu,\tau,\Jc^\mu_\nu,\ell^\mu ),
\end{equation}
on $M_{t_1,t_0}$ to the initial value problem consisting of the evolution
equations \eqref{tconf-ford-C.1}-\eqref{tconf-ford-C.8} and
the initial conditions \eqref{l-idata}-\eqref{hhu-idata}. Moreover, the following properties hold:
\begin{enumerate}[(a)]
\item Letting $\Wsc_{\!0} = \Wsc|_{\Sigma_{t_0}} \in H^{k}(\Tbb^{n-1})$
denote the initial data, there exists a $\delta>0$ such that if $\Wsc^*_{\!0} \in  H^{k}(\Tbb^{n-1})$ satisfies $\norm{\Wsc^*-\Wsc_{\!0}}_{H^k(\Tbb^{n-1})}<\delta$,    
then there exists a unique solution 
\begin{equation*}
    \Wsc^* \in \bigcap_{j=0}^{k}C^j\bigl((t_1,t_0], H^{k-j}(\Tbb^{n-1})\bigr)
\end{equation*}
of the evolution equations \eqref{tconf-ford-C.1}-\eqref{tconf-ford-C.8} on $M_{t_1,t_0}=(t_1,t_0]\times \Tbb^{n-1}$  that agrees with the initial data $\Wsc^*_{\!0}$ on the initial hypersurface $\Sigma_{t_0}$. 
\item The relations 
\begin{equation} \label{eq:Lag-constraints}
\del{\alpha}\ghu_{\mu\nu}=\Jc^\beta_\alpha \hhu_{\beta\mu\nu},\quad  \del{\alpha}\vhu_{\mu}=\Jc^\beta_\alpha \whu_{\beta\mu}, \quad \del{\alpha}\tau-\Jc^\beta_\alpha \zhu_{\beta},\quad \vhu_\mu=\zhu_\mu \AND \Jc^\mu_\nu = \del{\nu} l^\mu    
\end{equation}
hold in $M_{t_1,t_0}$.
\item The pair $\{g_{\mu\nu}=\del{\mu}l^\alpha\ghu_{\alpha\beta}\del{\nu}l^\beta,\tau\}$
determines a solution 
of the reduced conformal Einstein-scalar field equations
\begin{equation} \label{lag-redeqns}
    -2R_{\mu\nu}+2\nabla_{(\mu} X_{\nu)}=-\frac{2(1\added{-\Rtt})}{\tau}\nabla_\mu \nabla_\nu \tau
    , \quad g^{\alpha\beta}\Dc_\alpha\Dc_\beta \tau=\Wtt,
\end{equation}
on $M_{t_1,t_0}$ that satisfies the initial conditions \eqref{tauh-idata} and \eqref{dt-tau-idata}-\eqref{dt-g-idata}, where $X^\gamma = \frac{1}{2}g^{\mu\nu}g^{\gamma\lambda}
    (2\Dc_{\mu}g_{\nu\lambda}-\Dc_\lambda g_{\mu\nu})$
and $\Dc_\mu$ is the Levi-Civita connection of the flat metric
$\gc_{\mu\nu}=\del{\mu}l^\alpha \eta_{\alpha\beta} \del{\nu}l^\beta$
on $M_{t_1,t_0}$. 
\item The scalar field $\tau$ is specified by
\begin{equation} \label{tau-synch} 
    \tau = t-t_0 + \taur  
\end{equation}
in $M_{t_1,t_0}$ while the vector field
\begin{equation} \label{chi-def}
    \chi^\mu = \frac{1}{|\nabla\tau|^2_g}\nabla^\mu\tau 
\end{equation}
satisfies 
\begin{equation} \label{Lagrangian}
\chi^\mu = \delta^\mu_0
\end{equation}
in $M_{t_1,t_0}$.
\item If the initial data $\{\gr_{\mu\nu},\ggr_{\mu\nu},\taur,\taugr\}$ also satisfies the
constraint equations \eqref{grav-constr}-\eqref{wave-constr}
on $\Sigma_{t_0}$, then the pair $\{g_{\mu\nu},\tau\}$ solves
the conformal Einstein-scalar field equations
\begin{align} 
R_{\mu\nu}=\frac{(1-\Rtt)}{\tau}\nabla_\mu \nabla_\nu \tau,\quad
\Box_{g}\tau = \Wtt,\label{lag-confeqns}
\end{align}
and satisfies the wave gauge constraint 
\begin{equation} \label{lag-wave-gauge}
    X^\gamma :=\frac{1}{2}g^{\mu\nu}g^{\gamma\lambda}(2\Dc_\mu g_{\nu\lambda}-\Dc_{\lambda}g_{\mu\nu} )=0  
\end{equation}
in $M_{t_1,t_0}$.
\item If 
  \begin{equation}
    \label{eq:cont_crit1}
\max\biggl\{\sup_{M_{t_1,t_0}}\!\det(g_{\mu\nu}), \sup_{M_{t_1,t_0}}\!|\nabla\tau|_{g}^2\biggl\} <0
\end{equation}
and
\begin{align}
  \sup_{t_1<t<t_0}\Biggl(&\norm{g_{\mu\nu}(t)}_{W^{2,\infty}(\Tbb^{n-1})}+\norm{\del{t}g_{\mu\nu}(t)}_{W^{1,\infty}(\Tbb^{n-1})} \notag \\
  &\quad
  +\norm{\Dc_\nu \chi^\lambda(t)}_{W^{2,\infty}(\Tbb^{n-1})}+\norm{\del{t}(\Dc_\nu \chi^\lambda)(t)}_{W^{1,\infty}(\Tbb^{n-1})}\Bigr)<\infty,  \label{eq:cont_crit2}
\end{align}
then there exists a $t_1^*<t_1$ such that the solution $\Wsc$ can be uniquely continued to the time interval $(t_1^*,t_0]$.
\end{enumerate}
\end{prop}

\subsection{Temporal synchronization of the singularity\label{temp-synch}}
In order to temporally synchronize big bang singularities in our solutions, we need a time coordinate whose level set at a particular time, e.g., $0$, agrees with the spacelike singular hypersurface. The temporal synchronization of big bang singularities that occur in the solutions to the Einstein-scalar field equations with vanishing potential from \cite{BeyerOliynyk:2021} is achieved by using the scalar field $\tau$ as a time coordinate. We do the same in this article and again use $\tau$ to synchronize big bang singularities that appear, this time, in solutions of the Einstein-scalar field equations with non-vanishing potentials of the type \eqref{eq:exponentialpotential} . 

It follows from Proposition  \ref{lag-exist-prop}.(d) that $\tau$ coincides with the Lagrangian time coordinate $t$ if and only if $\taur=t_0$ on the initial hypersurface  $\Sigma_{t_0}$. However, if we want our results to apply to an open set of geometric data, we cannot in general assume that $\taur $ is constant. To resolve this tension, we proceed as in   \cite{BeyerOliynyk:2021} and allow for choices of
$\tau$ that are close to constant on $\Sigma_{t_0}$, i.e.,~$\tau = t_0 + \rhor$ in $\Sigma_{t_0}$ with $\rhor$ a sufficiently small function. As shown in the proof of \cite[Prop.~5.8]{BeyerOliynyk:2021}, which is easily adapted to allow for non-vanishing potentials, it possible to  evolve $\tau$ for a short amount of time to obtain a solution $\{g_{\mu\nu},\tau\}$ of the conformal Einstein-scalar field equations on $M_{t_1,t_0}$ for some $t_1<t_0$ with $t_1$ close to $t_0$ and then find a level surface of $\tau^{-1}(t^*_0)$
for some $t^*_0\in (t_1,t_0)$ that satisfies $\tau^{-1}(t^*_0) \subset (t_1,t_0)\times \Tbb^{n-1}$ and $\tau^{-1}(t^*_0)\cong \Tbb^{n-1}$. By replacing $\Sigma_{t_0}$ with $\tau^{-1}(t^*_0)$, this yields a hypersurface $\tau^{-1}(t^*_0)\cong \Tbb^{n-1}$ on which
$\tau$ is constant. This construction is made precise in the following proposition and we omit the proof since it follows from a straightforward modification of \cite[Prop.~5.8]{BeyerOliynyk:2021}.

\begin{prop} \label{synch-prop}
Suppose $k>(n-1)/2+1$, $t_0>0$, the initial data $\taur=t_0+\rhor$, $\rhor\in H^{k+2}(\Tbb^{n-1})$, $\taugr\in H^{k+1}(\Tbb^{n-1})$, $\gr_{\mu\nu}\in H^{k+1}(\Tbb^{n-1},\Sbb{n})$ and $\ggr_{\mu\nu}\in H^{k}(\Tbb^{n-1},\Sbb{n})$ is chosen so that 
the inequalities $\det(\gr_{\mu\nu})<0$ and $|\vr|_{\gr}^2 <0$ hold
and the constraint equations \eqref{grav-constr}-\eqref{wave-constr} are satisfied, and
let $\{\gtt_{\Lambda\Gamma},\Ktt_{\Lambda\Gamma},\taur,\taugr\}$
denote the geometric initial data on $\Sigma_{t_0}=\{t_0\}\times \Tbb^{n-1}$ that is determined from the initial data $\{\gr_{\mu\nu},\ggr_{\mu\nu},\taur=t_0+\rhor,\taugr\}$ via \eqref{gtt-def1}.
Then for any $\tilde{\delta}>0$, there exists a $\delta>0$ and times $t^*_1<t^*_0<t_0$ such that if $\norm{\rhor}_{H^{k+2}(\Tbb^{n-1})}<\delta$, then:
\begin{enumerate}[(a)]
    \item The solution
$W$ from Proposition \ref{lag-exist-prop} to the initial value problem consisting of the evolution
equations \eqref{tconf-ford-C.1}-\eqref{tconf-ford-C.8} and
the initial conditions \eqref{l-idata}-\eqref{hhu-idata} exists on $M_{t_1^*,t_0}=(t_1^*,t_0]\times \Tbb^{n-1}$.
\item The pair $\{g_{\mu\nu}=\del{\mu}l^\alpha\ghu_{\alpha\beta}\del{\nu}l^\beta,\tau\}$
determined from the solution $W$ via \eqref{eq:Wdef} defines a solution to the conformal Einstein-scalar field equations \eqref{lag-confeqns}.
\item The map
\begin{equation*}
    \Psi \: : \: (t_0,t_1)\times \Tbb^{n-1}  \longrightarrow  \Rbb \times \Tbb^{n-1}\: :\: (t,x)\longmapsto (\ttl,\xt) = (t+\rhor(x),x)
\end{equation*}
defines a diffeomorphism onto its image and the push-forward
$ \{ \gt_{\mu\nu} = (\Psi_*g)_{\mu\nu},\taut=\Psi_*\tau\}$ of the solution $\{g_{\mu\nu},\tau\}$ by this map 
determines geometric initial data 
$\{\gttt_{\Lambda\Sigma},\Kttt_{\Lambda\Sigma},\mathring{\taut},\grave{\taut}\}$ on the hypersurface $\Sigma_{t^*_0}=\{t^*_0\}\times\Tbb^n$ satisfying $\mathring{\taut}=t^*_0$ and
\begin{equation*}
\norm{\gttt_{\Lambda\Sigma}-\gtt_{\Lambda\Sigma}}_{H^{k+1}(\Tbb^{n-1})}+\norm{\Kttt_{\Lambda\Sigma}-\Ktt_{\Lambda\Sigma}}_{H^{k}(\Tbb^{n-1})}+\norm{t^*_0-\taur}_{H^{k+2}(\Tbb^{n-1})}+\norm{\grave{\taut}-\taugr}_{H^{k+1}(\Tbb^n)} <\tilde{\delta}.    
\end{equation*}
\end{enumerate}
\end{prop}

\begin{rem} \label{synch-rem}
Given the geometric initial data 
$\{\gttt_{\Lambda\Sigma},\Kttt_{\Lambda\Sigma},\mathring{\taut},\grave{\taut}\}$ on $\Sigma_{t^*_0}$ from Proposition \ref{synch-prop}, we can always impose the wave gauge constraint on $\Sigma_{t^*_0}$ by choosing an appropriate choice of the free initial data\footnote{It is worthwhile noting that this choice of free initial data is, in general, different from the lapse-shift pair computed from restricting the conformal Einstein-scalar field solution $\{\gt_{\mu\nu},\taut\}$ from Proposition \ref{synch-prop} to $\Sigma_{t^*_0}$.} $\{\Nttt,\bttt_\Lambda,\dot{\Nttt},\dot{\bttt}{}_\Lambda\}$; see Remark \ref{idata-rem} for details. Because of this, we lose no generality, as far as our stability results are concerned, by assuming that the initial data \eqref{gt-idata}-\eqref{dt-tau-idata} satisfies the gravitational and wave gauge constraint equations
\eqref{grav-constr}-\eqref{wave-constr}
along with the synchronisation condition $\taur = t_0$ on the initial hypersurface $\Sigma_{t_0}$,
which as a consequence of \eqref{tau-synch} implies that
\begin{equation} \label{tau-fix}
\tau = t \quad \text{in $M_{t_0,t_1}$}.
\end{equation}
\end{rem}

\section{A Fermi-Walker transported frame \label{Fermi}}

Proposition \ref{lag-exist-prop} establishes the local-in-time existence of solutions $\{g_{\mu\nu},\tau\}$ to the conformal Einstein–scalar field system \eqref{lag-confeqns} in Lagrangian coordinates 
$(x^\mu)$, subject to both the wave-gauge constraint \eqref{lag-wave-gauge} and the slicing condition \eqref{tau-fix}. While this representation is adequate for establishing local solutions, it is poorly suited for deriving global results, particularly near $t=0$.
To address this limitation, we follow the method of \cite{BeyerOliynyk:2021} and employ an orthonormal frame formulation, originally inspired by \cite{fournodavlos2020b}.

In the following, we employ the same orthonormal frame introduced in \cite[\S6]{BeyerOliynyk:2021}. After briefly summarising the basic setup, we state only those results from that reference that are needed for our analysis here, and we refer the reader to \cite[\S6]{BeyerOliynyk:2021} for the full derivations. 

The first step in constructing the orthonormal frame is to fix the temporal frame vector $e_0=e_0^\mu \del{\mu}$ via
\begin{equation}\label{e0-def}
e_0 := (-|\chi|_g^2)^{-\frac{1}{2}}\chi \oset{\eqref{chi-def}}{=}-\beta \nabla \tau
\end{equation}
where $\beta$ is defined by
\begin{equation}\label{alpha-def}
    \beta = (-|\nabla\tau|^2_g)^{-\frac{1}{2}}.
\end{equation}
It follows from \eqref{chi-def}, \eqref{Lagrangian} and  \eqref{e0-def} that the coordinate components of $e_0$ are determined by
\begin{equation}\label{e0-mu}
e_0^\mu = \beta^{-1}\chi^\mu =\beta^{-1}\delta^\mu_0.
\end{equation}

The timelike vector field $e_0$ can be completed to a orthonormal frame
$e_i=e_i^\mu\partial_{\mu}$ by evolving the spatial frame $e_I=e_I^\mu\partial_{\mu}$ using Fermi-Walker transport:
\begin{equation} \label{Fermi-A}
\nabla_{e_0}e_J = -\frac{g(\nabla_{e_0}e_0,e_J)}{g(e_0,e_0)} e_0.
\end{equation}
By solving this transport equation starting with initial data that is initially orthonormal, we obtain an orthonormal frame $e_i$ for the conformal metric $g$ on the entire region where the solution $\{g_{\mu\nu},\tau\}$ is defined.
So, letting 
\begin{equation*}
%\label{gij-def}
g_{ij}:=g(e_i,e_j) = e_i^\mu g_{\mu\nu} e^\nu_j
\end{equation*}
denote the frame metric, we have
\begin{equation}\label{orthog}
g_{ij}= \eta_{ij}:= -\delta_i^0\delta_j^0+ \delta_i^I\delta_j^J\delta_{IJ}.
\end{equation}

From this point on, all frame indices are expressed with respect to the above defined orthonormal frame $e_i$. In particular, the frame components
$e_i^j$ of the frame vector fields $e_i$ are
\begin{equation} \label{frame-def}
e_i^j = \delta^j_i.
\end{equation}
Also, relative to this orthonormal frame, the connection coefficients $\Gamma_{i}{}^k{}_{j}$ associated to the conformal metric $g$ are determined via
\begin{equation} \label{Gamma-def}
\nabla_{e_i} e_j = \Gamma_{i}{}^k{}_{j} e_k.
\end{equation}

As shown in \cite[\S6]{BeyerOliynyk:2021}, the spatial frame coefficients $e_I^\lambda$ satisfy
\begin{equation} \label{e0I-fix}
    e_I^0 = 0
\end{equation}
and
\begin{equation} \label{[e0,eI]-C}
\del{t}e_I^\Lambda=\beta (\gamma_0{}^J{}_I - \gamma_{I}{}^J{}_0)
e^\Lambda_J,
\end{equation}
while the derivative of $\beta$ can be expressed as
\begin{equation}\label{grad-alpha}
    e_i(\beta)=\nabla_i\beta =
    -\beta^2 \delta^k_0 (\Dc_i\Dc_k\tau
    -\Cc_i{}^l{}_k \Dc_l\tau).
\end{equation}
Furthermore, we recall that the connection coefficients $\gamma_i{}^k{}_j$ associated to the $g$-orthonormal frame $e_i$ and the background metric $\gc$, i.e.,~$\Dc_{e_i} e_j = \gamma_{i}{}^k{}_{j} e_k$,
satisfy
\begin{gather}
\gamma_{0}{}^k{}_j 
=
-\beta\delta^i_0\bigl(\delta^0_jg^{lk}+\delta^l_j\delta^k_0\bigr)
(\Dc_{i}\Dc_l\tau-\Cc_i{}^p{}_l\Dc_p\tau)  -\Cc_{0}{}^k{}_j,\label{gamma-0kj}\\
  \gamma_I{}^0{}_0 = \frac{1}{2}\delta^i_I\delta_0^j\delta^k_0\Dc_i g_{jk}, \quad
   \gamma_I{}^K{}_0 =-  \delta^{Kl} \delta_I^i\delta^j_0\Dc_i g_{jl}+ \eta^{KL}\gamma_I{}^0{}_L, \label{gamma-I00-IK0}
\intertext{and}
\del{t}\gamma_I{}^k{}_J=\beta\bigl(e_I(\gamma_0{}^k{}_J)-\gamma_I{}^l{}_J\gamma_0{}^k{}_l+\gamma_0{}^l{}_J\gamma_I{}^k{}_l+(\gamma_0{}^l{}_I-\gamma_I{}^l{}_0)\gamma_l{}^k{}_J\bigr).\label{gamma-Ikj-A}
\end{gather}

\section{First-order form}
\label{sec:FirstOrderForm}

In this section, we present a frame formulation of the reduced conformal Einstein–scalar field equations in first-order form and derive a Fuchsian formulation of these equations. This Fuchsian formulation is suitable for establishing the existence of solutions up to the big bang singularity at 
$t=0$. Our presentation closely follows \cite{BeyerOliynyk:2021,beyerStabilityFLRWSolutions2023}.

\subsection{Primary fields}
The derivation of the first-order equations begins with formulating the evolution equations in terms of the primary fields
\begin{equation}\label{p-fields}
g_{ijk}=\Dc_i g_{jk}  
\AND \tau_{ij}=\Dc_i\Dc_j \tau .
\end{equation}
After obtaining these equations, we derive evolution equations for the differentiated fields
\begin{equation}\label{d-fields}
    g_{ijkl}=\Dc_ig_{jkl} \AND \tau_{ijk}=\Dc_i\tau_{jk}.
\end{equation}

To start the derivation of the first-order equations for the primary fields, we use \eqref{red-Ricci} and \eqref{tau-fix} to express
\eqref{confESFF} as
\begin{equation}\label{confESSFJ}
g^{kj}\Dc_k \Dc_j g_{lm} =-\frac{2(1\added{-\Rtt})}{t}\bigl(
\Dc_l \Dc_m \tau - \Cc_l{}^p{}_m\Dc_p \tau  \bigr)- Q_{lm},
\end{equation}
where 
\begin{equation}\label{grad-tau}
    \Dc_i\tau= -\beta^{-1}g_{ij}e^j_0 = \beta^{-1}\delta_i^0 
\end{equation}
and
\begin{equation}\label{Cc-dt}
\Cc_l{}^p{}_m\Dc_p \tau 
=-\frac{1}{2}\beta^{-1}e_0^j\bigl(\Dc_l g_{jm}
+\Dc_m g_{jl} - \Dc_j g_{lm}\bigr)
\end{equation}
as a consequence of \eqref{Ccdef},   \eqref{e0-def} and \eqref{orthog}-\eqref{frame-def}. Next, multiplying \eqref{confESSFJ} by $-e_0^i$ produces
\begin{equation} \label{for-A}
-e^i_0 g^{kj}\Dc_k \Dc_j g_{lm} =\frac{2(1\added{-\Rtt})}{t}e^i_0
\bigl(
\Dc_l \Dc_m \tau - \Cc_l{}^p{}_m\Dc_p \tau  \bigr) + e^i_0 Q_{lm}
.
\end{equation}
Noting that $-g^{ik}e^j_0\Dc_k\Dc_j g_{lm} + g^{ij}e^k_0 \Dc_k \Dc_j g_{lm} 
=0$ by \eqref{commutator},
we can add this expression to \eqref{for-A} to get
\begin{align}\label{for-B}
B^{ijk}\Dc_kg_{jlm} =&
\frac{1-\added{\Rtt}}{t}\beta^{-1}\delta^i_0\delta^j_0(g_{ljm}+g_{mjl}-g_{jlm})
+\frac{2(1-\added{\Rtt})}{t}\delta^i_0\tau_{lm}+ \delta^i_0 Q_{lm}
\end{align}
where
\begin{align}
  B^{ijk} &= -\delta^i_0 \eta^{jk}-\delta^j_0\eta^{ik}+ \eta^{ij}\delta^k_0\label{B-def}
\end{align}
and we note that \eqref{frame-def}, \eqref{orthog} and \eqref{Cc-dt} were employed in the derivation. A similar analysis, using \eqref{grad-tau}, shows that the conformal scalar field equation
\eqref{confESFI} can be written in first-order
form as 
\begin{align}\label{for-D}
  B^{ijk}\Dc_k \tau_{jl} =& -\delta^i_0 \eta^{pm}\eta^{qm}g_{lmm}  \tau_{pq}    
   \added{
     -{\delta^i_0}\Dc_l\Wtt}.
\end{align}
Together, equations \eqref{for-B} and \eqref{for-D} constitute the desired first-order formulation of the reduced conformal Einstein–scalar field equations in terms of the primary fields \eqref{p-fields}.

For the analysis that follows, it proves advantageous to interpret \eqref{for-B} and \eqref{for-D} as transport equations for $g_{jlm}$ and 
$\tau_{jl}$, respectively, and to rewrite them in the form 
\begin{align}
\del{t}g_{rlm} =&
                 \frac{1\added{-\Rtt}}{t}\delta^0_r\delta^j_0(g_{ljm}+g_{mjl}-g_{jlm})
                 +\frac{2(1\added{-\Rtt})}{t}\delta_r^0\beta
                 \tau_{lm}-\delta_{ri}B^{ijK}\beta g_{Kjlm}+ \delta_{ri}\beta Q^i_{lm}
                 \label{for-F.1}
                   \intertext{and}
                   \del{t}\tau_{rl} =&
                                       \delta_{ri}\beta J^i_l-\delta_{ri}B^{ijK}\beta\tau_{Kjl}  
   \added{
     -{\delta^0_r \beta}\Dc_l\Wtt},
 \label{for-F.2}
\end{align}
respectively, where in deriving these expressions we employed \eqref{e0-mu}, \eqref{d-fields} and
\eqref{B-def} and have set
\begin{align}
  Q^i_{lm} &= \delta^i_0 Q_{lm}+\delta^{ij}\bigl(  \gamma_0{}^p{}_j  g_{plm} +\gamma_0{}^p{}_l  g_{jpm} +\gamma_0{}^p{}_m  g_{jlp}\bigr)
             \label{Qilm-def}
\intertext{and}
             J^i_l&=-\delta^i_0\eta^{kp}\eta^{jq}g_{lpq}\tau_{kj}+\delta^{ij}\bigl(  \gamma_0{}^p{}_j  \tau_{pl} +\gamma_0{}^p{}_l  \tau_{jp}\bigr).\label{J-def}
\end{align}

\begin{rem}
An important point for the analysis that follows is that we do not use equation \eqref{for-F.1} with 
$l=m=0$ in any subsequent arguments. This is because the wave
gauge condition \eqref{wave-gauge} can be used to determine $g_{i00}$ in terms of the other metric variables $g_{ilM}$. To see this, we use \eqref{Xdef}, \eqref{orthog}, and \eqref{p-fields} to rewrite the wave gauge constraint \eqref{wave-gauge} as
\begin{equation*}
2X_0 = -g_{000}+\delta^{JK}(2g_{JK0}-g_{0JK})=0 \AND 2 X_I = -(2g_{00I}-g_{I00})+\delta^{JK}(2g_{JKI}-g_{IJK})=0.
\end{equation*}
Rearranging and using the symmetry $g_{IJ0}=g_{I0J}$, this becomes
\begin{equation} \label{gi00}
g_{000}= -\delta^{JK}(g_{0JK}-2g_{J0K}) \AND g_{I00}= 2g_{00I}-\delta^{JK}(2g_{JKI}-g_{IJK}),
\end{equation}
which verifies the claim.
\end{rem}

Separating \eqref{for-F.1} into the components $(r,l,m)=(0,0,M)$ and $(r,l,m)=(R,0,M)$, $(r,l,m)=(0,L,M)$ and
$(r,l,m)=(R,L,M)$, we find,
with the help of \eqref{gi00}, that 
\begin{align}
  \del{t}g_{00M} &= \frac{1\added{-\Rtt}}{t}\bigl(2g_{00M}-\delta^{IJ}(2g_{IJM}-g_{MIJ})\bigr)
                   +\frac{2(1\added{-\Rtt})}{t}\beta\tau_{0M} - B^{0jK}\beta g_{Kj0M}
+\beta Q^0_{0M}
, \label{for-G.1}\\
\del{t}g_{R0M} &= -\delta_{RI}\beta B^{IjK}g_{Kj0M}+\delta_{RI}\beta Q^I_{0M}, \label{for-G.2}\\
  \del{t}g_{0LM} &= \frac{1\added{-\Rtt}}{t}(g_{L0M}+g_{M0L}-g_{0LM})
                   +\frac{2(1\added{-\Rtt})}{t}\beta\tau_{LM}- B^{0jK}\beta g_{KjLM}
                   +\beta Q^0_{LM}
  \label{for-G.3}
\intertext{and}
\del{t}g_{RLM} &= -\delta_{RI}\beta B^{IjK}g_{KjLM}+ \delta_{RI}\beta Q^I_{LM}. \label{for-G.4}
\end{align}
For the analysis, the evolution equation \eqref{for-G.3} for $g_{0LM}$ is problematic, and so, we instead consider $g_{0LM}-g_{L0M}-g_{M0L}$, which satisfies
\begin{align}
\del{t}(g_{0LM}-g_{L0M}-g_{M0L}) =&  -\frac{1\added{-\Rtt}}{t}(g_{0LM}-g_{L0M}-g_{M0L}) +\frac{2(1\added{-\Rtt})}{t}\beta\tau_{LM}+S_{LM}
\label{for-H}
\end{align}
where
\begin{align}
    S_{LM}=&-\beta B^{0jK}g_{KjLM}
+\beta Q^0_{LM} +\delta_{MI}\beta B^{IjK}g_{Kj0L} \notag \\
&-\delta_{MI}\beta Q^I_{0L} 
+\delta_{LI}\beta B^{IjK}g_{Kj0M}-\delta_{LI}\beta Q^I_{0M}.\label{S-def}
\end{align}

\subsection{Differentiated fields} Applying $\Dc_q$ to \eqref{confESSFJ}, we find, with the
help of \eqref{commutator}, \eqref{tau-fix}, \eqref{frame-def}, \eqref{orthog}, \eqref{p-fields}-\eqref{d-fields} and \eqref{grad-tau}-\eqref{Cc-dt}, that
\begin{align}
   g^{kj}\Dc_k \Dc_q \Dc_j g_{lm}
   =&-\frac{1\added{-\Rtt}}{t}\beta^{-1} e_0^j (\Dc_q\Dc_l g_{jm}
      +\Dc_q\Dc_m g_{jl}-\Dc_q\Dc_j g_{lm}) \notag \\
      &-\frac{2(1\added{-\Rtt})}{t}\Dc_q\Dc_l\Dc_m\tau -P_{qlm}
      ,\label{confESSFL}
\end{align}
where
\begin{align}
P_{qlm}=&  \frac{1\added{-\Rtt}}{t}\bigl(\Dc_q(\beta^{-1}) \delta_0^j+\beta^{-1}  \gamma_q{}^j{}_0\bigr)  (g_{ljm}+g_{mjl}-
g_{jlm}) - \eta^{kr}\eta^{js} g_{qrs}g_{kjlm} \notag\\
 &+ \Dc_q Q_{lm}
+(1\added{-\Rtt})\biggl(- \frac{1}{t^2}\beta^{-1} \delta_0^j (g_{ljm}+g_{mjl}-
g_{jlm})-\frac{2}{t^2}
 \tau_{lm}\biggr)\beta^{-1}\delta_q^0 \label{P-def}
\end{align}
and in deriving \eqref{P-def} we have used
 $\Dc_q e^j_0=\Dc_q \delta^j_0 = \gamma_q{}^j{}_0$.
We further observe by  \eqref{grad-alpha}, \eqref{Cc-dt} and
\eqref{gi00} that
\begin{equation}\label{grad-alpha-A}
\Dc_q\beta 
= e_q(\beta)=-\delta_q^0\Bigl(\beta^{2}\tau_{00}-\frac{1}{2}\beta \delta^{JK}(g_{0JK}-2g_{J0K}) \Bigr)
-\delta_q^P\Bigl(\beta^{2}\tau_{P0}+\frac{1}{2}\beta\bigl(2g_{00P}-\delta^{JK}(2g_{JKP}-g_{PJK})\bigr)\Bigr).
\end{equation}
Noting that
$-g^{ik}e^j_0\Dc_k\Dc_q\Dc_j g_{lm} + g^{ij}e^k_0 \Dc_k \Dc_q \Dc_j g_{lm} =0$
by \eqref{commutator}, we add this to \eqref{confESSFL} to obtain
\begin{align}
  B^{ijk}\beta\Dc_k g_{qjlm}
  =& \frac{1\added{-\Rtt}}{t} \delta_0^i\delta_0^j  (g_{qljm}+ g_{qmjl}-g_{qjlm})
     +\frac{2(1\added{-\Rtt})}{t}\delta^i_0\beta\tau_{qlm}+ \delta^i_0 \beta P_{qlm}.\label{for-I}
\end{align}
Similarly, applying $\Dc_q$ to \eqref{confESFI} gives
\begin{align} 
  B^{ijk}\beta\Dc_k\tau_{qjl}
  =&  \beta K^i_{ql} \added{     -\delta_0^i\beta\Dc_q\Dc_l\Wtt}\label{for-J}
\end{align}
where
\begin{equation} \label{K-def}
K^i_{ql}= \delta^{i}_0\bigl(
-\eta^{jr}\eta^{ks}g_{qrs}\tau_{kjl}
-\bigl(\eta^{jr}\eta^{ks}g_{qlrs}-2\eta^{jm}\eta^{rn}\eta^{ks}g_{qmn}
g_{lrs}\bigr)\tau_{jk} - \eta^{jr}\eta^{ks}g_{lrs}\tau_{qjk}\bigr)
\end{equation}
and in deriving this we have again used \eqref{grad-tau}.
Together, equations \eqref{for-I} and \eqref{for-J} determine a first-order system of evolution equations for the differentiated fields \eqref{d-fields}.

\begin{rem}
In the arguments that follow, we need only consider the $q=Q$ components from
\eqref{for-I} and \eqref{for-J}. This is because we can obtain $g_{0jlm}$ and
$\tau_{0jl}$ from the equations \eqref{for-B} and \eqref{for-D}; specifically,
\begin{align}
  g_{0rlm} =& \frac{1\added{-\Rtt}}{t}\beta^{-1}\delta_r^0\delta_0^j(g_{ljm}+g_{mjl}-g_{jlm})
             +\frac{2(1\added{-\Rtt})}{t}\delta_r^0\tau_{lm} - \delta_{ri}B^{ijK}g_{Kjlm}
             +\delta_r^0 Q_{lm}
             \label{for-K.1}
\intertext{and}
             \tau_{0rl} =&
                          -\delta_{ri}B^{ijK}\tau_{Kjl}
                           -\delta_r^0 \eta^{jp}\eta^{kq} g_{lpq} \tau_{kj}\added{
     -\delta^0_r\Dc_l\Wtt}. \label{for-K.2}
\end{align}
\end{rem}

\subsection{Frame and connection coefficient transport equations}
To close the system, we require evolution equations for the frame and connection coefficients. As their derivation follows the same steps as in \cite[\S7.3]{BeyerOliynyk:2021}, we omit the details and list the resulting equations below:
\begin{align}
\del{t}\beta 
&=-\beta^{3}\tau_{00}+\frac{1}{2}\beta^2 \delta^{JK}(g_{0JK}-2g_{J0K}), \label{for-L}\\
\del{t}e_I^\Lambda=&-\beta \Bigl(\frac{1}{2}\delta^{JK}(g_{0IK}-g_{I0K}-g_{K0I}) + \delta^{JK}\gamma_{I}{}^0{}_K \Bigr)
e^\Lambda_J, \label{for-M.1} \\
\del{t}\gamma_I{}^k{}_J=&
-\beta\Bigl(\delta^k_0\Bigl(\beta e_I(\tau_{0J})+\frac{1}{2}e_I(g_{J00})\Bigr)+\frac{1}{2}\eta^{kl}\bigl(e_I(g_{0Jl})+e_I(g_{J0l})-e_I(g_{l0J})\bigr)\Bigr)+L_I{}^k{}_J, \label{for-M.2}
\end{align}
where
\begin{align}
    L_I{}^k{}_J =& \beta\Bigl(\delta^k_0\Bigl(\beta^2 \tau_{I0}+\frac{1}{2}\beta \bigl(2g_{00I}-\delta^{KL}(2g_{KLI}-g_{IKL})\bigr)\Bigr)\tau_{0J} -\gamma_I{}^l{}_J\gamma_0{}^k{}_l+\gamma_0{}^l{}_J\gamma_I{}^k{}_l+ (\gamma_0{}^l{}_I-\gamma_I{}^l{}_0)\gamma_l{}^k{}_J\Bigr). \label{L-def}
\end{align}
We further recall from \cite[\S7.3]{BeyerOliynyk:2021} the following formulas for the background connection coefficients:
\begin{align}
 \gamma{}_0{}^k{}_j
    &=-(\delta^0_j\eta^{kl}+\delta^l_j\delta^k_0)\Bigl(\beta\tau_{0l}+\frac{1}{2}g_{l00} \Bigr)-\frac{1}{2}\eta^{kl}(g_{0jl}+g_{j0l}-g_{l0j}), \label{gamma-0kj-A}\\
     \gamma{}_0{}^0{}_0 &=-\frac{1}{2}g^{PQ}(g_{0PQ}-2g_{P0Q}),\label{gamma-000}\\
     \gamma{}_0{}^K{}_0 &= -\delta^{KL}(\beta \tau_{0L}+g_{00L}), \label{gamma-0K0}\\
     \gamma_0{}^0{}_J&=-\beta \tau_{0J}, \label{gamma-00J}\\
      \gamma_0{}^K{}_J&= -\frac{1}{2}\delta^{KL}(g_{0JL}+g_{J0L}-g_{L0J}), \label{gamma-0KJ}\\
    \gamma_I{}^0{}_0 &= g_{00I}-\frac{1}{2}\delta^{JK}(2 g_{JKI}-g_{IJK}), \label{gamma-I00}\\
    \gamma_I{}^J{}_0 &= -\delta^{JK}g_{I0K} + \delta^{JK}\gamma_{I}{}^0{}_K. \label{gamma-IJ0}
\end{align}

\section{Nonlinear decompositions}
\label{sec:NonlinDecomp}
In this section, we analyze the nonlinear terms appearing in the first-order equations derived in the previous section. To facilitate the analysis, we follow \cite[\S8]{BeyerOliynyk:2021} and introduce the following definitions:
\begin{align}
\kt&=(\kt_{IJ}) :=(g_{0IJ}-g_{I0J}-g_{J0I}),
\label{kt-def} \\
\ellt&=(\ellt_{IjK}) := (g_{IjK}), &&  \label{ellt-def}\\
  \mt&=(\mt_{I})  := (g_{00M}), \label{mt-def}\\
  \tau &= (\tau_{iJ},\tau_{Ji}), \label{tau-def}\\
\gt& 
=(g_{Ijkl}), \label{gt-def}\\    
\taut&
=(\tau_{Ijk}), \label{taut-def}\\
\intertext{and}
\psit&=(\psit_I{}^k{}_J):=(\gamma_I{}^k{}_J).   \label{psit-def}
\end{align}
It is worth noting that our definition of $\tau=(\tau_{iJ},\tau_{Ji})$ in \eqref{tau-def} differs from the definition of $\tau=(\tau_{ij})$ given in \cite{BeyerOliynyk:2021}. Unlike the analysis carried out there, the refined analysis developed here necessitates a careful distinction between the component $\tau_{00}$ and the collection $\tau=(\tau_{iJ},\tau_{Ji})$.

\subsection{$*$-notation\label{*-not}}
As in \cite{BeyerOliynyk:2021}, we adopt the following $*$-notation to denote multilinear maps for which it is not necessary
for subsequent arguments to know the exact values of their constant coefficients. For example,
$\kt*\ellt$ can be used to denote tensor fields of the form
\begin{align*}
[\kt*\ellt]_{ij}= C_{ij}^{KLMpQ}\kt_{KL}\ellt_{MpQ}
\end{align*}
where the coefficients $C_{ij}^{KLMpQ}$ are constants. We also use the notation
\begin{equation*}
    (1+\beta)\kt*\ellt= \kt*\ellt +\beta (\kt*\ellt)
\end{equation*}
where on the right hand side the two $\kt*\ellt$ terms correspond, in general, to distinct bilinear maps, e.g.,
\begin{equation*}
    [(1+\beta)\kt*\ellt]_{ij}:=  C_{ij}^{KLMpQ}\kt_{KL}\ellt_{MpQ}+\beta  \tilde{C}_{ij}^{KLMpQ}\kt_{KL}\ellt_{MpQ}.
\end{equation*}
More generally, the $*$ functions as a product that distributes over sums of terms where terms like
$\lambda \kt*\ellt$, $\lambda \in \Rbb\setminus\{0\}$, and $\ellt *\kt$ are identified. For example,
\begin{equation*}
    \ellt*(\kt+\mt)= (\kt+\mt)*\ellt= \ellt*\kt+\ellt*\mt.
\end{equation*}
With this notation, we would have
\begin{equation*}
    (\mt+\psit)*(\ellt+\kt) = \mt*\ellt + \mt*\kt + \psit*\ellt + \psit*\kt 
\AND
    (\kt+\ellt)*(\kt+\ellt)= \kt*\kt + \kt*\ellt + \ellt*\ellt.
\end{equation*}

\subsection{$Q_{ij}$ and $Q^k_{ij}$}

\begin{lem} \label{Qij-lem}
\begin{gather}
Q_{00}=  -\frac{3}{2}\bigl(\delta^{RS}\kt_{RS}\bigr)^2 + \frac{1}{2} \delta^{PQ}\delta^{RS}\kt_{PR} \kt_{QS}+
\Qf, \quad
Q_{0M}= \Qf_M \AND
Q_{LM}= \delta^{RS}\kt_{LR} \kt_{MS}+
\Qf_{LM}\label{Qij-lem.1}
\end{gather}
where  $\Qf = (\kt+\ellt+\mt)*(\ellt+\mt)$.
\end{lem}
\noindent The above lemma is identical to \cite[Lemma~8.1]{BeyerOliynyk:2021}; we refer the reader to that reference for the proof. On the other hand, although the lemmas presented below may appear similar or even identical to those in \cite[\S8]{BeyerOliynyk:2021}, they in fact constitute a refinement of those results. This refinement arises from the definition of $\tau$ used here, which does not involve the $\tau_{00}$ component. Since these refined lemmas play a crucial role in the analysis carried out below, we include complete proofs for each of them. 

\begin{lem}\label{Qijk-lem}
\begin{gather}
Q^0_{0M}= \Qft_M, \quad Q^I_{0M}=\Qft^I_M, \quad
Q^I_{LM}=  \Qft^I_{LM}
\AND
Q^0_{LM}= -\frac{1}{2}\delta^{RS}\kt_{RS}\kt_{LM}+ \Qft^0_{LM} \label{Qijk-lem.1}
\end{gather}
where
\begin{equation}
  \label{eq:defQft}
  \Qft =\beta(\kt+\ellt+\mt)*\tau +(\kt+\ellt+\mt)*(\ellt+\mt).
\end{equation}
\end{lem}
\noindent 
\begin{proof}
Following the proof of \cite[Lemma~8.2]{BeyerOliynyk:2021}, we define
\begin{equation*}% \label{Psi-def}
\Psi^{i}_{lm}= \delta^{ij}\bigl(\gamma_0{}^p{}_j g_{plm}+\gamma_0{}^p{}_l g_{jpm} +\gamma_0{}^p{}_m g_{jlp}\bigr)
\end{equation*}
and set $(i,l,m)=(0,0,M)$,  $(i,l,m)=(I,0,M)$,  $(i,l,m)=(0,L,M)$ and  $(i,l,m)=(I,L,M)$ to get
\begin{align*} 
\Psi^{0}_{0M}=& 2\gamma_0{}^0{}_0 g_{00M}+\gamma_0{}^P{}_0( g_{P0M}+ g_{0PM}) -\gamma_0{}^0{}_M\delta^{RS}(g_{0RS}-2g_{R0S})+\gamma_0{}^P{}_M g_{00P}, %\label{Psi-00M}
\\
\Psi^{I}_{0M}=& \delta^{IJ}\bigl(\gamma_0{}^0{}_J g_{00M}+\gamma_0{}^P{}_J g_{P0M}+\gamma_0{}^0{}_0 g_{J0M} +\gamma_0{}^P{}_0 g_{JPM} \notag \\
&\qquad +\gamma_0{}^0{}_M (2g_{00J}-\delta^{RS}(2g_{RSJ}-g_{JRS}))+\gamma_0{}^{P}{}_M g_{J0P}\bigr), %\label{Psi-I0M}
\\
\Psi^{I}_{LM}=& \delta^{IJ}\bigl(\gamma_0{}^0{}_J g_{0LM}+\gamma_0{}^P{}_J g_{PLM}+\gamma_0{}^0{}_L g_{J0M} +\gamma_0{}^P{}_L g_{JPM} +\gamma_0{}^0{}_M g_{JL0}+\gamma_0{}^{P}{}_M g_{JLP}\bigr) 
%\label{Psi-ILM}
\intertext{and}
\Psi^{0}_{LM}
=& -\frac{1}{2}\delta^{PQ}g_{0PQ}g_{0LM}-\delta^{PQ}g_{0LQ}g_{0MP}
+\Bigl[\delta^{PQ}g_{P0Q}g_{0LM}-\frac{1}{2}\delta^{PQ}(g_{L0Q}g_{0PM}-g_{Q0L}g_{0PM}) \notag \\
&-\frac{1}{2}\delta^{PQ}(g_{M0Q}g_{0LP}-g_{Q0M}g_{0LP})+
\gamma_0{}^P{}_0 g_{PLM}+\gamma_0{}^0{}_L g_{00M} +\gamma_0{}^0{}_M g_{0L0}\Bigr],
%\label{Psi-0LM}
\end{align*}
where in deriving these expressions we have employed \eqref{gi00}.
We also observe from \eqref{gi00}, \eqref{gamma-0kj-A}, \eqref{ellt-def}, \eqref{mt-def},
that
\begin{align}
  \label{eq:gammadecomp1}
  \gamma{}_0{}^k{}_j
  =&-(\delta^0_j\eta^{kL}-\delta^L_j\eta^{k0})\Bigl(\beta \tau_{0L}+ \mt_L-\frac{1}{2}\delta^{JK}(2\ellt_{JKL}-\ellt_{LJK}) \Bigr)\\
  &-\frac{1}{2}\eta^{k0}\delta_j^J (2\mt_{J}-\delta^{PQ}(2\ellt_{PQJ}-\ellt_{JPQ}))
  -\frac{1}{2}\eta^{kL}\delta_j^0 \delta^{PQ}(2\ellt_{PQL}-\ellt_{LPQ})\notag\\
  &+\frac{1}{2}\eta^{k0}\delta_j^0 \delta^{PQ}\kt_{PQ}  
  -\frac{1}{2}\eta^{kL}\delta_j^J(\kt_{JL}+2\ellt_{J0L}).\notag
\end{align}
From the above expansions and the definitions
\eqref{kt-def}-\eqref{mt-def}, we deduce that
\begin{align}
\Psi^0_{0M}&= \Qft_M,\label{Qijk-lem.5} \\
\Psi^I_{0M}&= \Qft^I_{0M}, \label{Qijk-lem.6}\\
\Psi^0_{LM}&=  -\frac{1}{2}\delta^{PQ}\kt_{PQ}\kt_{LM}-\delta^{PQ}\kt_{LQ}\kt_{MP}+\Qft_{LM} \label{Qijk-lem.7}
\intertext{and}
\Psi^I_{LM}&=  \Qft^I_{LM}, \label{Qijk-lem.8}
\end{align}
where $\Qft$ is as defined in the statement of the lemma. By \eqref{Qilm-def}
$Q^i_{lm}=\Psi^i_{lm}+\delta^i_0 Q_{lm}$ and the stated expansions follow directly from \eqref{Qijk-lem.1} and \eqref{Qijk-lem.5}–\eqref{Qijk-lem.8}.
\end{proof}

\subsection{$e_I(\tau_{jk})$, $e_I(g_{jkl})$ and $L_I{}^k{}_j$}

\begin{lem} \label{lem-cov-exp}
\begin{align}
     e_I(\tau_{jk})&=\tau_{Ijk} + \tf_{Ijk}, \label{lem-cov-exp.1}\\
     e_I(g_{jkl})&= g_{Ijkl}+\gf_{Ijkl},
     \label{lem-cov-exp.2}
     \intertext{and}
     L_I{}^k{}_J&=\beta\Lf_I{}^k{}_J\label{lem-cov-exp.3}
\end{align}
where
\begin{align*} %\label{gf-def}
  \tf= (\psit+\ellt+\mt)*(\tau_{00}+\tau),  
  \quad
  \gf= (\psit+\ellt+\mt)*(\kt+\ellt +\mt), 
\end{align*}
and
\begin{align} \label{Lf-def}
  \Lf&=  (\beta\tau+\psit+\kt+\ellt +\mt)*(\beta \tau+\psit+\ellt+\mt).
\end{align}
\end{lem}
\begin{proof}
By the definition of the covariant derivative $\Dc_j$ associated to the background metric $\gc_{ij}$, we have
\begin{align*}
    e_I(\tau_{jk})=&\delta_I^i\bigl(\Dc_i\tau_{jk}+
    \gamma_i{}^m{}_j\tau_{mk}+\gamma_i{}^m{}_k\tau_{jm}\bigr)
    = \tau_{Ijk}+
                     \gamma_I{}^m{}_j\tau_{mk}+\gamma_I{}^m{}_k\tau_{jm}\\
     &= \tau_{Ijk}+
       2\gamma_I{}^0{}_{(j}\delta^0_{k)}\tau_{00}
       +
       2(\gamma_I{}^0{}_{(j}\delta^P_{k)}+\gamma_I{}^P{}_{(j}\delta^0_{k)})\tau_{0P}       
       +
       2\gamma_I{}^M{}_{(j}\delta^P_{k)}\tau_{MP}
    %\label{eI(taujk)}
    \intertext{and}
     e_I(g_{jkl})=&\delta_I^i\bigl(\Dc_ig_{jkl}+
    \gamma_i{}^m{}_jg_{mkl}+\gamma_i{}^m{}_kg_{jml}+\gamma_i{}^m{}_l g_{jkm}\bigr)
    = g_{Ijkl}+
    \gamma_I{}^m{}_jg_{mkl}+\gamma_I{}^m{}_kg_{jml}+\gamma_I{}^m{}_l g_{jkm}. %\label{eI(gjkl)}
\end{align*}
Noting from \eqref{gi00}, \eqref{gamma-000}-\eqref{gamma-0KJ},
\eqref{kt-def}-\eqref{mt-def} and \eqref{psit-def} that
\begin{equation}
  \label{eq:gammadecomp2}
  \gamma_I{}^m{}_j=\delta_j^0 \delta^m_0\bigl(\mt_{I}-\frac{1}{2}\delta^{JK}(2 \ellt_{JKI}-\ellt_{IJK})\bigr)+\delta_j^0 \delta^m_M\delta^{MK}(-\ellt_{I0K} + \psit_{I}{}^0{}_K)
  +\delta_j^J\psit_I{}^m{}_J,
\end{equation}
it is clear that the expansions \eqref{lem-cov-exp.1}-\eqref{lem-cov-exp.2} follow directly from the above formulas for $e_I(\tau_{jk})$ and $e_I(g_{jkl})$. 

To complete the proof, we observe that 
\begin{align*}
    -\gamma_I{}^l{}_J\gamma_0^k{}_l + \gamma_0{}^l{}_J\gamma_I{}^k{}_l+ (\gamma_0{}^l{}_I-\gamma_I{}^l{}_0)\gamma_l{}^k{}_J
    &=  -\gamma_I{}^0{}_J\gamma_0{}^k{}_0-\gamma_I{}^L{}_J\gamma_0{}^k{}_L + \gamma_0{}^0{}_J\gamma_I{}^k{}_0+ \gamma_0{}^L{}_J\gamma_I{}^k{}_L \\
    & \qquad +(\gamma_0{}^0{}_I-\gamma_I{}^0{}_0)\gamma_0{}^k{}_J+(\gamma_0{}^L{}_I-\gamma_I{}^L{}_0)\gamma_L{}^k{}_J
\end{align*}
can be expressed, with the help of \eqref{kt-def}-\eqref{mt-def}, \eqref{psit-def},  \eqref{eq:gammadecomp1}, \eqref{eq:gammadecomp2} and \eqref{gamma-000}-\eqref{gamma-IJ0},
as 
\begin{gather*}
  -\beta \gamma_I{}^l{}_J\gamma_0{}^k{}_l + \beta \gamma_0{}^l{}_J\gamma_I{}^k{}_l + \beta(\gamma_0{}^l{}_I-\gamma_I{}^l{}_0)\gamma_l{}^k{}_J =\beta(\beta\tau+\psit+\kt+\ellt +\mt)*(\beta \tau+\psit+\ellt+\mt). 
\end{gather*}
Using this, it then straightforward to verify
that \eqref{lem-cov-exp.3} and \eqref{Lf-def} follows from \eqref{L-def} and \eqref{kt-def}-\eqref{mt-def}.
\end{proof}

For use below, we observe from \eqref{lem-cov-exp.1}-\eqref{lem-cov-exp.2} that the transport equation
\eqref{for-M.2} can be expressed
as 
\begin{align}
    \del{t}\gamma_I{}^k{}_J=
-\delta^k_0\Bigl(\beta^2 \tau_{I0J}+\frac{1}{2}\beta g_{IJ00}\Bigr)-\frac{1}{2}\eta^{kl}\bigl(\beta g_{I0Jl}+\beta g_{IJ0l}-\beta g_{Il0J}\bigr)+\beta \Lf_I{}^k{}_J\label{for-N},
\end{align}
where $\Lf$ is of the form \eqref{Lf-def}, and from \eqref{kt-def} and \eqref{psit-def} that the transport equations
\eqref{for-L} and \eqref{for-M.1} can be expressed as
\begin{align}
\del{t}\beta &= -\beta^3 \tau_{00}+\frac{1}{2} \delta^{JK}\beta^2\kt_{JK} \label{for-O.1}
\intertext{and}
\del{t}e^\Lambda_I &= - \Bigl(\frac{1}{2}\delta^{JL}\beta\kt_{IL}+\delta^{JK}\beta\psit_{I}{}^0{}_K\Bigr)e^\Lambda_J, \label{for-O.2}
\end{align}
respectively. We further observe from \eqref{grad-alpha-A} and
\eqref{ellt-def}-\eqref{mt-def} that
\begin{equation} \label{grad-alpha-B}
    e_I(\beta)= -\beta^2 \tau_{I0}-\frac{1}{2}\beta\bigl( 2 \mt_I -\delta^{JK}(2\ellt_{JKI}-\ellt_{IJK})\bigr).
\end{equation}

\subsection{$P_{Qlm}$, $\delta^q_Q \beta \Dc_k g_{qjlm}$ and $\delta^q_Q \beta \Dc_k \tau_{qjl}$} 

\begin{lem}\label{lem-P-exp}
\begin{align}
   \beta P_{Qlm} &= \Pf_{Qlm},\label{lem-P-exp.1} \\
  \delta^q_Q \beta \Dc_k g_{qjlm} &=  \delta_k^0\del{t}g_{Qjlm} +\delta_k^K\beta e_K(g_{Qjlm})+ \Gf_{Qkjlm}, \label{lem-P-exp.2}\\
   \intertext{and}
   \delta^q_Q \beta \Dc_k \tau_{qjl}&=
                                       \delta_k^0\del{t}\tau_{Qjl} +\delta_k^K\beta e_K(\tau_{Qjl})+ \Tf_{Qkjl}\label{lem-P-exp.3}\\
  &\added{+\beta\delta_j^0(\delta_k^K\psit_K{}^0{}_Q-\beta\delta_k^0\tau_{0Q})\Dc_l \Wtt}, \notag
\end{align}
where
\begin{align}
\Pf =& \frac 1t \beta\tau_{00}*(\ellt+\mt)+\frac{1}{t}(\beta\tau + \psit + \ellt +\mt)*(\kt+\ellt+\mt) + \beta(\kt+\ellt+\mt) *\bigl( \gt+(\ellt+\mt)*(\kt+\ellt+\mt)\bigr), \label{Pf-def}\\
    \Gf=&(\beta \tau +\psit) *\Bigl( \frac{1}{t}\bigl(\beta \tau_{00}+\beta\tau+\kt+\ellt+\mt\bigr)  
     +\beta(\kt+\ellt+\mt)*(\kt+\ellt+\mt)\Bigr) \notag\\
     &\qquad
       +\beta \bigl(\beta \tau+\psit+\kt+\ellt+\mt\bigr)*\gt
    \label{Gf-def}
    \intertext{and}
    \Tf=&\beta (\beta \tau+\psit+ \kt+ \ellt +\mt)*\taut
          +\beta (\beta \tau+\psit +\ellt +\mt)*(\kt + \ellt +\mt)*(\tau_{00}+\tau)\label{Tf-def}. 
\end{align}
\end{lem}
\begin{proof}
With the help of \eqref{orthog} and \eqref{p-fields}-\eqref{d-fields}, we observe from 
differentiating \eqref{Q-def} that
\begin{equation} \label{Dq-Qij}
    \Dc_q Q_{ij} =\Qtt^1_{qij}+ \Qtt^2_{qij}
\end{equation}
where 
\begin{align*}
\Qtt^{1}_{qij} &= \frac{1}{2}\eta^{kl}\eta^{mn}\Bigl(g_{qimk} g_{jn l}+g_{imk}g_{qjnl}
+2 g_{qnil}g_{kjm} +2g_{nil}g_{qkjm}\\ &\qquad- 2g_{qlin} g_{kjm}- 2g_{lin}g_{qkjm}
-2g_{qlin}g_{jmk}-2g_{lin}g_{qjmk} -2 g_{qimk}g_{ljn}-2g_{imk}g_{qljn}\Bigr)
\intertext{and}
    \Qtt^2_{qij} &=  -\frac{1}{2}(\eta^{kr}\eta^{ls}g^{mn}+ \eta^{kl} \eta^{mr}\eta^{ns}\bigr)g_{qrs}\Bigl(g_{imk}g_{jn l}
+2 g_{nil}g_{kjm} - 2g_{lin}g_{kjm}
-2g_{lin}g_{jmk} -2g_{imk}g_{ljn}\Bigr).
\end{align*}
Setting $q=Q$ in \eqref{Dq-Qij} yields
\begin{equation} \label{Dq-exp}
    \delta^q_Q\Dc_q Q_{ij} =\bigl[(\kt+\ellt +\mt)*\bigl( \gt+(\ellt+\mt)*(\kt+\ellt +\mt) \bigr)\bigr]_{Qij},
\end{equation}
where in deriving this we have employed  \eqref{gi00} and \eqref{kt-def}-\eqref{gt-def}.

Next, we set
\begin{equation*}
    \Ptt_{Qlm} =
 \frac{1\added{-\Rtt}}{t}\bigl(e_Q(\beta^{-1}) \delta_0^j+\beta^{-1}  \gamma_Q{}^j{}_0\bigr)  (g_{ljm}+g_{mjl}-
g_{jlm}) -\eta^{kr}\eta^{js}g_{Qrs}g_{kjlm},
\end{equation*}
where we note from \eqref{P-def} that $P_{Qlm}=\Ptt_{Qlm} +\Dc_Q Q_{lm}$. Using \eqref{gi00} and \eqref{ellt-def}-\eqref{mt-def}, we observe that $\Ptt_{Qlm}$ can be expressed as
\begin{align}
    \Ptt_{Qlm} &=
 \frac{1\added{-\Rtt}}{t}\bigl(e_Q(\beta^{-1})+\beta^{-1}  \gamma_Q{}^0{}_0\bigr)  (g_{l0m}+g_{m0l}-
g_{j0m})
 +\frac{1\added{-\Rtt}}{t}\beta^{-1}  \gamma_Q{}^J{}_0 (g_{lJm}+g_{mJl}-
g_{Jlm})\notag \\
&\qquad -\bigl(2\mt_Q-\delta^{RS}(2\ellt_{RSQ}-\ellt_{QRS})\bigr)g_{00lm}+\eta^{JS} \ellt_{Q0S}g_{0Jlm}
- \delta^{KR}\eta^{js} \ellt_{QsR}g_{Kjlm}.  \label{Ptt-def}
\end{align}
Furthermore, by \eqref{gi00}, \eqref{gamma-I00}, \eqref{gamma-IJ0}, 
\eqref{kt-def}-\eqref{mt-def}, \eqref{psit-def} and  
\eqref{grad-alpha-B}, we have
\begin{align}
\frac{1\added{-\Rtt}}{t}\bigl(e_Q(\beta^{-1})+\beta^{-1}  \gamma_Q{}^0{}_0\bigr)  (g_{l0m}&+g_{m0l}-
g_{j0m})
+\frac{1\added{-\Rtt}}{t}\beta^{-1}  \gamma_Q{}^J{}_0 (g_{lJm}+g_{mJl}-
g_{Jlm})\notag \\
&= \frac{1}{t}\beta^{-1}[(\beta \tau+\psit+\ellt+\mt)*(\kt+\ellt+\mt)]_{Qlm}, 
    \label{Ptt-exp-A}
\end{align}
while we see from \eqref{for-K.1} and \eqref{Qij-lem.1} that
\begin{align} \label{g0rlm-exp}
  g_{0rlm} = &\frac{1}{t}\beta^{-1}\delta_r^0(g_{l0m}+g_{m0l}-g_{0lm})
               +\frac{2}{t}\delta_r^0\tau_{lm} - \delta_{ri}B^{ijK}g_{Kjlm}
               +[(\kt+\ellt+\mt)*(\kt+\ellt+\mt)]_{rlm} .
\end{align}
The expansion \eqref{lem-P-exp.1} is then a direct consequence of \eqref{P-def}, \eqref{Dq-exp}-\eqref{g0rlm-exp} and \eqref{kt-def}-\eqref{mt-def}.

Finally, from the definition of the covariant derivative $\Dc_q$ and \eqref{e0-mu}, we have
\begin{align*}
    \delta^q_Q \beta \Dc_k g_{qjlm} &= \beta e_k(g_{Qjlm})-\beta \gamma_k{}^p{}_{Q} g_{pjlm}
    -\beta \gamma_k{}^p{}_{j} g_{Qplm}
     -\beta \gamma_k{}^p{}_{l} g_{Qjpm}
      -\beta \gamma_k{}^p{}_{m} g_{Qjlp}\\
     &= \delta_k^0\del{t}g_{Qjlm} +\delta_k^K\beta e_K(g_{Qjlm}) -\beta \gamma_k{}^0{}_{Q} g_{0jlm}
    -\Bigl(\beta \gamma_k{}^P{}_{Q} g_{Pjlm}\\
    &\hspace{4.0cm} +\beta \gamma_k{}^p{}_{j} g_{Qplm}
     +\beta \gamma_k{}^p{}_{l} g_{Qjpm}
      +\beta \gamma_k{}^p{}_{m} g_{Qjlp}\Bigr)
\end{align*}
From this expression, we see that \eqref{lem-P-exp.2} follows from \eqref{eq:gammadecomp1}, \eqref{eq:gammadecomp2}, \eqref{g0rlm-exp} and the definitions \eqref{kt-def}-\eqref{psit-def}. Furthermore, employing similar arguments, it is not difficult with the help of \eqref{for-K.2} to verify that the expansion \eqref{lem-P-exp.3} also holds. 
\end{proof}

\begin{lem}\label{lem-JK-exp}
\begin{align}
  \label{eq:defJ}
  \beta J^j_L=&   \bigl[\beta (\beta\tau+ \kt +\ellt +\mt) *\tau+\beta(\beta\tau+\mt+\ellt)\tau_{00}\bigr]^j_L,\\
  \beta K^i_{Ql} =& \Kf^i_{Ql} 
  \added{+\delta^i_0 \beta(2\mt_{Q}-2 \delta^{JK}\ellt_{JKQ}+\delta^{JK}\ellt_{QJK})\Dc_l\Wtt},
\end{align}
where
\begin{align}
  \label{eq:defKf}
  \Kf =& (\beta\kt+\beta\ellt+\beta\mt)*\taut + \beta\gt*\tau + (\ellt+\mt)*(\beta\kt+\beta\ellt+\beta\mt)*(\tau_{00}+\tau). 
\end{align}
\end{lem}
\begin{proof}
The expansion for $\beta J^j_l$ follows directly from \eqref{J-def}, \eqref{gi00}, \eqref{gamma-000}-\eqref{gamma-0KJ} and the definitions \eqref{kt-def}-\eqref{tau-def}. The validity of the second expression  for $\beta K^i_{Ql}$ is similarly a straightforward consequence of \eqref{gi00}, \eqref{K-def}, \eqref{for-K.2} and the definitions
\eqref{kt-def}-\eqref{taut-def}.
\end{proof}

\section{Fuchsian formulation}
\label{sec:FuchsianForm}
With the assistance of the preliminary calculations carried out in the previous sections, we are now in a position to transform the reduced conformal Einstein–scalar field equations into Fuchsian form that is suitable for establishing the existence of solutions up to the singularity. In particular, this formulation allows for the derivation of energy estimates from which uniform bounds on the gravitational and matter fields near the big bang singularity can be obtained. The remainder of this section is devoted to the construction of the Fuchsian formulation.

\subsection{$\beta$-renormalised equations}
The starting point for the derivation of the Fuchsian formulation is to rewrite the first-order formulation of the reduced conformal Einstein–scalar field system obtained in Section~\ref{sec:FirstOrderForm} in terms of the variables \eqref{kt-def}–\eqref{psit-def}, after a suitable renormalisation.
We begin the derivation of this renormalised system by introducing the $\beta$-weighted variables $\beta \gt_{Qjlm}$ in place of $\gt_{Qjlm}$. As observed in \cite{BeyerOliynyk:2021}, this rescaling is essential for obtaining a first-order formulation with a favourable Fuchsian structure. Using \eqref{grad-alpha-A} and \eqref{for-I}, it is straightforward to verify that $\beta \gt_{Qjlm}$ satisfies
\begin{align}
  \delta^{ij}\del{t}(\beta \gt_{Qjlm}) +B^{ijK}\beta e_K^\Lambda\partial_\Lambda(\beta \gt_{Qjlm})&= \frac{1\added{-\Rtt}}{t} \delta_0^i\delta_0^j  (\beta \gt_{Qljm}+ \beta \gt_{Qmjl}-\beta \gt_{Qjlm}) \notag \\
  &\quad+\frac{1}{2}\delta^{ij} \delta^{JK}\beta\kt_{JK}\beta\gt_{Qjlm}+\frac{2}{t}\delta^i_0\beta^2\taut_{Q(lm)}+\Hf^i_{Qlm}
 ,  \label{for-I.S2a}
\end{align}
where we have set
\begin{align}
    \Hf^i_{Qlm}=& -\delta^{ij}\beta^2 \tau_{00}\beta\gt_{Qjlm} + B^{ijK}\Bigl(-\beta^2 \tau_{K0}-\frac{1}{2}\beta\bigl( 2 \mt_K  
   -\delta^{JI}(2\ellt_{JIK}
   -\ellt_{KJI})\bigr)\Bigr)\beta \gt_{Qjlm}  \notag \\
   & 
     + \delta^i_0 \beta\Pf_{Q(lm)}-B^{ijk}\beta\Gf_{Qkj(lm)}.\label{Hf-def}
\end{align}
For similar reasons, we also introduce  the $\beta$-weighted variables $\beta\kt_{LM}$, which by \eqref{grad-alpha-A}, \eqref{for-H} and \eqref{kt-def} satisfy
\begin{align}
    \del{t}(\beta\kt_{LM}) =& -\frac{1\added{-\Rtt}}{t}\beta\kt_{LM}-\beta^2 \tau_{00} \beta\kt_{LM}  +\frac{2(1\added{-\Rtt})}{t}\beta^2\tau_{(LM)}\notag \\
    &- \beta B^{0jK}\beta\gt_{Kj(LM)} +2\beta B^{IjK}\beta\gt_{Kj0(L}\delta_{M)I} 
      +\Mf_{LM}
      , \label{for-H.S2a}
\end{align}
where
\begin{align} 
    \Mf_{LM} =&  \beta^2\Qft^0_{(LM)}  -2\delta_{MI}\beta^2 \Qft^I_{(L}\delta_{M)I}. 
    \label{Kf-def}
\end{align}

Before proceeding, it is worth noting that the corresponding definition of the quantity $\Hf^i_{Qlm}$ from \cite{BeyerOliynyk:2021} that appears in the evolution equations for $\beta \gt_{Qjlm}$ differs from \eqref{for-I.S2a} by the term $\frac{2}{t}\delta^i_0\beta^2\taut_{Q(lm)}$. In the present analysis, we need to treat this term separately in order to track the singular contributions that arise from this term, which was not necessary in  \cite{BeyerOliynyk:2021}.  For later use, we also note from Lemmas~\ref{Qijk-lem} and~\ref{lem-P-exp} that the quantities defined by \eqref{Hf-def} and \eqref{Kf-def} can be expressed as
\begin{align}
  \Hf=&
  \Bigl(\beta^2 \tau_{00}+\beta^2 \tau+\beta\kt+\beta\psit+\beta\mt+\beta\ellt\Bigr)\beta \gt\notag\\
  &+ (\beta\kt+\beta\ellt+\beta\mt) *(\beta\tau+\psit+\ellt+\mt)*(\beta\kt+\beta\ellt+\beta\mt)\notag\\ 
   &+ \frac 1t \beta\tau_{00}*(\beta^2\tau+\beta\ellt+\beta\mt+\beta\psit)
  +\frac{1}{t}(\beta\tau + \psit + \ellt +\mt)*(\beta^2\tau+\beta\kt+\beta\ellt+\beta\mt) 
  ,
  \label{Hf-exp}
\end{align}
{and}
\begin{align}
  \Mf &= (\beta\kt+\beta\ellt+\beta\mt)*(\beta^2\tau+\beta\ellt+\beta\mt),\label{Mf-exp}
\end{align}
respectively. These expressions constitute a slight refinement of those given in \cite[Lemma~8.4]{BeyerOliynyk:2021}, reflecting the fact that we now distinguish between the individual components $\tau_{00}$ and $\tau=(\tau_{iJ},\tau_{Ji})$. Indeed, in the present work, it becomes necessary to eliminate the variable $\tau_{00}$ entirely from the system of evolution equations and instead express it via the algebraic equation
\begin{equation}
  \label{eq:tau00identitygen.NN}
  \tau_{00}=\delta^{RL}\tau_{RL}
  +\frac{\Rtt}{t}\bigl(\beta^{-2}-1\bigr), 
\end{equation}
which follows from \eqref{confESFG}, \eqref{Ttt-Wtt-fix}, \eqref{tau-fix}, \eqref{orthog},  \eqref{p-fields} and \eqref{grad-tau}. In order to keep the equations that follow concise, we continue to write $\tau_{00}$, with the understanding that it is to be interpreted as the function defined by \eqref{eq:tau00identitygen.NN}.

The final computation required for our renormalised first-order formulation involve derivatives of the scalar field source term \eqref{Wtt-def} (see also \eqref{Ttt-Wtt-fix} ), that is,
\begin{equation}
    \Wtt 
     = \frac{\Rtt}{t} \bigl(g^{ij}\Dc_i\tau\Dc_j\tau+1\bigr). \label{eq:Wtt-def.2}
\end{equation}
Recalling that $\Dc_i$ is the covariant derivative of the flat background metric $\gc_{ij}$, a straightforward computation involving  \eqref{tau-fix},  \eqref{orthog}, \eqref{p-fields}, \eqref{d-fields} and \eqref{grad-tau} shows that
\begin{align}
  \Dc_k\Wtt
  =&\Rtt \beta^{-3}\delta_k^0t^{-2} 
              -\Rtt \beta^{-2}t^{-1}g_{k00}
              -2\Rtt \beta^{-1}t^{-1}\tau_{k0}
              -\delta_k^0\beta^{-1}\Rtt t^{-2},
              \\
  \Dc_l\Dc_k\Wtt
  =& 
                   -2t^{-3}\Rtt \beta^{-4}\delta_k^0\delta_l^0
                   +2\Rtt t^{-2}\beta^{-3}\delta_{(k}^0 g_{l)00}
                   +\Rtt t^{-2}\beta^{-2} \tau_{lk}
                   +4\Rtt t^{-2} \beta^{-2} \delta_{(k}^0\tau_{l)0}\notag \\
                  &
                    -\Rtt t^{-1}\beta^{-2}g_{lk00}
                    +2\Rtt t^{-1}\beta^{-2} g^{in} g_{ki0}g_{ln0}
                    +4\Rtt t^{-1}\beta^{-1}g^{ij}\tau_{i(l} g_{k)j0}\notag\\
                  &                    
                    -2\Rtt t^{-1}\beta^{-1}\tau_{lk0}
                    +2\Rtt t^{-1}g^{ij}\tau_{ki}\tau_{lj}-\Rtt t^{-2}\tau_{kl}
                     +2\Rtt t^{-3}\beta^{-2}\delta_k^0 \delta_l^0. \label{eq:Wtt-def.Last}
\end{align}
Then using \eqref{gi00}, \eqref{kt-def}--\eqref{psit-def}, \eqref{Hf-def},  \eqref{B0-def}--\eqref{BK-def} and \eqref{eq:Wtt-def.2}--\eqref{eq:Wtt-def.Last}, we arrive at the following first-order system:
\begin{align}
  \del{t}\beta =& -\beta^3 \tau_{00}+\frac{1}{2} t^{-1+\Rtt}\delta^{JK} (t^{1-\Rtt}\beta\kt_{JK}) \beta, \label{for-O.1.S2.2}\\
  \del{t}(t\taut_{Q0L}) -\delta^{JK}\beta e_K^\Lambda\partial_\Lambda (t\taut_{QJL})
  =&(2\Rtt+1) t^{-1} (t\taut_{Q0L})                                    
                    +\Rtt t^{-1}\beta^{-2}(t\beta g_{QL00})-t^{-1}\tau_{QL}\Rtt(\beta^{-1}-\beta)     \notag\\
                    &-\Rtt \beta^{-1}\Bigl(2 g^{in} g_{Qi0}g_{Ln0}
                    +4\beta g^{ij}\tau_{i(Q} g_{L)j0}
                    +2\beta^2 g^{ij}\tau_{Qi}\tau_{Lj}\Bigr)\notag\\
  &-\Rtt \beta^{-1}(2\mt_{Q}-2 \delta^{JK}\ellt_{JKQ}+\delta^{JK}\ellt_{QJK}
  \notag \\
  & +\beta \tau_{0Q})\bigl(
              g_{L00}
              +2\beta\tau_{L0}
                   \bigr) +t\Kf^0_{QL} -tB^{0jk}\Tf_{QkjL}, \label{for-J.S2.2.2.N.First}\\
  \del{t}(t\beta \gt_{Q000}) -\delta^{JK}\beta e_K^\Lambda\partial_\Lambda(t\beta \gt_{QJ00})=& \frac{2\added{-\Rtt}}{t} (t\beta \gt_{Q000}) +\frac{1}{2}t^{-1+\Rtt} \delta^{JK}(t^{1-\Rtt}\beta\kt_{JK})( t\beta\gt_{Q000}) \notag \\
                 &+\frac{2}{t}\beta^2(t\taut_{Q00})+t\Hf^0_{Q00},\\
  \del{t}(t\taut_{Q00}) -\delta^{JK}\beta e_K^\Lambda\partial_\Lambda (t\taut_{QJ0})=&(2\Rtt+1) t^{-1} (t\tau_{Q00})+\Rtt t^{-1}\beta^{-2}(t\beta g_{Q000})-2\Rtt t^{-1} \beta^{-1}\tau_{0Q}\notag \\
  &-t^{-1}\Rtt(2\mt_{Q}-2 \delta^{JK}\ellt_{JKQ}+\delta^{JK}\ellt_{QJK})
                   \notag\\       
  &              +\Rtt t^{-1} \beta^{-2} (2\mt_{Q}-2 \delta^{JK}\ellt_{JKQ}+\delta^{JK}\ellt_{QJK}
  \notag \\
  &+\beta \tau_{0Q})\Bigl[ t^{\Rtt}\delta^{JK} t^{1-\Rtt}\beta \kt_{JK}
              -2 \beta^2 t\tau_{00}                   
    \Bigr]\notag\\                 
 & +\Rtt \beta^{-2} \Bigl[-2\beta g^{in} g_{0i0}g_{Qn0} +4\beta^2\tau_{0(0} g_{Q)00}
 -4\beta^2\delta^{IJ}\tau_{I(0} g_{Q)J0}\notag\\
 &+2\beta^3 \tau_{Q0}\tau_{00}-2\beta^3 \delta^{IJ}\tau_{QI}\tau_{0J}\Bigr]
  +t\Kf^0_{Q0}-tB^{0jk}\Tf_{Qkj0}, \label{for-J.S2.2.1}\\
  \del{t}\mt_{M} =& \frac{1\added{-\Rtt}}{t}\bigl(2 \mt_{M}-\delta^{IJ}(2 \ellt_{IJM}-\ellt_{MIJ})\bigr)+\frac{2(1\added{-\Rtt})}{t}\beta\tau_{0M}\label{for-G.1.S2.2}\\
  &
                    + \frac 1t \eta^{JK} (t\beta \gt_{KJ0M})+\beta \Qft_M, \notag\\
\del{t}\tau_{0L} =&2\Rtt t^{-1}\tau_{0L}+\Rtt \beta^{-1}t^{-1}\bigl(
              2\mt_{L}-\delta^{JK}(2\ellt_{JKL}-\ellt_{LJK})
              \bigr) \notag \\
 &+t^{-1}\delta^{JK} \beta(t\taut_{KJL})+\beta J^0_L, \\
\del{t}e^\Lambda_I =& -\frac{1}{2}t^{-1+\Rtt}\delta^{JL}(t^{1-\Rtt}\beta\kt_{IL}) e^\Lambda_J -\delta^{JK}\beta\psit_{I}{}^0{}_K e^\Lambda_J,\label{for-M.1.S2.2}\\
\del{t}\psit_I{}^k{}_J=&
                         -\frac{1}{2}\delta^k_0 t^{-1}(t\beta \gt_{IJ00})-\frac{1}{2}t^{-1}\eta^{kl}\bigl(t\beta \gt_{I0Jl}+t\beta \gt_{IJ0l}-t\beta \gt_{Il0J}\bigr)\notag\\
                 &-t^{-1}\delta^k_0\beta^2 (t\taut_{IJ0})+\beta \Lf_I{}^k{}_J,\label{for-N.S2.2}\\
  \del{t}(t^{1-\Rtt}\beta\kt_{LM}) =& -\beta^2 \tau_{00} (t^{1-\Rtt}\beta\kt_{LM})  +2(1\added{-\Rtt}) t^{-\Rtt}\beta^2\tau_{(LM)}\notag \\
    &+ t^{-\Rtt}\beta \delta^{JK} (t\beta\gt_{KJ(LM)}) -2 t^{-\Rtt}\beta (t\beta\gt_{K00(L})\delta_{M)}^{K} 
\notag \\
  &+t^{1-\Rtt}\Mf_{LM}, \\
  \del{t}\ellt_{R0M} =& t^{-1}(t\beta\gt_{R00M})+\delta_{RI}\beta \Qft^I_M, \label{for-G.2.S2.2}\\
  \del{t}\ellt_{RLM} =& t^{-1} (t\beta\gt_{R0LM})+
                        \delta_{RI}\beta \Qft^I_{LM}, \label{for-G.4.S2.2}\\
\del{t}\tau_{RL} =&t^{-1}\beta (t\taut_{RL0})+
                    \delta_{RI}\beta J^I_L, \label{for-F.2.S2.2.3}\\
  \del{t}(t\beta \gt_{Q0LM}) -\delta^{JK}\beta e_K^\Lambda\partial_\Lambda(t\beta \gt_{QJLM})=& \frac{1\added{-\Rtt}}{t} (t\beta \gt_{QL0M}+ t\beta \gt_{QM0L})+\frac{\Rtt}t (t\beta \gt_{Q0LM}) \notag \\
                 &+\frac{2}{t}\beta^2(t\taut_{Q(LM)})+\frac{1}{2}t^{-1+\Rtt} \delta^{JK}(t^{1-\Rtt}\beta\kt_{JK})(t\beta\gt_{Q0LM})
                   \notag\\
                   &+t\Hf^0_{QLM},\\
  \del{t}(t\beta \gt_{QNLM}) -\beta e_N^\Lambda\partial_\Lambda(t\beta \gt_{Q0LM})=&\frac 1t (t\beta \gt_{QNLM})+\frac{1}{2}t^{-1+\Rtt}\delta^{JK}(t^{1-\Rtt}\beta\kt_{JK})(t\beta\gt_{QNLM}) \notag \\
                 &+t\delta_{IN}\Hf^I_{QLM},\\
  \del{t}(t\beta \gt_{Q00M}) -\delta^{JK}\beta e_K^\Lambda\partial_\Lambda(t\beta \gt_{QJ0M})=& \frac 1t (t\beta \gt_{Q00M})+\frac{1\added{-\Rtt}}{t} (t\beta \gt_{QM00}) \notag  \\
  &+\frac{1}{2} t^{-1+\Rtt}\delta^{JK}(t^{1-\Rtt}\beta\kt_{JK})(t\beta\gt_{Q00M})+\frac{2}{t}\beta^2(t\taut_{QM0})\notag\\
                 &+t\Hf^0_{Q0M},\\
  \del{t}(t\beta \gt_{QN0M}) -\beta e_N^\Lambda\partial_\Lambda(t\beta \gt_{Q00M})=&\frac 1t (t\beta \gt_{QN0M})+\frac{1}{2}t^{-1+\Rtt}\delta^{JK}(t^{1-\Rtt}\beta\kt_{JK})(t\beta\gt_{QN0M}) \notag \\
                 &+t\delta_{IN}\Hf^I_{Q0M},\\
  \del{t}(t\beta \gt_{QN00}) -\beta e_N^\Lambda\partial_\Lambda(t\beta \gt_{Q000})=&\frac 1t (t\beta \gt_{QN00})+\frac{1}{2}t^{-1+\Rtt}\delta^{JK}(t^{1-\Rtt}\beta\kt_{JK})(t\beta\gt_{QN00}) \notag \\
                 &+t\delta_{IN}\Hf^I_{Q00},\\
  \del{t} (t\taut_{QN0}) -\beta e_N^\Lambda\partial_\Lambda (t\taut_{Q00})=&                                                                     t^{-1} (t\taut_{QN0})+\beta^{-2}t^{-1}\psit_N{}^0{}_Q \Rtt\bigl(1  +\beta^2\bigr)\notag \\
     &+\Rtt \beta^{-2}t^{-1}\psit_N{}^0{}_Q
     \bigl(t^{\Rtt}\delta^{JK}(t^{1-\Rtt}\beta\kt_{JK})
     -2 \beta^2t\tau_{00}\bigr)\notag\\
     & +\delta_{IN}t\Kf^I_{Q0}
       -\delta_{IN} t B^{Ijk}\Tf_{Qkj0}
       ,  \label{for-J.S2.2.1.N.N}\\
  \del{t} (t\taut_{QNL}) -\beta e_N^\Lambda\partial_\Lambda (t\taut_{Q0L})=&
       t^{-1} (t\taut_{QNL})-\Rtt  \beta^{-1}\psit_N{}^0{}_Q\bigl(g_{L00}
              +2\beta t\tau_{L0}\bigr)         +\delta_{IN}t\Kf^I_{QL} \notag \\
              & -\delta_{IN} B^{Ijk}t\Tf_{QkjL}. \label{for-J.S2.2.2.N}
\end{align}
where $Q^i_{lm}$ is determined by \eqref{Q-def} and \eqref{Qilm-def}, $S_{LM}$ and $J^i_l$ are given by \eqref{S-def} and \eqref{J-def}, respectively, and $g_{000}$ and $g_{I00}$ are determined by \eqref{gi00}. Moreover, $P_{qlm}$ is determined by \eqref{P-def} and \eqref{grad-alpha-A}, $K^i_{ql}$ by \eqref{K-def}, $g_{0rlm}$ and $\tau_{0rl}$ by \eqref{for-K.1} and \eqref{for-K.2}, respectively, and $L_I{}^k{}_J$ by \eqref{L-def}. 
In addition, the coefficients $\gamma{}_0{}^0{}_0$, $\gamma{}_0{}^K{}_0$, $\gamma_0{}^0{}_J$, $\gamma_0{}^K{}_J$, $\gamma_I{}^0{}_0$ and $\gamma_I{}^J{}_0$ are determined by \eqref{gamma-000}–\eqref{gamma-IJ0}, while $\Lf_I{}^k{}_J$ is given by \eqref{lem-cov-exp.3}, $\Pf_{Qlm}$, $\Gf_{Qkjlm}$ and $\Tf_{Qkjl}$ by \eqref{lem-P-exp.1}–\eqref{lem-P-exp.3}, and $\Qft_M$, $\Qft^I_M$, $\Qft^0_{LM}$, and $\Qft^I_{LM}$ by \eqref{Qijk-lem.1}. We further note that
\begin{align}
  B^{ij0}&=\delta^{ij}, \label{B0-def}\\
B^{ijK} &= -\delta^i_0 \eta^{jK}-\delta^j_0\eta^{iK} \label{BK-def}.
\end{align}

\subsection{Initial data\label{frame-idata}}
 
Before continuing, we describe, following \cite[\S9.2]{BeyerOliynyk:2021}, how the initial data $\{\gr_{\mu\nu},\ggr_{\mu\nu},\taur=t_0,\taugr\}$ for the reduced conformal Einstein–scalar field equations, together with an initial choice of spatial frame
\begin{equation} \label{frame-idata-A}
e_I^\mu\bigl|_{\Sigma_{t_0}} = \er^\mu_I,
\end{equation}
determine initial data for the first-order system \eqref{for-J.S2.2.2.N.First}–\eqref{for-J.S2.2.2.N}. The first step is to observe, by Remark~\ref{rem-Lag-idata}, that the initial data set ${\gr_{\mu\nu},\ggr_{\mu\nu},\taur=t_0,\taugr}$ determines, via \eqref{l-idata}–\eqref{chi-idata}, a corresponding Lagrangian initial data set
\begin{equation} \label{Lag-idata-set-A}
\bigl\{g_{\mu\nu}\bigl|_{\Sigma_{t_0}},\del{0}g_{\mu\nu}\bigl|_{\Sigma_{t_0}}, \tau\bigl|_{\Sigma_{t_0}}=t_0,\del{0}\tau\bigl|_{\Sigma_{t_0}}=1\bigr\}.
\end{equation}
Then, with the help of the reduced conformal Einstein-scalar field equations \eqref{lag-redeqns}, this initial data determines the following higher order partial derivatives on the initial hypersurface:
\begin{equation} \label{Lag-idata-set-B}
\bigl\{\del{\beta}g_{\mu\nu}\bigl|_{\Sigma_{t_0}},\del{\alpha}\del{\beta}g_{\mu\nu}\bigl|_{\Sigma_{t_0}},\del{\alpha}\tau\bigl|_{\Sigma_{t_0}}=\delta^0_\alpha,\del{\alpha}\del{\beta}\tau\bigl|_{\Sigma_{t_0}}, \del{\alpha}\del{\beta}\del{\gamma}\tau\bigl|_{\Sigma_{t_0}}\bigr\}.
\end{equation}
Furthermore, we note the flat metric $\gc_{\mu\nu}=\del{\mu}l^\alpha \eta_{\alpha\beta}\del{\nu} l^\beta$ and its first and second order partial derivatives on the initial hypersurface
\begin{equation} \label{Lag-idata-set-C}
\bigl\{\gc_{\mu\nu}\bigl|_{\Sigma_{t_0}},\del{\alpha}\gc_{\mu\nu}\bigl|_{\Sigma_{t_0}},\del{\alpha}\del{\beta}\gc_{\mu\nu}\bigl|_{\Sigma_{t_0}} \bigr\},
\end{equation}
are determined from the initial data set $\{\gr_{\mu\nu},\ggr_{\mu\nu},\taur=t_0,\taugr\}$ by
\eqref{Jc-def}, \eqref{vr-def}--\eqref{Jcr-def} and the
reduced conformal Einstein-scalar field equations.
Together, \eqref{Lag-idata-set-A}--\eqref{Lag-idata-set-C} determine the initial data
\begin{equation} \label{Lag-idata-set-D}
\bigl\{g_{\mu\nu}\bigl|_{\Sigma_{t_0}},\Dc_{\beta}g_{\mu\nu}\bigl|_{\Sigma_{t_0}},\Dc_{\alpha}\Dc_{\beta}g_{\mu\nu}\bigl|_{\Sigma_{t_0}},\tau\bigl|_{\Sigma_{t_0}}=t_0,\Dc_{\alpha}\tau\bigl|_{\Sigma_{t_0}}=\delta^0_\alpha,\Dc_{\alpha}\Dc_{\beta}\tau\bigl|_{\Sigma_{t_0}},\Dc_{\alpha}\Dc_{\beta}\Dc_{\gamma}\tau\bigl|_{\Sigma_{t_0}} \bigr\}.
\end{equation}

Next, on $\Sigma_{t_0}$, the frame vector $e_0^\mu$ is, see \eqref{alpha-def}-\eqref{e0-mu}, given by 
\begin{equation} \label{frame-idata-E}
e_0^\mu \bigl|_{\Sigma_{t_0}} = \frac{1}{\beta|_{\Sigma_{t_0}}}\delta^\mu_0   
\end{equation}
where 
\begin{equation} \label{frame-idata-F}
\beta|_{\Sigma_{t_0}} = (- g^{00}\bigl|_{\Sigma_{t_0}})^{-\frac{1}{2}}.
\end{equation}
Then using \eqref{frame-idata-E} and a Gram-Schmidt algorithm, we select (non-unique) spatial frame initial data \eqref{frame-idata-A} so that the frame $e_i=e_i^\mu\del{\mu}$ is initially orthonormal, i.e., $g(e_i,e_j)|_{\Sigma_{t_0}}=\eta_{ij}$.
From the relations
\begin{gather*}
\Dc_i g_{jk} = e_i^\beta e_j^\mu e_k^\nu \Dc_{\beta}g_{\mu\nu}, \quad \Dc_i\Dc_j g_{kl} = e_i^\alpha e_j^\beta e_k^\mu e_l^\nu \Dc_{\alpha}\Dc_{\beta}g_{\mu\nu}, \\
\Dc_i \Dc_{j}\tau = e_i^\beta e_j^\beta \Dc_{\beta}\Dc_{\beta}g_{\mu\nu}, \quad \Dc_i\Dc_j \Dc_k \tau = e_i^\alpha e_j^\beta e_k^\gamma \Dc_{\alpha}\Dc_{\beta}\Dc_{\gamma}\tau,
\end{gather*}
and the definitions \eqref{p-fields}--\eqref{d-fields}, it is then clear that \eqref{frame-idata-A}, \eqref{Lag-idata-set-D} and \eqref{frame-idata-E}--\eqref{frame-idata-F} determine the initial data
\begin{equation}\label{frame-idata-H}
    \bigl\{g_{ijk}\bigl|_{\Sigma_{t_0}},g_{ijkl}\bigl|_{\Sigma_{t_0}},\tau_{ij}\bigl|_{\Sigma_{t_0}},\tau_{ijk}\bigl|_{\Sigma_{t_0}}\bigr\}.
\end{equation}

Using the fact that the frame $e_j^\mu$ is orthonormal with respect to the metric $g_{\mu\nu}$, a direct computation shows that
$\Gamma_{I}{}^0{}_J{}|_{\Sigma_{t_0}} = \Ktt_{\Lambda\Omega}\er_I^\Lambda \er_J^\Omega$,
where $\Ktt=\Ktt_{\Lambda\Omega}dx^\Lambda\otimes dx^\Omega$ denotes the second fundamental form (cf.\ \eqref{gtt-def1}), which is determined by the initial data ${g_{\mu\nu}|_{\Sigma_{t_0}},\del{0}g_{\mu\nu}|_{\Sigma_{t_0}}}$. Combining this identity with \eqref{Ccdef}, we deduce that
\begin{equation} \label{frame-idata-J}
    \gamma_I{}^0{}_J\bigl|_{\Sigma_{t_0}} = \Ktt_{\Lambda\Omega}\er_I^\Lambda \er_J^\Omega +\frac{1}{2}(g_{IJ0}+g_{JI0}-g_{0IJ})\bigl|_{\Sigma_{t_0}}.
\end{equation}
Moreover, initial data for the connection coefficients
$\Gamma_I{}^K{}_J$ can be computed from the initial data for the spatial frames $\ett_I = \er^\Lambda_I \del{\Lambda}$
and the spatial metric
$\gtt=\gr_{\Lambda\Omega}dx^\Lambda \otimes dx^\Omega$, c.f.\ \eqref{gtt-def1}, using
$\Gamma_I{}^K{}_J|_{\Sigma_{t_0}} = \delta^{KL}\gtt(\Dtt_{\ett_I}\ett_J,\ett_L)$,
where $\Dtt$ is the Levi-Civita connection of $\gtt$. This identity can then be combined with \eqref{Ccdef} to obtain
\begin{equation}  \label{frame-idata-K}
    \gamma_I{}^K{}_J \bigl|_{\Sigma_{t_0}} = \delta^{KL}\gtt(\Dtt_{\ett_I}\ett_J,\ett_L)
    - \frac{1}{2}\delta^{KL}(g_{IJL}+g_{JIL}-g_{LIJ})\bigl|_{\Sigma_{t_0}}.
\end{equation}
Together, \eqref{frame-idata-A} and \eqref{frame-idata-F}--\eqref{frame-idata-K} determine, via \eqref{kt-def}--\eqref{psit-def}, initial data on $\Sigma_{t_0}$
for the system \eqref{for-J.S2.2.2.N.First}--\eqref{for-J.S2.2.2.N}.

\subsection{$t$-renormalised equations}
\label{sec:tnormalisedequations}
The next step needed to bring the equations \eqref{for-O.1.S2.2}--\eqref{for-J.S2.2.2.N} into a form that is suitable for analysing the \emph{positive ekpyrotic regime} is to renormalise the first-order variables \eqref{kt-def}-\eqref{psit-def} by a suitable power of $t$ as follows: 
\begin{align}
  \gac_{Q000}&=t^{1+\ep}\beta g_{Q000}, \label{gac-Q000}\\
  \tauac_{Q00}&=t^{1+\ep}\tau_{Q00},\\
  \tauac_{QN0}&=t^{1+\ep}\tau_{QN0},\\
  m_Q&=t^{\ep}\mt_Q, \label{m-def}\\
  \xi_{0L}&=t^{\ep}\tau_{0L},\label{xi-def.N}\\
  \xi_{RL}&=t^{\ep}\tau_{RL},\\%\label{xi-def.N}
\psi_I{}^k{}_J&=t^{\ep}\psit_I{}^k{}_J,\label{psi-def.NN}\\
  k_{IJ}&=t^{1-\Rtt}\beta\kt_{IJ},\label{k-def.NN}\\
  \ell_{RiJ}&=t^{\ep}\ellt_{RiJ},\label{l-def.NN}\\
  \gac_{Q0LM}&=t^{1+\ep}\beta g_{Q00M},\\
  \gac_{QNLM}&=t^{1+\ep}\beta g_{QN0M},\\
  \gac_{Q00M}&=t^{1+\ep}\beta g_{Q00M},\\
  \gac_{QN0M}&=t^{1+\ep}\beta g_{QN0M},\\
  \tauac_{Q0L}&=t^{1+\ep}\tau_{Q0L}, \\
  \tauac_{QNL}&=t^{1+\ep}\tau_{QNL}, \\
  \label{gacclastdef}
  \gac_{QN00}&=t^{1+\ep}\beta g_{QN00},
\end{align}
where $\ep>0$ is, for the moment, an arbitrary positive constant.

\begin{rem}\label{rem:peFLRW}
From the discussion in Section~\ref{sec:FLRWEinsteinSF}, it is not difficult to verify that the positive ekpyrotic-FLRW solution can be expressed as
\begin{equation*}
 % \label{eq:EkpyroticSol.Form1}
  \bigl(\beta,e_{I}^\lambda,
  k_{IJ},m_Q,\ell_{RiJ},\xi_{iL},\psi=0,\gac_{Qijk},\tauac_{Qij}\bigr)=(1,\delta^\Lambda_I,0,0,0,0,0,0,0).
\end{equation*}
\end{rem}

The final step needed to bring the evolution equations into a form suitable for analysing small perturbations of the positive ekpyrotic-FLRW solution is to replace $\beta$ by a renormalised version $\betah$, where the two are related via
\begin{equation}
\label{eq:ekpyroticbeta}
 \beta=1+t^{\Rtt-\ep}\betah
\end{equation}
with $\ep>0$ a positive constant to be determined. The heuristic motivation for this change of variables is that $\beta$ should approach $1$ in the ekpyrotic regime as $t\searrow 0$, and the new variable is designed to capture this expected convergence.

In what follows, we replace $\beta$ in favour of $\betah$ and replace its evolution equation with
\begin{equation}
\label{eq:betahevol_N}
\del{t}\betah=t^{-1}\bigl(
\Rtt-\ep+\Rtt t^{\Rtt-\ep}\betah
(3+t^{\Rtt-\ep}\betah)\bigr)\betah
+ \frac{1}{2} t^{-1+\ep}\delta^{JK} k_{JK} \beta
-t^{-\Rtt}\beta^3\delta^{RL}\xi_{RL}.
\end{equation}
For example, \eqref{eq:tau00identitygen.NN} becomes
\begin{equation*}
%\label{eq:5}
\tau_{00}=t^{-\ep}\delta^{RL}\xi_{RL}-\Rtt t^{-1+\Rtt-\ep}\beta^{-2}(2+t^{\Rtt-\ep}\betah)\betah.
\end{equation*}
To keep the equations concise, we occasionally continue to write $\beta$, with the understanding that it denotes the function defined by \eqref{eq:ekpyroticbeta}.

Collecting \eqref{for-J.S2.2.2.N.First}--\eqref{for-J.S2.2.2.N} and \eqref{eq:betahevol_N} together, we arrive at a first-order formulation of the reduced conformal Einstein-scalar field equations that is suitable for analysing small perturbations of the positive ekpyrotic-FLRW solution: 
\begin{align}
  \del{t}e^\Lambda_I =& -\frac{1}{2}t^{-1+\Rtt}\delta^{JL}k_{IL} e^\Lambda_J -\delta^{JK}t^{-\ep}\beta\psi_{I}{}^0{}_K e^\Lambda_J,\label{final.first.N}\\
  \del{t}k_{LM} =& -\beta^2 \tau_{00} k_{LM}
                   +2(1\added{-\Rtt}) t^{-\Rtt-\ep}\beta^2\xi_{(LM)} \notag \\
                                    &+ t^{-\Rtt-\ep}\beta \delta^{JK} \gac_{KJ(LM)}
                                      -2 t^{-\Rtt-\ep}\beta \gac_{K00(L}\delta_{M)}^{K} 
    {+t^{\ep}\beta \Lf_I{}^k{}_J},\\
  \del{t}\betah=&t^{-1}\bigl(\Rtt-\ep\bigr)\betah+t^{-1+\Rtt-\ep}\Rtt 
    (2+\beta)\betah^2
    + \frac{1}{2} t^{-1+\ep}\delta^{JK} k_{JK} \beta\notag\\
    &-t^{-\Rtt}\beta^3\delta^{RL}\xi_{RL}, \\
  \del{t} \gac_{Q000} 
  -\delta^{JK}\beta e_K^\Lambda\partial_\Lambda\gac_{QJ00}
  =& \frac{2-\Rtt+\ep}{t} \gac_{Q000}
  +2t^{-1}(1+t^{\Rtt-\ep}\betah)^2\tauac_{Q00} \notag \\
  &+\frac{1}{2}t^{-1+\Rtt} \delta^{JK}k_{JK}\gac_{Q000}
     {+t^{1+\ep}\Hf^0_{Q00}},\\
  \del{t}\tauac_{Q00} -\delta^{JK}\beta e_K^\Lambda\partial_\Lambda \tauac_{QJ0}
  =&  (2\Rtt+1+\ep) t^{-1} \tauac_{Q00}
  +\Rtt t^{-1}(1+t^{\Rtt-\ep}\betah)^{-2}g_{Q000} \notag \\
  &-2\Rtt t^{-1} (1+t^{\Rtt-\ep}\betah)^{-1} \xi_{0Q}
  -2 \Rtt t^{-1} m_{Q} 
  +\Rtt t^{-1}\bigl(2 \delta^{JK}\ell_{JKQ}\notag \\
  &-\delta^{JK}\ell_{QJK}\bigr) 
  +\Rtt t^{-1} \beta^{-2} \Bigl[\bigl(m+ \ell+\beta \xi\bigr)*\bigl(t^{\Rtt}k+\beta^2 t\tau_{00}+t^{1-\ep}\beta m
   \notag \\
   &+t^{1-\ep}\beta\ell+t^{1-\ep}\beta^2\xi\bigr) \Bigr]_Q    
  +t^{1+\ep}\Kf^0_{Q0}-t^{1+\ep}B^{0jk}\Tf_{Qkj0}, \label{final.last.N.2}\\
  2\Rtt\del{t}m_{M} =&
    \frac{4\Rtt(1\added{-\Rtt})+2\Rtt\ep}{t} m_{M}
    +\frac{4\Rtt(1\added{-\Rtt})}{t}(1+t^{\Rtt-\ep}\betah)\xi_{0M} \\
    &-\frac{2\Rtt(1\added{-\Rtt})}{t}\delta^{IJ}(2 \ell_{IJM}-\ell_{MIJ}) 
    + 2\Rtt t^{-1}\delta^{JK} \gac_{KJ0M} + 2\Rtt t^{\ep}\beta\Qft_M,\\
  2(1-\Rtt)\del{t}\xi_{0L} =&2(1-\Rtt)(2\Rtt+\ep) t^{-1}\xi_{0L}
  +4(1-\Rtt)\Rtt (1+t^{\Rtt-\ep}\betah)^{-1}t^{-1}m_{L} \notag \\
    &-2(1-\Rtt)\Rtt (1+t^{\Rtt-\ep}\betah)^{-1}t^{-1}
      \delta^{JK}(2\ell_{JKL}-\ell_{LJK})\notag\\
    &+2(1-\Rtt)t^{-1}\delta^{JK} (1+t^{\Rtt-\ep}\betah) \tauac_{KJL}+2(1-\Rtt)t^{\ep}\beta J^0_L,\\
  \del{t}\xi_{RL} =&\ep\, t^{-1} \xi_{RL}
  + t^{-1}(1+t^{\Rtt-\ep}\betah) \tauac_{R0L}+\delta_{RI}t^{\ep}\beta J^I_L,\\%\label{for-F.2.S2.2.3}\\
\del{t}\psi_I{}^k{}_J=&\ep\,t^{-1} \psi_I{}^k{}_J 
-t^{-1}\delta^k_0(1+t^{\Rtt-\ep}\betah)^2\tauac_{I0J}\notag\\
&-\frac{1}{2}t^{-1}\eta^{kL}\bigl(\gac_{I0JL}+\gac_{IJ0L}-\gac_{IL0J}\bigr)+t^{1-\Rtt}\Mf_{LM},\\
  \del{t}\ell_{R0M} =& \ep\, t^{-1} \ell_{R0M} + t^{-1}\gac_{R00M}+\delta_{RI}t^{\ep}\beta \Qft^I_M,\\%\label{for-G.2.S2.2}\\
  \del{t}\ell_{RLM} =& \ep\, t^{-1} \ell_{RLM}  
  + t^{-1} \gac_{R0LM}+\delta_{RI}t^{\ep}\beta \Qft^I_{LM}, \\
  \del{t}\gac_{Q0LM} -\delta^{JK}\beta e_K^\Lambda\partial_\Lambda\gac_{QJLM}=& 
  \frac{\Rtt+\ep}t \gac_{Q0LM}
  +\frac{1\added{-\Rtt}}{t} (\gac_{QL0M}+ \gac_{QM0L}) 
  +{2}t^{-1}(1+t^{\Rtt-\ep}\betah)^2\tauac_{Q(LM)} \notag \\
  &+\frac{1}{2}t^{-1+\Rtt} \delta^{JK}k_{JK}\gac_{Q0LM}
                   {+t^{1+\ep}\Hf^0_{QLM}},\\
  \del{t}\gac_{QNLM} -\beta e_N^\Lambda\partial_\Lambda\gac_{Q0LM})=&(1+\ep)\frac 1t \gac_{QNLM}
  +\frac{1}{2}t^{-1+\Rtt}\delta^{JK}k_{JK}\gac_{QNLM} {+t^{1+\ep}\delta_{IN}\Hf^I_{QLM}},\\
  \del{t}\gac_{Q00M} -\delta^{JK}\beta e_K^\Lambda\partial_\Lambda\gac_{QJ0M}=& 
  (1+\ep)\frac 1t \gac_{Q00M}
  +\frac{1\added{-\Rtt}}{t} \gac_{QM00} 
  +{2}t^{-1}(1+t^{\Rtt-\ep}\betah)^2\tauac_{Q0M} \notag \\
    &+\frac{1}{2} t^{-1+\Rtt}\delta^{JK}k_{JK}\gac_{Q00M}{+t^{1+\ep}\Hf^0_{Q0M}},\\
  \del{t}\gac_{QN0M} -\beta e_N^\Lambda\partial_\Lambda\gac_{Q00M}
  =&\frac {1+\ep}t  \gac_{QN0M}+\frac{1}{2}t^{-1+\Rtt}\delta^{JK}k_{JK}\gac_{QN0M} 
    {+t^{1+\ep}\delta_{IN}\Hf^I_{Q0M} },\\
  \del{t}\tauac_{Q0L} -\delta^{JK}\beta e_K^\Lambda\partial_\Lambda \tauac_{QJL}
  =&(2\Rtt+1+\ep) t^{-1}\tauac_{Q0L}
  +\Rtt t^{-1}(1+t^{\Rtt-\ep}\betah)^{-2}\gac_{QL00}\\     
  &+t^{-1+\Rtt-\ep}\Rtt\xi_{QL}\beta^{-1}
  \betah(1+\beta) {+t^{1+\ep}\Kf^0_{QL} -t^{1+\ep}B^{0jk}\Tf_{QkjL}} \notag\\
  &-\Rtt t^{-\ep}\beta^{-1}\Bigl[(m+\ell+\beta\xi)*(m+\ell+\beta\xi)
  \Bigr]_{QL}, \label{final.symm.N}\\
  \del{t} \tauac_{QNL} -\beta e_N^\Lambda\partial_\Lambda \tauac_{Q0L}
  =&(1+\ep)t^{-1} \tauac_{QNL} 
     -\Rtt  \beta^{-1}t^{-\ep}\psi_N{}^0{}_Q\Bigl[m+\ell+t\beta \xi\Bigr] \notag \\
     &{+\delta_{IN}t^{1+\ep}\Kf^I_{QL} -\delta_{IN} B^{Ijk}t^{1+\ep}\Tf_{QkjL}}, \label{final.symm.N.2}\\
  \del{t}\gac_{QN00} -\beta e_N^\Lambda\partial_\Lambda\gac_{Q000}
  =&\frac {1+\ep}t \gac_{QN00}+\frac{1}{2}t^{-1+\Rtt}\delta^{JK}k_{JK}\gac_{QN00} 
  {+t^{1+\ep}\delta_{IN}\Hf^I_{Q00}}, \notag\\
 \del{t} \tauac_{QN0} -\beta e_N^\Lambda\partial_\Lambda \tauac_{Q00}
  =& (1+\ep)t^{-1} \tauac_{QN0} -{t^{-1+\Rtt-\ep}\Rtt\psi_N{}^0{}_Q \beta^{-2}
      \betah(1+\beta)} \notag \\
      &+\Rtt \beta^{-2}\psi_N{}^0{}_Q\Bigl(t^{-1+\Rtt}\delta^{JK}k_{JK}
      -2 \beta^2\tau_{00}\Bigr)\notag\\
      &{+\delta_{IN}t^{1+\ep}\Kf^I_{Q0}
      -\delta_{IN} t^{1+\ep} B^{Ijk}\Tf_{Qkj0}}.\label{final.last.N}
\end{align} 

We briefly comment on the variables $\tauac_{Q0L}$ and $\tauac_{QL0}$. Although these agree by their definition \eqref{p-fields} and \eqref{d-fields}, we treat them as independent unknowns here. As in \cite{BeyerOliynyk:2021}, this is necessary to obtain a symmetric hyperbolic system: \eqref{final.symm.N} and \eqref{final.symm.N.2} form one symmetrised pair, and \eqref{final.last.N} and \eqref{final.last.N.2} form another. We present the equations in this particular order, especially with the latter two equations separated, to make a certain upper triangular structure of the leading singular lower-order terms on the right-hand side manifest, which is crucial for the Fuchsian analysis. In any case, one readily checks that if $\tauac_{Q0L}=\tauac_{QL0}$ initially, then the full system implies that this equality is preserved for as long as the solution exists.

We further note that the source terms satisfy
\begin{align*}
t^{1+\ep}\Kf =& 
(t^{-1+\Rtt} k+t^{-\ep}\beta\ell+t^{-\ep}\beta m)*\tauac 
+ t^{-\ep}\gac*\xi \\
&+ (\ell+m)*(t^{\Rtt} k+t^{1-\ep}\beta\ell+t^{1-\ep}\beta m)*(\tau_{00}+t^{-\ep}\xi),\\
t^{1+\ep}\Tf=& (\beta^2 t^{-\ep}\xi+t^{-\ep}\beta\psi+ t^{-1+\Rtt} k+ t^{-\ep}\beta\ell +t^{-\ep}\beta m)*\tauac\\
&+ (\beta \xi+\psi +\ell +m)*(t^{\Rtt} k + t^{1-\ep}\beta\ell +t^{1-\ep}\beta m)*(\tau_{00}+t^{-\ep}\xi),\\
t^{1+\ep}\Hf
=& \Bigl(\beta^2 \tau_{00}+t^{-\ep}\beta^2 \xi+t^{-1+\Rtt} k+t^{-\ep}\beta\psi+t^{-\ep}\beta m+t^{-\ep}\beta\ell\Bigr)\gac\notag\\
&+ (t^{-1+\Rtt} k+t^{-\ep}\beta\ell+t^{-\ep}\beta m) *(\beta\xi+\psi+\ell+m)*(t^{\Rtt} k+t^{1-\ep}\beta\ell+t^{1-\ep}\beta m)\notag\\ 
 &+ \beta\tau_{00}*(\beta^2\xi+\beta\ell+\beta m+\beta\psi)\\
&+ (\beta\xi + \psi + \ell +m)*(t^{-\ep}\beta^2\xi+t^{-1+\Rtt} k+t^{-\ep}\beta\ell+t^{-\ep}\beta m), \\
t^{\ep}\beta\Qft &=(t^{-1+\Rtt}k+t^{-\ep}\beta\ell+t^{-\ep}\beta m)*\beta\xi 
+(t^{-1+\Rtt} k+t^{-\ep}\beta\ell+t^{-\ep}\beta m)*(\ell+m),\\
t^{\ep}\beta J &= (\beta^2t^{-\ep}\xi+ t^{-1+\Rtt}k +t^{-\ep}\beta\ell +t^{-\ep}\beta m) *\xi
+\beta(\beta\xi+m+\ell)\tau_{00},\\
t^{\ep}\beta\Lf&=  (t^{-\ep}\beta^2\xi
+t^{-\ep}\beta\psi+t^{-1+\Rtt} k+t^{-\ep}\beta\ell +t^{-\ep}\beta m)*(\beta \xi+\psi+\ell+m),\\
t\Mf &= t^{1-\ep}(t^{-1+\Rtt} k+t^{-\ep}\beta\ellt+t^{-\ep}\beta m)*(\beta^2\xi+\beta\ell+\beta m),
\end{align*}
where
$\Qft$,  $\Lf$, $\Tf$,  $J$, $\Kf$, $\Hf$ and $\Mf$ are defined by \eqref{eq:defQft}, \eqref{Lf-def}, \eqref{Tf-def},  \eqref{eq:defJ}, \eqref{eq:defKf},  \eqref{Hf-exp} and
 \eqref{Mf-exp}, respectively.

\subsection{Fuchsian formulation\label{sec:Fuch-form}}
We now cast the system \eqref{final.first.N}--\eqref{final.last.N} in matrix form by setting 
\begin{equation}\label{W-def}
  \begin{split}
W =\bigl(&e^\Lambda_I-\delta^\Lambda_I,k_{LM},
\betah,
\gac_{Q000},\tauac_{Q00},
m_M,\xi_{0L},
\xi_{RL},\psi_I{}^k{}_{J}, \ell_{R0M},\ell_{RLM},\\
&\gac_{Q0LM},\gac_{QNLM},
\gac_{Q00M},\gac_{QN0M},
\tauac_{Q0L},\tauac_{QNL},
\gac_{QN00},\tauac_{QN0}
\bigr)^{\tr},
  \end{split}
\end{equation}
In terms of $W$, the evolution equations take the Fuchsian form
\begin{equation} \label{Fuch-ev-A-W}
A^0\del{t}W -e_K^\Lambda A^K \del{\Lambda} W =
\frac{1}{t}\Ac\Pbb W +\frac{1}{t^{1-q}} F(t,W)
\end{equation}
where 
\begin{align} 
  q&=\min\{1-\Rtt-\epsilon, \Rtt-\ep,\ep,\Rtt,\ep\}, \label{q}\\
  A^0 &=\diag\Bigl(
    \delta^{\It I}\delta_{\Lambdat \Lambda},
    \delta^{\Lt L}\delta^{\Mt M},
    1,
    \delta^{\Qt Q},\delta^{\Qt Q},
    2\Rtt\delta^{\Mt M}, 2(1-\Rtt)\delta^{\Lt L},
    \delta^{\Rt R} \delta^{\Lt L}, \notag \\
    &\hspace{1.3cm}\delta^{\It I}\delta_{\kt k} \delta^{\Jt J}, 
    \delta^{\Rt R}\delta^{\Mt M},
    \delta^{\Rt R} \delta^{\Lt L} \delta^{\Mt M},
    \delta^{\Qt Q}\delta^{\Lt L} \delta^{\Mt M}, \delta^{\Qt Q}\delta^{\Nt N}\delta^{\Lt L} \delta^{\Mt M}, \notag  \\
    &\hspace{1.3cm}
    \delta^{\Qt Q} \delta^{\Mt M}, \delta^{\Qt Q}\delta^{\Nt N} \delta^{\Mt M},
    \delta^{\Qt Q} \delta^{\Lt L}, \delta^{\Qt Q}\delta^{\Nt N} \delta^{\Lt L},
    \delta^{\Qt Q}\delta^{\Nt N},\delta^{\Qt Q}\delta^{\Nt N}
    \Bigr), \label{A0-def}\\
  \Pbb &=\small \diag\Bigl(
    0,
    0,
    1,
    \delta^{\Qt Q},\delta^{\Qt Q},
    \delta^{\Mt M}, \delta^{\Lt L},
    \delta^{\Rt R} \delta^{\Lt L}, \delta^{\It I}\delta_{\kt k} \delta^{\Jt J}, \delta^{\Rt R}\delta^{\Mt M}, \notag \\
    &\hspace{1.3cm}\delta^{\Rt R} \delta^{\Lt L} \delta^{\Mt M}, 
    \delta^{\Qt Q}\delta^{\Lt L} \delta^{\Mt M}, \delta^{\Qt Q}\delta^{\Nt N}\delta^{\Lt L} \delta^{\Mt M},
    \delta^{\Qt Q} \delta^{\Mt M}, \notag \\
    &\hspace{1.3cm}\delta^{\Qt Q}\delta^{\Nt N} \delta^{\Mt M},
    \delta^{\Qt Q} \delta^{\Lt L}, \delta^{\Qt Q}\delta^{\Nt N} \delta^{\Lt L},
    \delta^{\Qt Q}\delta^{\Nt N},\delta^{\Qt Q}\delta^{\Nt N}
    \Bigr), \label{Pbb-def} 
\end{align}
and
$F(t,W)$ is a polynomial in $(W,\beta^{-1})$ and
satisfies
\begin{equation*}
F(t,0)=0.
\end{equation*}

Because $\Rtt \in (0,1)$, it is clear from \eqref{A0-def} that the matrix $A^0$ is symmetric and positive definite. 
Moreover, the matrices $A^K$ and $ \Ac$ are constant and can be expressed in block form as
\begin{align}
   A^K&=(A^K_{\Pc\Qc}), \label{AK-def}\\
     \Ac&=(\Ac_{\Pc\Qc}), \label{Ac-def}
\end{align}
where the the non-vanishing blocks of $A^K$ are 
\begin{equation} \label{AK-sym}
\begin{gathered}
A^K_{4\,18}=A^K_{18\,4} = \delta^{NK}\delta^{\Qt Q}, \\
A^K_{5\,19}=A^K_{19\,5} =\delta^{NK}\delta^{\Qt Q}, \\
A^K_{12\,13}=A^K_{13\,12} =\delta^{NK}\delta^{\Qt Q}\delta^{\Lt L} \delta^{\Mt M},\\
A^K_{14\,15}=A^K_{15\,14} =\delta^{NK}\delta^{\Qt Q} \delta^{\Lt L}, \\
A^K_{16\,17}=A^K_{17\,16} =\delta^{NK}\delta^{\Qt Q} \delta^{\Lt L},\\
\end{gathered}
\end{equation}
and the non-vanishing blocks of $\Ac$ are
\begin{equation} \label{Ac-blocks}
\begin{gathered}
    \Ac_{1\,1}=\delta^{\It I}\delta_{\Lambdat \Lambda}, \quad
    \Ac_{2\,2}=\delta^{\Lt L}\delta^{\Mt M}, \quad
    \Ac_{3\,3}=\Rtt-\ep, \quad 
    \Ac_{4\,4}=(2-\Rtt+\ep)\delta^{\Qt Q},\quad 
    \Ac_{4\,5}=2\delta^{\Qt Q}, \\
    \Ac_{5\,4}=\Rtt\delta^{\Qt Q},\quad 
    \Ac_{5\,5}=(2\Rtt+1+\ep)\delta^{\Qt Q},\quad 
    \Ac_{6\,6}=(4\Rtt(1-\Rtt)+2\Rtt\ep)\delta^{\Mt M},\\ 
    \Ac_{6\,7}=\Ac_{7\,6}=4\Rtt(1-\Rtt)\delta^{\Mt M},\quad 
    \Ac_{7\,7}=2(1-\Rtt)(2\Rtt+\ep)\delta^{\Lt L}, \\
    \Ac_{8\,8}=\ep\delta^{\Rt R} \delta^{\Lt L}, \quad 
    \Ac_{9\,9}=\ep\delta^{\It I}\delta_{\kt k} \delta^{\Jt J}, \quad 
    \Ac_{10\,10}=\ep\delta^{\Rt R}\delta^{\Mt M},\quad 
    \Ac_{11\,11}=\ep\delta^{\Rt R} \delta^{\Lt L} \delta^{\Mt M},\\
    \Ac_{12\,12}=(\Rtt+\ep)\delta^{\Qt Q}\delta^{\Lt L} \delta^{\Mt M}, \quad
    \Ac_{13\,13}=(1+\ep)\delta^{\Qt Q}\delta^{\Nt N}\delta^{\Lt L} \delta^{\Mt M}, \\
    \Ac_{14\,14}=(1+\ep)\delta^{\Qt Q} \delta^{\Mt M}, 
    \Ac_{15\,15}=(1+\ep)\delta^{\Qt Q}\delta^{\Nt N} \delta^{\Mt M}, 
    \Ac_{16\,16}=(2\Rtt+1+\ep)\delta^{\Qt Q} \delta^{\Lt L}, \\
    \Ac_{17\,17}=(1+\ep)\delta^{\Qt Q}\delta^{\Nt N} \delta^{\Lt L}, \quad
    \Ac_{18\,18}=(1+\ep)\delta^{\Qt Q}\delta^{\Nt N}, \quad
    \Ac_{19\,19}=(1+\ep)\delta^{\Qt Q}\delta^{\Nt N},
\end{gathered}
\end{equation}
along with constant matrices $\Ac_{5\,6}$, $\Ac_{5\,7}$, $\Ac_{5\,11}$, $\Ac_{6\,11}$, $\Ac_{7\,11}$, $\Ac_{9\,12}$, $\Ac_{10\,14}$, $\Ac_{6\,15}$, $\Ac_{9\,15}$, $\Ac_{12\,15}$, $\Ac_{8\,16}$, $\Ac_{9\,16}$, $\Ac_{14\,16}$, $\Ac_{7\,17}$, $\Ac_{12\,17}$, $\Ac_{14\,18}$, $\Ac_{16\,18}$ whose particular form is not important for the analysis.

The matrices $A^K$ and $\Ac$ satisfy two critical structural conditions. The first is that the $A^K$ are symmetric, and the second is that $\Ac$ is upper-block triangular with respect to the block decomposition determined by the following diagonal blocks:
\begin{equation} \label{Ac-d-blocks} 
\begin{gathered}
\Ac_{1\,1}, \quad \Ac_{2\,2}, \quad \Ac_{3\,3}, \quad \begin{pmatrix}\Ac_{4\,4} & \Ac_{4\,5} \\
\Ac_{5\,4} & \Ac_{5\,5} \end{pmatrix}, \quad
\begin{pmatrix}\Ac_{6\,6} & \Ac_{6\,7} \\
\Ac_{7\,6} & \Ac_{7\,7} \end{pmatrix}, \quad \Ac_{8\,8}, \quad \Ac_{9\,9}\quad  \Ac_{10\,10}, \\
 \Ac_{11\,11}, \quad  \Ac_{12\,12}, \quad  \Ac_{13\,13}, \quad  \Ac_{14\,14}, \quad  \Ac_{15\,15}, \quad  \Ac_{16\,16}, \quad  \Ac_{17\,17}, \quad  \Ac_{18\,18}, \quad \Ac_{19\,19}.
 \end{gathered}
\end{equation}

Now since $\Rtt\in (0,1)$, it is clear from \eqref{q} that 
\begin{equation} \label{q-pos}
q>0
\end{equation}
if and only if $\ep$ satisfies
\begin{equation}
  \label{epcond}
  0<\ep<\min\{1-\Rtt,\Rtt\}.
\end{equation}
For the remainder of this article, we assume that these inequalities hold. Given this, it is then not difficult to verify that 
\begin{equation} \label{F0-smooth}
F(t,W)\in C^0_b((0,t_0],C^\infty(B_R(\Rbb^{\udim}))
\end{equation}
for any $t_0$ provided 
$R>0$ is chosen small enough so that $\beta^{-1}$, see \eqref{eq:ekpyroticbeta}, is well defined. Here, $\udim$ denotes the dimension of the vector space that $W$ takes values in, c.f.~\eqref{W-def}. 

We further claim that assumption  \eqref{epcond} implies that the diagonal blocks of the matrix $\Ac$ are all bounded below by a positive number times the identity operator. Indeed from the above formulas, this is obvious for all the diagonal blocks except for
\begin{equation*}
 \begin{pmatrix}\Ac_{4\,4} & \Ac_{4\,5} \\
\Ac_{5\,4} & \Ac_{5\,5} \end{pmatrix} \AND
\begin{pmatrix}\Ac_{6\,6} & \Ac_{6\,7} \\
\Ac_{7\,6} & \Ac_{7\,7} \end{pmatrix}.
\end{equation*}
However, the fact that these operators are bounded below by a positive multiple of the identity follows directly from the $2\times2$-matrix bounds\footnote{The first inequality in these bounds is obtained from the fact that a square matrix $A$ is bounded below by $A\geq a \mathbbm{1}$ where  $a$ is the lowest eigenvalue of $\frac{1}{2}(A+A^T)$.}  
\begin{align*}
  \begin{pmatrix}
    2-\Rtt+\ep & 2\\
    \Rtt & 2\Rtt+1+\ep
  \end{pmatrix} &\geq \frac{1}{2} \Bigl(-\sqrt{10 \Rtt^2-2 \Rtt+5}+\Rtt+2 \epsilon +3\Bigr)\mathbbm{1}\\
  &\ge \frac{4-\sqrt{13}+ 2\ep}{2}\mathbbm{1},
\end{align*}
and
\begin{align*}
  2\begin{pmatrix}
    2(1-\Rtt)\Rtt+\Rtt\ep & 2\Rtt(1-\Rtt)\\
    2\Rtt(1-\Rtt) & 2\Rtt(1-\Rtt)+\ep(1-\Rtt)
  \end{pmatrix} &\ge \Bigl(\epsilon + 4 \Rtt(1 -  \Rtt) - \sqrt{\epsilon^2 (1 - 2 \Rtt)^2 + 16 (1 - \Rtt)^2 \Rtt^2}\Bigr)\mathbbm{1} \\
  & \ge \epsilon (1-|1-2\Rtt|)\mathbbm{1}.
\end{align*}

Because the matrices $A^0$ and $A^K$ are symmetric and $A^0\gtrsim \mathbbm{1}$, the system \eqref{Fuch-ev-A-W} is symmetric hyperbolic and can therefore be solved locally in time using the standard local existence and uniqueness theory for symmetric hyperbolic equations. To obtain global-in-time solutions that are valid for $t\in (0,t_0]$, we aim to apply the Fuchsian global existence theory developed in \cite{BOOS:2021,Oliynyk:CMP_2016}. In order to use this theory, the matrix $\Asc$ must satisfy $\Asc \gtrsim \mathbbm{1}$, which is equivalent to requiring that the symmetric matrix $\frac{1}{2}(\Ac+\Ac^T)$ has strictly positive eigenvalues. Owing to the complicated block upper-triangular structure of $\Ac$, this condition is difficult to verify directly. 

On the other hand, since $\Ac$ is block upper-triangular with diagonal blocks that are bounded below by a positive constant times the identity operator, \cite[Lemma A.1]{BeyerOliynyk:2021} guarantees the existence of a positive definite matrix $\Asc$ such that $\Asc A^0\gtrsim \mathbbm{1}$. However, to directly exploit this improvement, one would need to multiply \eqref{Fuch-ev-A-W} on the left by $\Asc$, which yields
\begin{equation*} 
\Asc A^0\del{t}W -e_K^\Lambda \Asc A^K \del{\Lambda} W =
\frac{1}{t}\Asc\Ac\Pbb W +\frac{1}{t^{1-q}} \Asc F(t,W).
\end{equation*}
The difficulty that then arises is that there is no guarantee, and indeed no reason to expect, that the matrices $\Asc A^K$ are symmetric. This symmetry is also required in order to apply the Fuchsian global existence theory from \cite{BOOS:2021,Oliynyk:CMP_2016}.

A similar situation involving competing requirements, namely, the lower bound $\Ac\gtrsim \mathbbm{1}$ on the singular matrix $\Ac$, and the symmetry of the coefficient matrices $A^0$ and $A^K$, was encountered previously in \cite{BeyerOliynykZheng:2025}, where stability in the contracting direction of a family of subcritical Kasner--scalar field spacetimes was established. In that work, a technique based on taking a sufficiently high number of spatial derivatives of the evolution equations was developed to resolve the tension between these competing requirements. Similar difficulties were also encountered in \cite{fournodavlos2020b,groeniger2023} and were again resolved by considering sufficiently high-order derivatives, although the technical details differed substantially. It is worth noting that the results of \cite{Li:2024} provide compelling evidence that the need to consider higher derivatives is intrinsic to the problem, rather than an artefact of the proofs in \cite{BeyerOliynykZheng:2025,fournodavlos2020b,groeniger2023}.

To overcome the competing requirements encountered here, we adapt the technique introduced in \cite{BeyerOliynykZheng:2025}. This approach requires repeatedly differentiating the evolution equation \eqref{Fuch-ev-A-W} with respect to the spatial variables and suitably renormalising the resulting spatial derivatives by powers of $t$. The evolution equations obtained by taking the highest number of spatial derivatives are collected together and referred to as the \textit{high-order equations}, which are treated as a PDE system. In contrast, all equations arising from lower-order derivatives are collectively referred to as the \textit{low-order equations} and are treated as ODEs. This splitting allows us to exploit the favourable structure of the matrix $\Asc \Ac$ in the low-order equations. Moreover, taking $l$ spatial derivatives improves the singular matrix $\Ac\Pbb$ to $\Ac\Pbb+l\nu A^0$ for a constant $\nu\in (0,1)$ without affecting the symmetry of the matrices $A^0$ and $A^K$. By choosing $l$ sufficiently large, we can ensure that $\Ac\Pbb+l\nu A^0\gtrsim \mathbbm{1}$. Altogether, this procedure yields a Fuchsian equation with the structure required to apply the Fuchsian global existence theory from \cite{BOOS:2021,Oliynyk:CMP_2016}. The details of this construction are presented in the following section. 

\begin{rem}
In light of similar situations encountered in \cite{BeyerOliynykZheng:2025,fournodavlos2020b,groeniger2023} and the results of \cite{Li:2024}, it would be worthwhile to determine whether the use of higher derivatives is in fact unavoidable. In the Kasner regime, anisotropy is responsible for the high-derivative requirements. However, this explanation does not apply in the ekpyrotic regime, where isotropisation occurs. Instead, the necessity of high-order derivatives appears to be related to the extreme nature of the potential.
\end{rem}

\subsection{The differentiated system}
\label{sec:differentiated_system}

\subsubsection{Low-order equations} 
We begin by letting $l\in \Nbb_{0}$ be a non-negative integer, for the moment unspecified, which determines the top-order derivatives. All derivatives of order less than $l$ are referred to as \textit{low-order}, while derivatives of order $l$ are referred to as \textit{high-order}. 

Next, we fix $\nu \in \Rbb$ satisfying 
\begin{equation} \label{eigenval-posC}
   0<\nu<1,
\end{equation}
and for a given a multiindex $\bc = (\bc_1,\ldots,\bc_{n-1})\in \Nbb_0^{n-1}$, set
\begin{equation} \label{low-order-vars}
    W_\bc = t^{|\bc|\nu}\del{}^\bc W.
\end{equation}
Note that $W=W_0$. Since the matrices $A^0$, $\Pbb$, $A^K$ and $\Ac$ are constant, c.f.~\eqref{A0-def}, \eqref{Pbb-def}, \eqref{AK-def} and \eqref{Ac-def}, applying  $t^{|\bc|\nu}\del{}^{\bc}$ to 
\eqref{Fuch-ev-A-W}, we find, with the help of the product rule, that
\begin{equation} \label{Fuch-ev-B}
  A^0 \del{t}W_\bc
  -
  e_K^\Lambda A^K \del{\Lambda}W_{\bc} 
  =
  \frac{1}{t}(\Ac\Pbb+{|\bc|\nu}A^0) W_\bc 
  +\frac{1}{t^{1-\bar q}} F_{\bc}
\end{equation}
where
\begin{equation} \label{F0-bc-def}
 F_{\bc}=
  t^{q-\bar q} t^{|\bc|\nu}\del{}^{\bc}(F(t,W))
  +t^{1-\nu-\bar q}\sum_{\substack{\bc'\leq\bc\\\bc'\neq \bc}}
  \begin{pmatrix}
  \bc\\
  \bc'
  \end{pmatrix} t^{|\bc-\bc'|\nu}(\del{}^{\bc-\bc'} e_K^\Lambda) A^K t^{(|\bc'|+1)\nu}\del{\Lambda}W_{\bc'}, 
\end{equation}
\begin{equation}
  \label{qbar}
  \bar q=\min\{q,1-\nu\},
\end{equation}
and in deriving this we employed the relation
\[|\bc|-|\bc'|-|\bc-\bc'|=\sum_{i=1}^N \bc_i-\sum_{i=1}^N \bc'_i-\sum_{i=1}^N (\bc_i-\bc'_i)=0.\]

Using Fa\'{a} di Bruno's formula, we express $\del{}^{\bc}(F(t,W))$ as
\begin{equation*}
    \del{}^{\bc}(F(t,W))=\Fc_{\bc}\bigl(t,W,(\del{}^{\bc'}W)_{\bc'\leq \bc}\bigr)
\end{equation*}
where $\Fc_{\bc}\bigl(t,W,(Y_{\bc'})_{\bc'\leq \bc}\bigr)$ is smooth in the variables $(W,(Y_{\bc'})_{\bc'\leq \bc})$ uniformly for $t\in (0,t_0]$ provided $W\in B_R(\Rbb^{\udim})$ with $R>0$ chosen as above, c.f.~\eqref{F0-smooth}. The map $\Fc_{\bc}\bigl(t,W,(Y_{\bc'})_{\bc'\leq \bc}\bigr)$ is also polynomial in the variables $(Y_{\bc'})_{\bc'\leq \bc}$  and satisfies
\begin{equation*}
c^{|\bc|}\Fc_{\bc}\bigl(t,W,(Y_{\bc'})_{\bc'\leq \bc}\bigr)=\Fc_{\bc}\bigl(t,W,(c^{|\bc'|}Y_{\bc'})_{\bc'\leq \bc}\bigr)
\end{equation*}
for all $c\geq 0$. This, in particular, allows us to express \eqref{F0-bc-def} as
\begin{equation} \label{F0-bc-form}
 F_{\bc}=
  t^{q-\bar q} \Fc_{\bc}\bigl(t,W,(t^{|\bc'|\nu}\del{}^{\bc'}W)_{\bc'\leq \bc}\bigr)
  +t^{1-\nu-\bar q}\sum_{\substack{\bc'\leq\bc\\\bc'\neq \bc}}
  \begin{pmatrix}
  \bc\\
  \bc'
  \end{pmatrix} t^{|\bc-\bc'|\nu}(\del{}^{\bc-\bc'} e_K^\Lambda) A^K t^{(|\bc'|+1)\nu}\del{\Lambda}W_{\bc'}. 
\end{equation}

For the low-order equations, where  $|\bc|<l$, we view the term $e_K^\Lambda A^K \del{\Lambda}W_{\bc}$ in \eqref{Fuch-ev-B}  as a source term and move it over to the right hand side to get
\begin{equation} \label{Fuch-ev-C}
  A^0 \del{t}W_\bc
  =
  \frac{1}{t}(\Ac\Pbb+{|\bc|\nu}A^0) W_\bc 
  +\frac{1}{t^{1-\bar q}} G_{\bc}, \quad |\bc|<l,
\end{equation}  
where
\begin{equation} \label{G0-bc-def}
 G_{\bc}=
  t^{q-\bar q} \Fc_{0,\bc}\bigl(t,W,(t^{|\bc'|\nu}\del{}^{\bc'}W)_{\bc'\leq \bc}\bigr)
  +t^{1-\nu-\bar q}\sum_{\bc'\leq\bc}
  \begin{pmatrix}
  \bc\\
  \bc'
  \end{pmatrix} t^{|\bc-\bc'|\nu}(\del{}^{\bc-\bc'} e_K^\Lambda) A^K t^{(|\bc'|+1)\nu}\del{\Lambda}W_{\bc'}.
\end{equation}
Since the matrix $\Ac$ is block-upper triangular with diagonal blocks \eqref{Ac-d-blocks} that are bounded below by a positive number times the identity operator, we know from \cite[Lemma A.1]{BeyerOliynyk:2021} that there exist positive constants $\sigma_{\af}>0$, $\af=1,2,\ldots,17$, and $\eta_0>0$, which depend on  $\ep$ and $\Rtt$, such that
\begin{equation} \label{Asc-Ac-lbnd}
\Asc \Ac \geq \eta_0 \Asc
\end{equation}
where 
\begin{align}
 \Asc &=\diag\Bigl(
    \sigma_1\delta^{\It I}\delta_{\Lambdat \Lambda},
    \sigma_2\delta^{\Lt L}\delta^{\Mt M},
    \sigma_3,
    \sigma_4\delta^{\Qt Q},\sigma_4\delta^{\Qt Q},
    \sigma_5\delta^{\Mt M}, \sigma_5\delta^{\Lt L},
    \sigma_6\delta^{\Rt R} \sigma_7\delta^{\Lt L}, \sigma_8\delta^{\It I}\delta_{\kt k} \delta^{\Jt J}, \notag \\
    &\hspace{1.3cm}\sigma_9 
    \delta^{\Rt R}\delta^{\Mt M}, \sigma_{10}
    \delta^{\Rt R} \delta^{\Lt L} \delta^{\Mt M},
    \delta^{\Qt Q}\delta^{\Lt L} \delta^{\Mt M}, \sigma_{11}\delta^{\Qt Q}\delta^{\Nt N}\delta^{\Lt L} \delta^{\Mt M}, 
    \sigma_{12}\delta^{\Qt Q} \delta^{\Mt M}, \notag  \\
    &\hspace{1.3cm}\sigma_{13}\delta^{\Qt Q}\delta^{\Nt N} \delta^{\Mt M},
    \sigma_{14}\delta^{\Qt Q} \delta^{\Lt L}, \sigma_{15}\delta^{\Qt Q}\delta^{\Nt N} \delta^{\Lt L},
    \sigma_{16}\delta^{\Qt Q}\delta^{\Nt N},\sigma_{17}\delta^{\Qt Q}\delta^{\Nt N}
    \Bigr). \label{Asc-def}
\end{align}
From \eqref{A0-def}--\eqref{Pbb-def}, \eqref{Asc-Ac-lbnd} and \eqref{Asc-def}, it is then clear that
\begin{gather}
 \Asc A^0=(\Asc A^0)^T, \label{AscA0-sym}\\
 [\Asc A^0,\Pbb]=0,\label{AscA0-Pbb-comm}
\end{gather}
and that there exists a $\eta>0$, which again depends on $\ep$ and $\Rtt$, such that
\begin{equation} \label{Asc-lbnds-A}
\Asc A^0\geq \eta \mathbbm{1} \AND \Asc \Ac \geq \eta \mathbbm{1}.
\end{equation}
This, in turn, implies by \eqref{AscA0-Pbb-comm} that
\begin{equation} \label{Asc-lbns-B}
\Asc A^0\Pbb \geq \eta\Pbb \geq 0.
\end{equation}
Finally, multiplying \eqref{Fuch-ev-C} on the left by $\Asc$ allow us to write the low-order equations as
\begin{equation} \label{Fuch-ev-D}
  \Asc A^0 \del{t}W_\bc
  =
  \frac{1}{t}(\Asc\Ac\Pbb+{|\bc|\nu}\Asc A^0) W_\bc 
  +\frac{1}{t^{1-\bar q}} \Asc G_{\bc}, \quad |\bc|<l.
\end{equation}  

\subsubsection{High-order equations}
The high-order equations are simply the equations \eqref{Fuch-ev-B} with $|\bc|=l$, which we can write as 
\begin{equation} \label{Fuch-ev-E}
  A^0 \del{t}W_\bc
  -
  e_K^\Lambda A^K \del{\Lambda}W_{\bc} 
  =
  \frac{1}{t}(\Ac\Pbb+{l\nu}A^0) W_\bc 
  +\frac{1}{t^{1-\bar q}} F_{\bc}, \quad |\bc|=l.
\end{equation}

\subsubsection{The extended system\label{sec:extended}}
We now gather the low-order and high-order equations \eqref{Fuch-ev-D} and \eqref{Fuch-ev-E} together and express them as the following single system, which we call the \textit{extended system}:
\begin{equation} \label{FuchFinal-A}
    \Bv^0 \del{t}\Wv - e_K^\Lambda\Bv^K \del{\Lambda}\Wv = \frac{1}{t}\Bcv\Pv\Wv + \frac{1}{t^{1-\bar q}}\Fv,
\end{equation}
where
\begin{align} 
    \Wv &= \bigl( W,(W_{\bc})_{|\bc|=1}, \ldots,(W_{\bc})_{|\bc|=l-1},(W_\bc)_{|\bc|=l}\bigr)^{\tr}, \label{Wv-def}\\
    \Bv^0 &= \diag\bigl(\Asc A^0,\Asc A^0,\ldots,\Asc A^0,A^0\bigr), \label{Bv0-def}\\
    \Bv^K &= \diag\bigl(0,0,\ldots,0,A^K\bigr), \label{Bv-Lambda-def}\\
    \Bcv &= \diag\Bigl(\Asc\Ac,
    \nu\Asc A^0+\Asc\Ac\Pbb,\ldots,(l-1)\nu\Asc A^0+\Asc\Ac\Pbb,{l\nu}A^0+\Ac\Pbb\Bigr),
    \label{Bcv-def} \\
    \Pv &= \diag\bigl(\Pbb,\mathbbm{1},\ldots,\mathbbm{1},\mathbbm{1}\bigr), \label{Pv-def}\\
    \label{Fv0-def}
    \Fv&=\bigl(G_{0}, (G_{\bc})_{|\bc|=1},\ldots,(G_{\bc})_{|\bc|=l-1},
      (F_{\bc})_{|\bc|=l}
    \bigr),
\end{align}
We conclude this section with the following important observations:
\begin{enumerate}[(1)] 
\item From \eqref{A0-def}, \eqref{AK-def}, \eqref{AK-sym}, \eqref{AscA0-sym}, and \eqref{Bv0-def}--\eqref{Bv-Lambda-def}, it follows that the \textit{constant} matrices $\Bv^0$ and $\Bv^K$ are symmetric, that is,
\begin{equation}\label{Bv0-BvK-sym}
\Bv^0=(\Bv^0)^T \AND \Bv^K=(\Bv^K)^T,
\end{equation}
and that there exist positive constants $\bv_0,\bv_1>0$ such that
\begin{equation} \label{Bv0-posdef}
\bv_0 \mathbbm{1}\leq \Bv^0 \leq \bv_1 \mathbbm{1}.
\end{equation}
Taken together, \eqref{Bv0-BvK-sym} and \eqref{Bv0-posdef} imply that the extended system \eqref{FuchFinal-A} is symmetric hyperbolic.

\item Since $A^0$ is positive definite, see \eqref{A0-def}, $\Ac$ and $\Pbb$ are constant matrices, see \eqref{Pbb-def}, \eqref{Ac-def} and \eqref{Ac-blocks}, and $\nu \in (0,1)$, it follows that for any given $\kappa_0>0$ there exists an $l\in \Nbb$ such that
\begin{equation*}
l\nu A^0 + \Ac \Pbb \geq \kappa_0 \mathbbm{1}.
\end{equation*}
Combining this estimate with \eqref{Bv0-posdef}, \eqref{Asc-lbnds-A}--\eqref{Asc-lbns-B}, and \eqref{Bcv-def}, we conclude that there exist an $l\in \Nbb$ and a constant $\kappa>0$ such that
\begin{equation} \label{eq:Bvpositivity}
\Bcv \geq \kappa \Bv^0.
\end{equation}
\item
From \eqref{AscA0-Pbb-comm}, \eqref{Pv-def}, \eqref{Bv0-def} and \eqref{Bcv-def}, we note that
\begin{equation}\label{Bv0-Pv}
\Pv \Bv^0 \Pv^\perp=\Pv^\perp \Bv^0 \Pv=0,
\end{equation}
 where $\Pv^\perp=\mathbbm{1}-\Pv$, and
\begin{equation}\label{Bsc-comm}
[\Pv,\Bsc]=0.
\end{equation}

\item The vector-valued map $\Fv=\Fv(t,\Wv)$ is smooth in $\Wv$ uniformly for $t\in (0,t_0]$ provided $W\in B_R(\Rbb^{\udim})$ with $R>0$ chosen as above, c.f.~\eqref{F0-smooth}. Furthermore, $\Fv$ satisfies
\begin{equation} \label{F-vanish}
\Fv(t,0)=0.
\end{equation}

\item We consider solutions to \eqref{FuchFinal-A} that are not obtained from differentiating \eqref{Fuch-ev-A-W}. In this case,  the relations \eqref{low-order-vars} do not hold unless the are satisfied by the initial data.  
\end{enumerate}

\section{The main theorem}
\label{sec:main_theorem}

\subsection{Statement and discussion}
We are now ready to state and prove the main result of this article, namely the past global stability of positive ekpyrotic-FLRW solutions. The precise formulation of this stability statement is given in the theorem below, with the proof deferred to Section \ref{sec:proof_globstab}. The proof builds on a preliminary stability result for the trivial solution of the Fuchsian formulation of the reduced conformal Einstein–scalar field equations introduced in the previous section. This auxiliary result is stated and proved separately in Section \ref{sec:proof_globstab_Fuchsian}.

\begin{thm}[Past global stability of the positive ekpyrotic-FLRW solution ]\label{glob-stab-thm}
Suppose $t_0>0$, $0<\ep<\min\{\Rtt,1-\Rtt\}$,  $n\in\Zbb_{\ge 3}$,  $k \in \Zbb_{>\frac{n+3}{2}}$, $V_0=-1$, $s>s_c$ where $s_c$ is defined by \eqref{sc-def},  and $\Rtt$, $\Phi_0$, $\Ltt$, $\Lambda$ and $\alpha$ chosen according to \eqref{eq:ekpyroticparameterconstants.N}.  
Then there 
exists a $l\ge \Nbb_{0}$ and $\delta_0>0$ such that for every $\delta\in (0,\delta_0]$, $\gr_{\mu\nu}\in H^{k+l+2}(\Tbb^{n-1},\mathbb{S}_n)$, 
$\ggr_{\mu\nu}\in H^{k+l+1}(\Tbb^{n-1},\mathbb{S}_n)$, $\taur=t_0$ and
$\taugr\in H^{k+l+2}(\Tbb^{n-1})$ satisfying 
\begin{equation} 
    \norm{\gr_{\mu\nu}-\eta_{\mu\nu}}_{H^{k+l+2}(\Tbb^{n-1})}
  +\norm{\ggr_{\mu\nu}}_{H^{k+l+1}(\Tbb^{n-1})}
  +\norm{\taugr-1}_{H^{k+l+2}(\Tbb^{n-1})}  
  <\delta \label{glob-stab-thm-idata-A}
\end{equation}
as well as the gravitational and wave gauge constraints \eqref{grav-constr}--\eqref{wave-constr}, 
there exists a unique classical solution $\Wsc\in C^1(M_{0,t_0})$, see \eqref{eq:Wdef}, 
of the Lagrangian formulation of the reduced Einstein-scalar field equations \eqref{tconf-ford-C.1}--\eqref{tconf-ford-C.8} on $M_{0,t_0}=(0,t_0]\times \Tbb^{n-1}$ 
that satisfies the corresponding initial
conditions \eqref{l-idata}--\eqref{hhu-idata} on $\Sigma_{t_0}=\{t_0\}\times\Tbb^{n-1}$ and
the constraints \eqref{eq:Lag-constraints} in $M_{0,t_0}$. Additionally,
\begin{equation}
  \label{eq:WregFin}
  \Wsc \in \bigcap_{j=0}^{k}C^j\bigl((0,t_0], H^{k-j}(\Tbb^{n-1})\bigr),
\end{equation} 
and the pair $\{g_{\mu\nu}=\del{\mu}l^\alpha\ghu_{\alpha\beta}\del{\nu}l^\beta,\tau=t\}$ determined by $\Wsc$ defines a solution of the conformal Einstein-scalar field equations \eqref{lag-confeqns} on $M_{0,t_0}$ that satisfies the wave gauge constraint \eqref{lag-wave-gauge}. Furthermore, the pair
\begin{equation}
  \label{eq:conf2phys.thm}
  \biggl\{\gb_{ij}=e^{2\Phi}g_{ij},\,\phi=\frac{2(n-1)s}{(n-1)s^2-s_c^2}\log t\biggr\},
\end{equation}
where
\begin{equation}
  \label{eq:conffactorekpy.thm}
  e^\Phi=\frac{\sqrt{2}(n-1)\sqrt{s^2-s_c^2}}{(n-1)s^2-s_c^2}t^{\frac{1-\Rtt}{n-2}},
\end{equation}
determines a solution to the \textit{physical Einstein-scalar field equations} \eqref{ESF.1}--\eqref{ESF.2} with potential \eqref{eq:exponentialpotential} on $M_{0,t_0}$, and there exists a $\kappat>0$ such the following hold: 

\begin{enumerate}[(a)]
  \item Let $e_0^\mu=\beta^{-1}\delta_0^\mu$ with $\beta= (-g(dx^0,dx^0))^{-\frac{1}{2}}$, and $e^\mu_I$ be the unique solution of the Fermi-Walker transport equations \eqref{Fermi-A} with initial conditions $e^\mu_I |_{\Sigma_{t_0}}=\delta^\mu_\Lambda\er^\Lambda_I$ where the functions $\er^\Lambda_I \in H^k(\Tbb^{n-1})$ are chosen to satisfy $\norm{\er^\Lambda_I-\delta^\Lambda_I}_{H^{k}(\Tbb^{n-1})} < \delta$
  and make the frame $e^\mu_i$ orthonormal on $\Sigma_{t_0}$ with respect to the conformal metric $\gr_{\mu\nu}$.
  Then $e^\mu_i$ is a well-defined frame in $M_{0,t_0}$ satisfying
  $e^0_I=0$ and $g_{ij}=\eta_{ij}$,
  where $g_{ij}=e_i^\mu g_{\mu\nu}e^\nu_j$ is the
  frame metric. Furthermore, there exist $\ef_I^\Lambda\in H^k (\mathbb T^{n-1})$ such that
  \begin{equation}
    \norm{e_0^\mu - \delta_0^\mu}_{H^{k}(\mathbb T^{n-1})} 
    + \norm{e_I^\Lambda - (\delta_I^\Lambda+\ef_I^\Lambda)}_{H^{k}(\mathbb T^{n-1})} \lesssim t^{\tilde\kappa}\label{frameconv.thm}
  \end{equation}
  for $0<t\leq t_0$. 
  
  \item The frame\footnote{We refer to this frame $\eb_{\ib}^\mu$ as the \emph{physical orthonormal frame} to distinguish it from the  frame $e_i$, and we employ over-bar lower case Latin letters, e.g.,  $\ib,\jb,\bar{k}$, that run from $0$ to $n-1$ to index components relative to this frame. Furthermore, we use over-bar upper case Latin letter, e.g.,~ $\Ib,\Jb,\bar{K}$ that run from $1$ to $n-1$ to index spatial components relative to the physical spatial orthonormal frame  $\eb_{\Ib}=\eb_{\Ib}^\mu\del{\mu}$.} $\eb_{\ib}^\mu=e^{-\Phi} e_i^\mu \delta_{\ib}^i$, see \eqref{eq:conf2phys.thm}--\eqref{eq:conffactorekpy.thm}, is orthonormal with respect to the physical metric $\gb_{\mu\nu}$, that is,
  \begin{equation*}
   \gb_{\ib\jb}=\gb_{\mu\nu}\eb^{\mu}_{\ib}\eb^{\nu}_{\jb} = \eta_{\ib\jb},
  \end{equation*}
  and satisfies
  \begin{align}
    &\Biggl\| t^{\frac{1-\Rtt}{n-2}} \eb_{\overline 0}^\mu - \frac{(n-1)s^2-s_c^2}{\sqrt{2}(n-1)\sqrt{s^2-s_c^2}}\delta_{\overline 0}^\mu\Biggr\|_{H^{k}(\mathbb T^{n-1})} \lesssim t^{\tilde\kappa},\label{frameconv.phys.thm.1}\\
    &\Biggl\|t^{\frac{1-\Rtt}{n-2}} e_{\Ib}^\Lambda - \frac{(n-1)s^2-s_c^2}{\sqrt{2}(n-1)\sqrt{s^2-s_c^2}}(\delta_{\Ib}^\Lambda+\ef_{\Ib}^\Lambda)\Biggr\|_{H^{k}(\mathbb T^{n-1})} \lesssim t^{\tilde\kappa},\label{frameconv.phys.thm.2}
  \end{align}
  for $0<t\leq t_0$. 

  \item 
With respect to the conformal metric $g_{\mu\nu}$, the future pointing unit normal to the foliation $M_{0,t_0}=\cup_{0<t\leq t_0}\Sigma_t$ is given by $n^\mu=e_0^\mu$. The shift vector associated with this foliation vanishes, while the conformal lapse $\Ntt$ satisfies 
\begin{equation}
    \norm{\Ntt-1}_{H^{k}(\Tbb^{n-1})}\lesssim t^{\Rtt-\ep+\tilde\kappa}\label{lapseconv.thm}
  \end{equation}
for $0<t\leq t_0$.
Furthermore, the spatial frame $e_I=e^\Lambda_{I}\del{\Lambda}$ is orthonormal with respect to the spatial metric $\cspatmetr_{\Lambda\Gamma}$ induced on the leaves $\Sigma_t$ by $g_{\mu\nu}$, that is,
 \begin{equation}
    \label{eq:cspatmetr.thm}
    \cspatmetr_{IJ}=\cspatmetr_{\Lambda\Gamma}e^\Lambda_I e^\Gamma_J=\delta_{IJ}.
  \end{equation}
There also exists a symmetric tensor field $\kf_{M}{}_{L}\in H^{k}(\mathbb T^{n-1})$ such that the conformal shear $\sigma_{IJ}$ and conformal expansion $\Htt$ satisfy 
  \begin{equation}  \biggl\| t^{1-\Rtt} \cshear_L{}{}^J-\Bigl(\kf_{L}{}^{J}-\frac{1}{n-1}\kf_{M}{}^{M}\gtt_{L}{}^{J}\Bigr)\biggr\|_{H^{k}(\mathbb T^{n-1})} 
    + \biggl\| t^{1-\Rtt} \cexp-\frac 1{n-1}\kf_{M}{}^{M}\biggr\|_{H^{k}(\mathbb T^{n-1})} 
    \lesssim t^{\kappat}+t^{\Rtt-\ep+\tilde\kappa}
    \label{eq:confHshear.thm}
  \end{equation}
for $0<t\leq t_0$. Here
 \begin{equation*}
    \Ktt_{IJ}=\sigma_{IJ}+H \cspatmetr_{IJ}
  \end{equation*}
denotes the standard decomposition of the second fundamental form $\Ktt_{IJ}$ of the foliation $M_{0,t_0}=\cup_{0<t\leq t_0}\Sigma_t$  into the conformal shear tensor $\sigma_{IJ}$ and the conformal expansion scalar $\Htt$. Recall that $\Rtt\in (0,1)$ as a consequence of \eqref{eq:ekpyroticparameterconstants.N} with $s>s_c$. 

\item With respect to the physical metric $\gb_{\mu\nu}$ (see \eqref{eq:conf2phys.thm}), the future pointing unit normal of the foliation $M_{0,t_0}=\cup_{0<t\leq t_0}\Sigma_t$  is $\nb^\mu=e^{-\Phi} n^\mu$. The shift vector associated with this foliation vanishes, while the  
  physical lapse $\bar\Ntt=e^\Phi\Ntt$ satisfies
  \begin{equation}
    \Bnorm{t^{-\frac{1-\Rtt}{n-2}}\bar\Ntt-\frac{\sqrt{2}(n-1)\sqrt{s^2-s_c^2}}{(n-1)s^2-s_c^2}}_{H^{k}(\Tbb^{n-1})}\lesssim t^{\Rtt-\ep+\tilde\kappa}\label{lapseconv.phys.thm}
  \end{equation}
  for $0<t\leq t_0$.
  Furthermore, the spatial metric $\pspatmetr_{\Lambda\Gamma}$ induced on the leaves $\Sigma_t$ by $\gb_{\mu\nu}$ is given by
  \begin{equation*}
    %\label{eq:spatmetrrel.thm}
    \pspatmetr_{\Lambda\Gamma}=e^{2\Phi}\cspatmetr_{\Lambda\Gamma}
  \end{equation*}
  and the spatial frame $\eb^\Lambda_{\Ib}=e^{-\Phi}e^{\Lambda}_{I}\delta^{I}_{\Ib}$
  is orthonormal with respect $\pspatmetr_{\Lambda\Gamma}$, that is,
  \begin{equation*}
      \pspatmetr_{\Ib\Jb}= \pspatmetr_{\Lambda\Gamma}\eb^\Lambda_{\Ib}\eb^\Gamma_{\Jb} = \delta_{\Ib\Jb}.
  \end{equation*}
Letting
 \begin{equation*}
    \bar\Ktt_{\Ib\Jb}=\bar\sigma_{\Ib\Jb}+
    \bar\Htt\bar\cspatmetr_{\Ib\Jb}
  \end{equation*}
denotes the standard decomposition of the second fundamental form $\bar\Ktt_{\Ib\Jb}$ of the foliation $M_{0,t_0}=\cup_{0<t\leq t_0}\Sigma_t$  into the shear tensor $\bar\sigma_{\Ib\Jb}$ and the expansion scalar $\bar\Htt$, then  
  \begin{align}
    &\biggl\|t^{\frac{n-1-\Rtt}{n-2}} \pexp
    -{\scriptstyle  \frac{s_c^2}{\sqrt{2}(n-1)\sqrt{s^2-s_c^2}}} 
   \biggr\|_{H^{k}(\mathbb T^{n-1})} 
    \lesssim t^{\Rtt}+t^{\Rtt-\ep+\kappat}
    ,\label{eq:pexpbound.thm}\\
    &\biggl\|t^{\frac{n-1}{n-2}(1-\Rtt)}\Ntt\bar\sigma_{\Lb}{}{}^{\Jb}
    -{\scriptstyle \frac{(n-1)s^2-s_c^2}{\sqrt{2}(n-1)\sqrt{s^2-s_c^2}}}\Bigl(\kf_{L}{}^{J}-\frac{1}{n-1}\kf_{M}{}^{M}\gtt_{L}{}^{J}\Bigr)\delta^L_{\Lb}\delta_J^{\Jb}\biggr\|_{H^{k}(\mathbb T^{n-1})} 
    \lesssim t^{\kappat}+t^{\Rtt-\ep+\kappat}
    ,\label{eq:pshearbound.thm}
  \end{align}
  for $0<t\leq t_0$, while the
  \emph{physical scale invariant shear tensor} $\bar\Sigma_{\ib \jb}=\bar\sigma_{\ib \jb}/{\bar\Htt}$ satisfies
  \begin{equation}
    \label{eq:pshearrescbound.thm}
    \Bnorm{t^{-\Rtt}\bar\Sigma_{\Lb}{}^{\Jb}-\frac{(n-1)s^2-s_c^2}{s_c^2}\Bigl(\kf_{L}{}^{J}-\frac{1}{n-1}\kf_{M}{}^{M}\gtt_{L}{}^{J}\Bigr)\delta^L_{\Lb}\delta_J^{\Jb}}_{H^{k}(\mathbb T^{n-1})}
    \lesssim {t^{\kappat}
    +t^{\Rtt} }
  \end{equation}
  for $0<t\leq t_0$. Due to the uniform blow-up of
  the physical expansion scalar $\bar\Htt$ as $t\searrow 0$, the hypersurface $\Sigma_{0}=\{0\}\times \Tbb^{n-1}$  is a \emph{crushing singularity} \cite{eardley1979} and past directed timelike geodesics are incomplete.
  \item The spatial Riemann curvature tensor is asymptotically bounded by $\pexp^2$ in the following sense 
  \begin{equation}
    \label{eq:spatcurvaturedecay.thm}
    \biggl\|\frac{\pspatcurv_{\Ib\Jb\bar K}{}^{\Mb}}{\pexp^2}\biggr\|_{H^{k-1}(\mathbb T^{n-1})}\lesssim t^{\Rtt}
  \end{equation}
  for $0<t\leq t_0$.
  \item \label{thmitem:Ricci} The physical Ricci curvature tensor satisfies
  \begin{align}
    \label{eq:Ricciblowup.thm.1}
    \biggl\|\frac{\Rb_{\ib \bar k}\pnormal^{\ib}\pnormal^{\bar k}}{\pexp^2}-(n-1)\frac{(n-1)s^2-s_c^2}{s_c^2}\biggr\|_{H^{k}(\mathbb T^{n-1})}
    &\lesssim t^{\Rtt-\ep+\kappat} +t^{\Rtt},\\
    \label{eq:Ricciblowup.thm.2}
    \biggl\|\frac{\Rb_{\ib \bar k}\pspatmetr_{\jb}{}^{\ib}\pspatmetr_{\lb}{}^{\bar k}}{\pexp^2}
    +\frac{(n-1)(s^2-s_c^2)}{s_c^2}\pspatmetr_{\jb\lb}\biggr\|_{H^{k}(\mathbb T^{n-1})}
    &\lesssim
    t^{\Rtt-\ep+\kappat} +t^{\Rtt},
  \end{align}
  for $0<t\leq t_0$, while $\Rb_{\ib \jb}\nb^{\jb}\gtt^{\ib}{}_{\bar{k}}$ vanishes identically. In particular, these estimates imply that $\Rb_{\ib\jb}\Rb^{\ib\jb}\approx \bar{\Htt}^4 \approx t^{ -\frac{4(n-1-\Rtt)}{n-2}}$ near $t=0$. As a consequence, the solution \eqref{eq:conf2phys.thm} is $C^2$-inextendable at $t=0$ and the hypersurface $\Sigma_{0}=\{0\}\times \Tbb^{n-1}$ is a big bang singularity. 
  \item \label{thmitem:Weyl} The physical Weyl tensor $ \Wb_{{\ib}{\jb}{\bar k}{\bar l}}$ admits the (unique) decomposition 
  \begin{align}
    \label{eq:Weylrepresentation.thm}
    \Wb_{{\ib}{\jb}{\bar k}{\bar l}}=&\nb_{\ib}\pWeyl_{{\bar k}{\bar l}{\jb}}
    -\nb_{\jb}\pWeyl_{{\bar k}{\bar l}{\ib}}
    +\nb_{\bar k}\pWeyl_{{\ib}{\jb}{\bar l}}
    -\nb_{\bar l}\pWeyl_{{\ib}{\jb}{\bar k}}\\
    &+\nb_{\ib}\nb_{\bar k} \pWeyl_{{\bar m}{\jb}{\bar p}{\bar l}}\pspatmetr^{{\bar m}{\bar p}} 
    +\nb_{\ib}\nb_{\bar l} \pWeyl_{{\bar m}{\jb}{\bar k}{\bar p}}\pspatmetr^{{\bar m}{\bar p}} 
    +\nb_{\jb}\nb_{\bar k} \pWeyl_{{\ib}{\bar m}{\bar p}{\bar l}}\pspatmetr^{{\bar m}{\bar p}}
    +\nb_{\jb}\nb_{\bar l} \pWeyl_{{\ib}{\bar m}{\bar k}{\bar p}}\pspatmetr^{{\bar m}{\bar p}}
    +\pWeyl_{{\ib}{\jb}{\bar k}{\bar l}}\notag
  \end{align}
  where $\pWeyl_{{\bar k}{\bar l}{\jb}}$ and $\pWeyl_{{\ib}{\bar k}{\bar l}{\jb}}$ are purely spatial\footnote{That is, $\pWeyl_{{\bar k}{\bar l}{\jb}}$ and $\pWeyl_{{\ib}{\bar k}{\bar l}{\jb}}$ vanish whenever $\ib$, $\bar{k}$, $\lb$ or $\jb$ equal $0$.} and satisfy
  \begin{equation}
    \label{eq:Weylestimate.thm.1}
    \biggl\|\frac{\pWeyl_{\Mb\Nb\Pb}}{\pexp}\biggr\|_{H^{k-1}(\mathbb T^{n-1})}\lesssim t^{\Rtt}+t^{\Rtt-\ep+\kappat},
  \end{equation}
  and 
  \begin{equation}
    \label{eq:Weylestimate.thm.2}
    \biggl\|\frac{\pWeyl_{\Mb\Nb\Pb\Qb}}{\pexp^2}\biggr\|_{H^{k-1}(\mathbb T^{n-1})}\lesssim t^{\Rtt}+t^{\Rtt-\ep+\kappat},
  \end{equation}
  for $0<t\leq t_0$. 
  These estimates together with \eqref{eq:Ricciblowup.thm.1}--\eqref{eq:Ricciblowup.thm.2} imply that the big bang singularity is Ricci dominated in the sense of \cite[equation~(2.11)]{goode1985}. 
  \item The fields $\phi_0$ and $\phi_1$ defined as
  \begin{equation*}
    \phi_1=\frac{\pnormal^i\nablab_i \phi}{(n-1)\pexp} \AND
    \phi_0=\phi+\frac{s_c^2}{s^2}\phi_1\log((n-1){\bar\Htt}),
  \end{equation*}
  respectively, satisfy
  \begin{equation}
    \label{eq:phi_1limit.thm}
    \biggl\| \phi_1-\frac{2s}{s_c^2}\biggr\|_{H^{k}(\mathbb T^{n-1})} 
    \lesssim t^{\Rtt},
  \end{equation}
  and
  \begin{equation}
    \label{eq:phi_0limit.thm}
    \biggl\|\phi_0
    -\frac{1}{s}\log\left(\frac{s_c^4}{2(s^2-s_c^2)}\right)
    \biggr\|_{H^{k}(\mathbb T^{n-1})} 
    \lesssim t^\Rtt(1+\log t) +t^{\Rtt-\ep+\kappat},
  \end{equation}
  for $0<t\leq t_0$.
\item 
The physical acceleration $\pacc$ determined by the foliation $M_{0,t_0}=\cup_{0<t\leq t_0}\Sigma_t$, c.f.~\eqref{defphysextr}, is purely spatial and satisfies
\begin{equation} \label{ndot-bar-est}
\biggl\|\frac{\pacc_{\Jb}}{\bar\Htt}\biggr\|_{H^k(\Tbb^{n-1})}
\lesssim  t^{1-\ep+\kappat}
\end{equation}
for $0<t\leq t_0$.
\item The solution is AVTD in the sense of \cite[\S1.3]{BeyerOliynyk:2021}, in that it satisfies, up to an error term that is integrable in time over $(0,T_0]$, the \textit{velocity term dominated} (VTD) equations obtained from the evolution system %\eqref{FuchFinal-A} 
\eqref{Fuch-ev-A-W}
by removing all explicit spatial derivative terms.
\end{enumerate}
\end{thm}

\subsection{Fuchsian stability}
\label{sec:proof_globstab_Fuchsian}
The first step in the proof of Theorem~\ref{glob-stab-thm} is to establish the past stability of the trivial solution $W \equiv 0$ to the Fuchsian equation \eqref{Fuch-ev-A-W}, which by Remark \ref{rem:peFLRW} and \eqref{eq:ekpyroticbeta} corresponds to the positive ekpyrotic-FLRW solution. More precisely, in Proposition \ref{prop:globalstability} below we establish the existence of solutions to \eqref{Fuch-ev-A-W} on $M_{0,t_0}$ that satisfy the initial condition
$W(t_0)=W_0$
for sufficiently small initial data $W_0$. The proof of this proposition follows from a straightforward application of the Fuchsian global existence theory developed in \cite{BOOS:2021,Oliynyk:CMP_2016} to the extended system \eqref{FuchFinal-A}.

\begin{prop}
  \label{prop:globalstability}
  Suppose $n\in\Zbb_{\ge 3}$,  $k \in \Zbb_{>\frac{n+1}{2}}$, $\Rtt\in (0,1)$, $t_0>0$, $\{\ep,\nu\}$ satisfy \eqref{epcond} and \eqref{eigenval-posC}, and let $\Pbb^\perp = \mathbbm{1} -\Pbb$.
  Then there exists $l\in \Nbb_0$ and 
  a constant $\delta_0 > 0$ such that for every $\delta \in (0,\delta_0]$ and $W_0\in H^{k+l}(\mathbb T^{n-1})$ with
  \begin{equation}
    \label{glob-stab-thm-idata-A-Fuchsian.Cor}
    \norm{W_0}_{H^{k+l}}(\mathbb T^{n-1})< \delta,
  \end{equation}
  the Fuchsian equation \eqref{Fuch-ev-A-W} with initial condition $W(t_0)=W_0$ admits a unique past global solution
  \begin{equation*}
    W \in C^0\bigl((0,t_0],H^{k+l}(\mathbb T^{n-1})\bigr)\cap C^1\bigl((0,t_0],H^{k+l-1}(\mathbb T^{n-1})\bigr).
  \end{equation*}
  Moreover, 
  \begin{enumerate}[(a)] 
  \item the limit $\Pbb^\perp W(0):=\lim_{t\searrow 0} \Pbb^\perp W(t)$ exists in $H^{k}(\mathbb T^{n-1})$, 
  \item $W$ satisfies the energy estimate
  \begin{equation}
    \sum_{|\bc|\le l}\norm{t^{|\bc|\nu} \del{}^\bc W(t)}_{H^k(\mathbb T^{n-1})}^2 
    + \int^{t_0}_t \frac{1}{s} \sum_{|\bc|\le l}\norm{s^{|\bc|\nu} \del{}^\bc\Pbb W(s)}_{H^k(\mathbb T^{n-1})}^2\, ds
    \lesssim \norm{W_0}^2_{H^{k+l}(\mathbb T^{n-1})}
      \label{eq:resenergy.Cor}
  \end{equation}
  for $0<t\leq t_0$,
  \item and there exists $\tilde\kappa>0$ such that the decay estimates
  \begin{align}
  \label{eq:PbbuDecay.Cor}
    \norm{\Pbb W(t)}_{H^{k}(\mathbb T^{n-1})} &\lesssim t^{\tilde\kappa},\\
    \label{eq:PbbuDecay2.Cor}
    \norm{\Pbb^\perp W(t) - \Pbb^\perp W(0)}_{H^{k}(\mathbb T^{n-1})} &\lesssim t^{\tilde\kappa},\\
    \label{eq:PbbuDecay3.Cor}
    \norm{\del{}^\bc W(t)}_{H^{k}(\mathbb T^{n-1})} &\lesssim t^{\tilde\kappa-|\bc|\nu},\\
   \label{eq:PbbuDecay4.Cor}
   \norm{\del{}^{\tilde\bc} W(t)}_{H^{k-1}(\mathbb T^{n-1})} &\lesssim t^{\tilde\kappa-l\nu},
  \end{align}
  hold for $0<t\leq t_0$ and $1\le|\bc|\le l-1$.
\end{enumerate}
  \medskip

  \noindent The implicit constants in the energy and decay estimates are all independent of the choice of $\delta\in (0,\delta_0]$.
\end{prop}
\begin{proof}
Suppose $W_0\in H^{k+l}(\Tbb^{n-1})$ satisfies
\begin{equation*}
\norm{W_0}_{H^{k+l}(\Tbb^{n-1})}<\delta. 
\end{equation*}
Then, since \eqref{Fuch-ev-A-W} is symmetric hyperbolic, we know from standard existence theory for symmetric hyperbolic equations, see \cite[Thm.~10.1]{BenzoniSerre:2007} or \cite[Thm.~2.1]{Majda:1984}, that there exists a $t_1\in [0,t_0)$ and a unique solution 
\begin{equation*}
  W\in C^0((t_1,t_0], H^{k+l}(\Tbb^{n-1}))\cap C^1((t_1,t_0], H^{k+l-1}(\Tbb^{n-1}))
\end{equation*}
of \eqref{Fuch-ev-A-W} satisfying the initial condition $W(t_0)=W_0$. Moreover, if $t_1\neq 0$ and
\begin{equation} \label{prop:globalstability-continuation}
  \sup_{t_1<t<t_0}\norm{W(t)}_{W^{1,\infty}(\Tbb^{n-1})}< R,  
\end{equation}
where $R>0$ is as defined in comment (4) from Section \ref{sec:extended}, then this solution can be uniquely continued, see \cite[Thm.~10.3]{BenzoniSerre:2007} or \cite[Thm.~2.2]{Majda:1984}, to a larger time interval $(t_*,t_0]$ where $t_*\in [0,t_1)$.

On the other hand, $W$ determines via \eqref{low-order-vars} and \eqref{Wv-def} a solution
\begin{equation*}
  \Wv\in C^0((t_1,t_0], H^k(\Tbb^{n-1}))\cap C^1((t_1,t_0], H^{k-1}(\Tbb^{n-1})),
\end{equation*}
where
\begin{equation}\label{prop:globalstability-Wv}
   \Wv = \bigl( W,(t^{|\bc|\nu}\del{}^{\bc}W)_{|\bc|=1}, \ldots,(t^{|\bc|\nu}\del{}^{\bc}W)_{|\bc|=l}\bigr)^{\tr},
\end{equation}
of the extended system \eqref{FuchFinal-A}
satisfying the initial condition
\begin{equation*}
\Wv(t_0)=\Wv_0 := \bigl(W_0,(\del{}^\bc W_0)_{|\bc|=1},\cdots,(\del{}^\bc W_0)_{|\bc|=l}\bigr)^T.
\end{equation*}
Note also that
\begin{equation*} %\label{prop:globalstability-idata}
\norm{\Wv(t_0)}_{H^k(\Tbb^{n-1})}\lesssim  \delta.
\end{equation*}

Next, we choose $l\in \Nbb_0$ sufficiently large so that \eqref{eq:Bvpositivity} holds for some $\kappa>0$. Since $\nu\in (0,1)$ and $q>0$ by \eqref{q-pos}-\eqref{epcond}, we see from \eqref{qbar} that $\qb \in (0,1)$.
Moreover, from \eqref{Pv-def} we observe that the projection matrix $\Pv$ satisfies
\begin{equation*}
\Pv^2 = \Pv,  \quad  \Pv^{\tr} = \Pv, \quad \del{t}\Pv =0 \AND \del{\Lambda} \Pv =0.
\end{equation*}
Using these identities together with \eqref{Bv0-BvK-sym}-\eqref{F-vanish}, it is straightforward to verify that there exists a $R>0$, see comment (4) from Section \ref{sec:extended} above, such that the extended system \eqref{FuchFinal-A} satisfies all coefficient assumptions from \cite[\S3.4]{BOOS:2021} for the parameter values\footnote{The exact values of the other parameters, i.e.,~$\theta$, $\gamma_0$, and $\gamma_1$, from \cite[\S3.4]{BOOS:2021} are not needed, only that they are positive and finite.} $p=\qb$, $\kappa$ as above, $\alpha=0$, $\lambda_1=\lambda_2=\lambda_3=0$, and $\beta_{\af}=0$, $\af=0,1,\ldots,7$. Thus using the time transformation from \cite[\S3.4]{BOOS:2021}, we can apply \cite[Thm.~3.8]{BOOS:2021} or alternatively \cite[Thm.~B.1]{Oliynyk:CMP_2016} to the extended system \eqref{FuchFinal-A} to deduce the existence of a $\delta_0>0$ such that, for any $\delta\in (0,\delta_0]$, there exists a unique solution\footnote{The improvement in regularity in \eqref{Wvt-sol} over the regularity from a direct application of \cite[Thm.~3.8]{BOOS:2021} is a consequence of \cite[Rem.~A.3.(ii)]{BeyerOliynyk:2020}. } 
\begin{equation} \label{Wvt-sol}
  \tilde{\Wv}\in C^0((0,t_0], H^k(\Tbb^{n-1}))\cap C^1((t_1,t_0], H^{k-1}(\Tbb^{n-1}))
\end{equation}
to \eqref{FuchFinal-A} satisfying $\tilde{\Wv}(0)=W_0$ for which the limit $\Pv^\perp \tilde{W}(0):=\lim_{t\searrow 0}\Pv^\perp \tilde{W}(t)$ exists in $H^{k-1}(\Tbb^{n-1})$ where $\Pv^\perp=\mathbbm{1}-\Pv$. Moreover, $\tilde{\Wv}$ satisfies the energy estimates
\begin{equation}
  \norm{\tilde{\Wv}(t)}_{H^k(\mathbb T^{n-1})}^2 + \int^{t_0}_t \frac{1}{s} \norm{\Pv \tilde{\Wv}(s)}_{H^k(\mathbb T^{n-1})}^2\, ds
  \lesssim \norm{\Wv_0}^2_{H^k(\mathbb T^{n-1})},
    \label{eq:resenergy.D}
\end{equation}
for $0<t\leq t_0$, and there exists a $\kappat>0$
such that the decay estimates
\begin{align}
  \label{eq:PbbuDecay.D}
  \norm{\Pv \tilde{\Wv}(t)}_{H^{k-1}(\mathbb T^{n-1})} &\lesssim t^{\tilde\kappa},\\
  \label{eq:PbbuDecay2.D}
  \norm{\Pv^\perp \tilde{\Wv}(t) - \Pv^\perp \tilde{\Wv}(0)}_{H^{k-1}(\mathbb T^{n-1})} &\lesssim
  t^{\tilde\kappa},
\end{align}
hold for $0<t\leq t_0$. Noting from \eqref{FuchFinal-A} that no spatial derivatives appear in the equation for the components\footnote{That is, $\Wt_{\bc}$ satisfies an equation of the form $A^0\del{t}\Wt_{\bc} = \frac{1}{t}(\Ac\Pb+|\bc|\nu A^0) \Wt_{\bc} + \frac{1}{t^{1-\qb}} G_{\bc}(t,\tilde{\Wv})$ for $|\bc|<l$, c.f.~\eqref{Fuch-ev-C}.} $\Wt_{\bc}$, $|\bc|<l$, from
\begin{equation*}
 \tilde{\Wv} = \bigl( \Wt,(\Wt_{\bc})_{|\bc|=1}, \ldots,(\Wt_{\bc})_{|\bc|=l-1},(\Wt_{\bc})_{|\bc|=l}\bigr)^{\tr},
\end{equation*}
it is not difficult to see from the proof of the decay estimates from \cite[Thm.~3.8]{BOOS:2021} or alternatively \cite[Thm.~B.1]{Oliynyk:CMP_2016} that the decay estimates for $\Wt_{\bc}$, $|\bc|<l$, improve to
\begin{align}
  \label{eq:PbbuDecay.E}
  \norm{\Pbb \Wt(t)}_{H^{k}(\mathbb T^{n-1})} &\lesssim t^{\tilde\kappa},\\
  \label{eq:PbbuDecay2.E}
  \norm{\Pbb^\perp \Wt(t) - \Pbb^\perp \Wt(0)}_{H^{k}(\mathbb T^{n-1})} &\lesssim
  t^{\tilde\kappa},\\
   \label{eq:PbbuDecay3.E}
  \norm{\Wt_{\bc}(t)}_{H^{k}(\mathbb T^{n-1})}  &\lesssim t^{\tilde\kappa}, \quad 1\leq |\bc|<l,
\end{align}
for $0<t\leq t_0$.

Now, from the uniqueness of solutions to the extended system \eqref{FuchFinal-A}, we deduce that $\Wv(t)=\tilde{\Wv}(t)$ for $t\in (t_1,t_0]$. By shrinking $\delta_0>0$ if necessary, this result in conjunction with the energy estimate \eqref{eq:resenergy.D} and Sobolev's inequality  \cite[Ch.~13, Prop.~2.4]{taylor2011} implies that \eqref{prop:globalstability-continuation} holds and that we can continue the solution $\Wv$ to the time interval $(t_*,t_1]$, $t_*\in [0,t_1)$, while satisfying   
$\sup_{t_1<t<t_0}\norm{W(t)}_{W^{1,\infty}(\Tbb^{n-1})}< R/2$. Thus, we conclude that $t_*=t_1=0$ and
\begin{equation*}
\Wv(t)=\tilde{\Wv}(t)
\end{equation*}
for $0<t\leq t_0$. The proof now follows directly from this relation, \eqref{prop:globalstability-Wv} and the estimates  \eqref{eq:resenergy.D}, \eqref{eq:PbbuDecay.D}-\eqref{eq:PbbuDecay2.D} and \eqref{eq:PbbuDecay.E}-\eqref{eq:PbbuDecay3.E}.
\end{proof}

\subsection{Proof of Theorem~\ref{glob-stab-thm}}
\label{sec:proof_globstab}
Fix $t_0>0$, $n\in\Zbb_{\ge 3}$, $k \in \Zbb_{>\frac{n+3}{2}}$, $V_0=-1$, $\nu\in (0,1)$, and $s>s_c$, where $s_c$ is defined in \eqref{sc-def}. Let $\Rtt$, $\Phi_0$, $\Ltt$, $\Lambda$ and $\alpha$ be as defined in \eqref{eq:ekpyroticparameterconstants.N}. Furthermore, choose $\ep$ satisfying \eqref{epcond}, and let $R>0$, $l\in \Nbb_0$ and $\delta_0>0$ be as in Proposition~\ref{prop:globalstability}.

\subsubsection{Fuchsian stability}  
Next, we fix $\delta \in (0,\delta_0]$ and assume that $\gr_{\mu\nu}\in H^{k+l+2}(\Tbb^{n-1},\mathbb{S}_n)$, 
$\ggr_{\mu\nu}\in H^{k+l+1}(\Tbb^{n-1},\mathbb{S}_n)$, $\taur=t_0$ and
$\taugr\in H^{k+l+2}(\Tbb^{n-1})$ satisfy \eqref{glob-stab-thm-idata-A} as well as the constraint equations \eqref{grav-constr}--\eqref{wave-constr}. Recalling that $\{\gr_{\mu\nu},\ggr_{\mu\nu},\taur=t_0,\taugr\}$ determines via \eqref{dt-tau-idata}-\eqref{dt-g-idata} initial data $\{g_{\mu\nu}|_{\Sigma_{t_0}},\del{0}g_{\mu\nu}|_{\Sigma_{t_0}},\tau=t_0,\del{0}\tau =1\}$ for the conformal metric $g_{\mu\nu}$ and the scalar field $\tau$ in Lagrangian coordinates on $\Sigma_{t_0}$,
we set 
\begin{equation*}
  e_0^\mu=(-|\chi|_g^2)^{-\frac{1}{2}}\chi^\mu,  
\end{equation*} 
and note its value on $\Sigma_{t_0}$ is determined by the initial data via \eqref{chi-idata}.
 Using a Gram-Schmidt orthogonalisation procedure, we then select spatial frame initial data $e^\mu_I\bigl|_{\Sigma_{t_0}}=\er^\Lambda_I \in H^{k+l}(\Tbb^{n-1})$ so that
\begin{equation*} 
\norm{\er^\Lambda_I-\delta^\Lambda_I}_{H^{k+l}(\Tbb^{n-1})} < \delta
\end{equation*}
and the frame $e_i=e^\mu_i\del{\mu}$ is orthonormal on $\Sigma_{t_0}$ with respect to the conformal metric $g_{\mu\nu}$, see \eqref{g-idata}.

Following the procedure outlined in Section~\ref{frame-idata}, we construct initial data $W\bigl|_{\Sigma_{t_0}}=W_0$ for  \eqref{Fuch-ev-A-W} using the variable definitions \eqref{gac-Q000}-\eqref{gacclastdef}, \eqref{eq:ekpyroticbeta} and \eqref{W-def}. 
Using Sobolev and Moser inequalities, e.g.,~ Propositions 2.4., 3.7.~and 3.9.~from   \cite[Ch.~13]{taylor2011}, it is then straightforward to verify that $W_0$ satisfies \begin{equation}\label{W0-bnd}
    \norm{W_0}_{H^{k+l}(\Tbb^{n-1})} \lesssim \delta.
\end{equation}
Then by Proposition~\ref{prop:globalstability}, there exists a unique solution   
\begin{equation}\label{eq:symhypreg_appl}
  W \in C^0_b\bigl((0,t_0],H^{k+l}(\Tbb^{n-1})\bigr)\cap C^1\bigl((0,t_0],H^{k+l-1}(\Tbb^{n-1})\bigr)
\end{equation}
to \eqref{Fuch-ev-A-W} satisfying the initial condition $W(t_0)=W_0$.
Moreover, this solution 
satisfies the energy 
and decay estimates \eqref{eq:resenergy.Cor}--\eqref{eq:PbbuDecay4.Cor} for some $\kappat>0$. By definition $\Pbb^\perp=\mathbbm{1}-\Pbb$, and so by
\eqref{Pbb-def}, we see that
  \begin{equation}\label{Pbb-perp-def}
  \Pbb^\perp = \diag\Bigl(\delta^{\It I}\delta_{\Lambdat \Lambda},
    \delta^{\Lt L}\delta^{\Mt M},0,\ldots,0\Bigr).
\end{equation}
Using this and \eqref{W-def}, we express the limit $\Pbb^\perp W(0)\in H^{k}(\mathbb T^{n-1})$ as
\begin{equation} \label{Pbb-perp-u(0)}
\Pbb^\perp W(0)=\bigl(\ef_I^{\Lambda},2\kf_{IJ},0,\ldots,0\bigr),
\end{equation}
where $\kf_{IJ}\in H^{k}(\mathbb T^{n-1},\Sbb{n-1})$ and $\ef_I^{\Lambda}\in H^{k}(\mathbb T^{n-1},M_{n-1})$. 

\subsubsection{Frame convergence:}
Recalling \eqref{e0-mu} and \eqref{e0I-fix}, we have
\begin{equation} \label{e0i-fix}
    e_0^\mu=\beta^{-1}\delta_0^\mu \AND e_I^0=0, 
\end{equation}
while 
\begin{align}
  \norm{\beta-1}_{H^{k}(\Tbb^{n-1})}=\bnorm{t^{\Rtt-\ep}\betah}_{H^{k}(\Tbb^{n-1})}&\lesssim t^{\Rtt-\ep+\tilde\kappa},\label{betaconv}\\
  \norm{e_I^\Lambda - (\delta_I^\Lambda+\ef_I^\Lambda)}_{H^{k}(\mathbb T^{n-1})} &\lesssim t^{\tilde\kappa},\label{frameconv}
\end{align}
by \eqref{eq:ekpyroticbeta}, \eqref{Pbb-def},  \eqref{W-def},  \eqref{Pbb-perp-u(0)}, and the decay estimates  \eqref{eq:PbbuDecay.Cor}--\eqref{eq:PbbuDecay2.Cor}.   
Together, \eqref{e0i-fix} and \eqref{betaconv}-\eqref{frameconv} imply that the conformal orthonormal frame components $e_i^\mu$ converges as $t\searrow 0$.

\subsubsection{Past stability of the positive ekpyrotic-FLRW solution:} Proposition~\ref{lag-exist-prop} implies, for some $t_1\in (0,t_0]$, the existence of a unique solution 
\begin{equation*}
    \Wsc \in \bigcap_{j=0}^{k}C^j\bigl((t_1,t_0], H^{k-j}(\Tbb^{n-1})\bigr)
\end{equation*} to the system \eqref{tconf-ford-C.1}-\eqref{tconf-ford-C.8} on $M_{t_1,t_0}=(t_1,t_0]\times \Tbb^{n-1}$ that satisfies the initial conditions \eqref{l-idata}-\eqref{hhu-idata}. Since the conformal Einstein-Euler-scalar field initial data satisfy the gravitational and wave gauge constraint equations by assumption, Proposition~\ref{lag-exist-prop} further implies that $\Wsc$ satisfies the constraints \eqref{eq:Lag-constraints}, and consequently, determines a solution $\{g_{\mu\nu},\tau\}$ of the conformal Einstein-scalar field equations \eqref{lag-confeqns} in Lagrangian coordinates on $M_{t_1,t_0}=(t_1,t_0]$ that satisfies the wave gauge constraint \eqref{lag-wave-gauge} and
\begin{equation} \label{tau=t}
\tau=t.
\end{equation}

On the other hand, from the derivation of the system the Fuchsian equation \eqref{Fuch-ev-A-W}, we know that $\Wsc$ determines a solution
\begin{equation*}
  \Wt \in C^0_b\bigl((0,t_0],H^{k+l}(\Tbb^{n-1})\bigr)\cap C^1\bigl((0,t_0],H^{k+l-1}(\Tbb^{n-1})\bigr)
\end{equation*}
to \eqref{Fuch-ev-A-W} on $M_{t_1,t_0}$ satisfying $\Wt|_{\Sigma_{t_0}}=W_0$, and which in turn, implies by uniqueness, see Proposition~\ref{prop:globalstability}, that
$\Wt = W|_{M_{t_1,t_0}}$. 

With the help of the energy bounds \eqref{eq:resenergy.Cor} and Sobolev's inequality, we deduce that
\begin{equation*} %\label{ut-bnd}
    \sup_{t_1<t<t_0}\norm{\Wt(t)}_{W^{2,\infty}(\Tbb^{n-1})} \lesssim \delta.
\end{equation*}
From this bound, \eqref{p-fields}, \eqref{gi00},   \eqref{kt-def}-\eqref{tau-def}, \eqref{psit-def}, \eqref{gac-Q000}-\eqref{l-def.NN}, \eqref{W-def}, \eqref{eq:ekpyroticbeta}, and \eqref{e0i-fix}, we then get
\begin{align}
\sup_{t_1<t<t_0}\Bigl(\norm{e_j^\mu(t)}_{W^{2,\infty}(\Tbb^{n-1})}&+\norm{\Dc_i g_{jk}(t)}_{W^{2,\infty}(\Tbb^{n-1})}+
\norm{\beta(t)}_{W^{2,\infty}(\Tbb^{n-1})}
\notag \\
&
+ 
\norm{ \kt_{IJ}(t)}_{W^{2,\infty}(\Tbb^{n-1})}
+\norm{\psit_I{}^k{}_J(t)}_{W^{2,\infty}(\Tbb^{n-1})}\notag \\
&
+
\norm{\gamma_I{}^k{}_J(t)}_{W^{2,\infty}(\Tbb^{n-1})}
+\norm{\Dc_i\Dc_j\tau}_{W^{2,\infty}(\Tbb^{n-1})}\Bigr)<\infty.
\label{glob-cont-bnd-A}
\end{align}
Using this bound, we observe from the evolution equations \eqref{for-O.1.S2.2} and \eqref{for-M.1.S2.2} for $\beta$ and $e^\Lambda_I$, respectively, together with \cite[Lemma~A.2]{BeyerOliynyk:2021}, that
\begin{equation*}
\inf_{M_{t_1,t_0}}\bigr\{\beta,\det(e^\Lambda_I)\bigl\} > 0,
\end{equation*}
and hence, by \eqref{e0i-fix} that
\begin{equation}\label{glob-cont-bnd-C}
\inf_{M_{t_1,t_0}}\det(e^\mu_j) > 0.
\end{equation}

Due to the orthonormality of the frame $e_i^\mu$, the components of the conformal metric in the Lagrangian coordinates are given by $g_{\mu\nu}=e_\mu^i \eta_{ij} e_\nu^j$. From this representation and the bound \eqref{glob-cont-bnd-A} and \eqref{glob-cont-bnd-C}, we see that
\begin{equation} \label{glob-cont-bnd-D}
\sup_{t_1<t<t_0}\norm{g_{\mu\nu}(t)}_{W^{2,\infty}(\Tbb^{n-1})} < \infty \AND \sup_{M_{t_1,t_0}}\det(g_{\mu\nu})<0.
\end{equation}

Next, with the help of the calculation
\begin{align*}
    e^\mu_i e^\nu_j \del{t}g_{\mu\nu} &= \Ld_{\del{t}}g_{ij}
    \oset{\eqref{e0i-fix}}{=}\Ld_{\beta e_0}g_{ij}\\
    &=\beta e_0^k \Dc_{k}g_{ij}+\Dc_i (\beta e_0^k)\eta_{kj}
    + \Dc_j (\beta e_0^k)\eta_{jk} \\
    & =\beta \delta_0^k \Dc_{k}g_{ij}+ e_i(\beta)\eta_{0j}+\beta \gamma_i{}^k{}_0 \eta_{kj}
    + e_j(\beta)\eta_{0i}+\beta \gamma_j{}^k{}_0\eta_{ik},
\end{align*}
we observe that the bound
\begin{equation} \label{glob-cont-bnd-E}
\sup_{t_1<t<t_0}\norm{\del{t}g_{\mu\nu}(t)}_{W^{1,\infty}(\Tbb^{n-1})} < \infty
\end{equation}
is an immediate consequence of
\eqref{glob-cont-bnd-A}-\eqref{glob-cont-bnd-C},
the relations  \eqref{p-fields}, \eqref{gamma-000}-\eqref{gamma-0K0}, \eqref{gamma-I00}-\eqref{gamma-IJ0} and \eqref{e0i-fix}, and the evolution
equation \eqref{for-O.1.S2.2}. Employing similar arguments, it is also not difficult to 
verify 
\begin{equation} \label{glob-cont-bnd-F}
\sup_{t_1<t<t_0}\Bigl(\norm{\Dc_\nu \chi^\mu(t)}_{W^{2,\infty}(\Tbb^{n-1})}+\norm{\del{t}(\Dc_\nu \chi^\mu)(t)}_{W^{1,\infty}(\Tbb^{n-1})}\Bigr) < \infty,
\end{equation}
where we recall that $\chi^\mu$ is defined by \eqref{chi-def}.

Collectively, the definition \eqref{alpha-def} and the  bounds \eqref{glob-cont-bnd-A}, \eqref{glob-cont-bnd-D}, \eqref{glob-cont-bnd-E} and \eqref{glob-cont-bnd-F}, allow us to conclude from the continuation criterion, Proposition \ref{lag-exist-prop}.(f), that the solution $\Wsc$ can be continued past $t_1$, and in fact, that $t_1=0$.
Consequently, $\Wsc$ determines a solution
$\{g_{\mu\nu},\tau\}$
of the conformal Einstein-scalar field equations \eqref{lag-confeqns} in Lagrangian coordinates on $M_{0,t_0}$ that satisfies the wave gauge constraint \eqref{lag-wave-gauge} and the gauge condition \eqref{tau=t}. This, in particular, proves the past stability of the positive ekpyrotic-FLRW solution.

\subsubsection{Conformal and physical $(n-1)+1$-decompositions} 
As shown in Appendix~\ref{spacetimedecomp}, the foliation determined by \eqref{tau=t} yields an $(n-1)+1$ decomposition of the conformal metric with vanishing shift and conformal lapse given by
\begin{equation} \label{Ntt=betat}
\Ntt=\beta.
\end{equation}
In particular, the estimate \eqref{lapseconv.thm} follows from \eqref{betaconv}.

Since $e_0^i=n^i$, the conformal spatial metric $\cspatmetr_{ij}$, see \eqref{eq:confspatmetr}, is given by \eqref{eq:cspatmetr.thm}, while
the conformal second fundamental form $\Ktt_{ij}$, see \eqref{eq:confsecff}, is determined by
\begin{equation*}
  t^{1-\Rtt}\beta \cextcurv_{LJ}=\frac 12 k_{LJ}
  +t^{1-\ep-\Rtt}\beta \psi_{(L}{}^0{}_{J)},
\end{equation*}       
where in deriving this we have used \eqref{psit-def} and  \eqref{psi-def.NN}--\eqref{k-def.NN}. The bound
\begin{equation*}
 % \label{eq:cf2ffbd}
  \norm{ t\beta \cextcurv_{LJ}(t)-t^{\Rtt}\kf_{LJ}}_{H^{k}(\mathbb T^{n-1})} 
  \lesssim t^{\Rtt+\kappat}
  +t^{1-\ep+\kappat} 
\end{equation*}
then follows directly from \eqref{W-def}, \eqref{eq:PbbuDecay.Cor}, \eqref{Pbb-perp-u(0)} and \eqref{betaconv}. Decomposing the confomal second fundamental form, c.f.~\eqref{eq:confexpshear}, in terms of conformal expansions and shear as 
$\Ktt_{IJ} = \sigma_{IJ}+\Htt \gtt_{IJ}$,
the above estimate implies
\begin{equation}\label{Htt-sigma-bnds}
    \biggl\| t\beta \cexp-\frac 1{n-1}t^\Rtt\kf_{M}{}^{M}\biggr\|_{H^{k}(\mathbb T^{n-1})} +
  \biggl\| t\beta \cshear_L{}{}^J-t^{\Rtt}\Bigl(\kf_{L}{}^{J}-\frac{1}{n-1}\kf_{M}{}^{M}\gtt_{L}{}^{J}\Bigr)\biggr\|_{H^{k}(\mathbb T^{n-1})} 
  \lesssim t^{\kappat+\Rtt}
  +t^{1-\ep+\kappat}.
\end{equation}
This, together with \eqref{Ntt=betat} and \eqref{lapseconv.thm},
implies \eqref{eq:confHshear.thm}.

Turning to the physical quantities, we first note that the formula \eqref{eq:conffactorekpy.thm} for the conformal factor $e^\Phi$ follows directly from \eqref{confmet}, \eqref{eq:conf2phys}, \eqref{eq:ekpyroticparameterconstants.N} and  \eqref{tau=t}. 
Using \eqref{eq:conffactorekpy.thm}, it is then immediate that the physical frame estimates \eqref{frameconv.phys.thm.1}--\eqref{frameconv.phys.thm.2} follow from \eqref{frameconv.thm}. 

Next, from \eqref{confmet} and the orthonormality of the frame $e_i=e_i^\mu\del{\mu}$ with respect to the conformal metric $g$, we observe that (physical) frame
\begin{equation} \label{eq:porthonframe}
\eb_{\ib}=\eb_{\ib}^\mu\del{\mu}, \quad \eb_{\ib}^\mu=e^{-\Phi} e_i^\mu \delta_{\ib}^i
\end{equation}
is orthonormal with respect to the physical metric $\gb$, that is,
\begin{equation*}
\gb_{\ib\jb}:=\gb(\eb_{\ib},\eb_{\jb})=\eta_{\ib\jb}.
\end{equation*}
We also note that the future pointing unit normal $\nb$ to the foliation \eqref{tau=t}, as determined by the physical metric, has components relative to the physical and orthonormal frames given by 
\begin{equation}
  \label{eq:trafounitnormal}
  \pnormal_{\ib}=-\delta_{\ib}^0 \AND \pnormal_{i}=e^\Phi \cnormal_i=-e^\Phi\delta_{i}^0,
\end{equation}
respectively. Moreover, the shift vanishes, and the (physical) lapse $\Nb$ is given by 
\begin{equation*}%\label{barNtt}
  \bar\Ntt=e^\Phi\Ntt.
\end{equation*}
Combining these relations with \eqref{eq:conffactorekpy.thm} and \eqref{lapseconv.thm}, we obtain the bound \eqref{lapseconv.phys.thm} on the physical lapse.

Using \eqref{eq:pextcurv} and \eqref{e0i-fix}, we express the components of the physical second fundamental form relative to the conformal orthonormal frame as
\begin{equation*}
  \label{eq:physWeingarten}
  t\beta\pextcurv_{I}{}^{J}=e^{-\Phi}\biggl(t\beta\cextcurv_{I}{}^{J}
  +\frac{s_c^2}{(n-1)s^2-s_c^2}\cspatmetr_{I}{}^{J}\biggr).
\end{equation*}
Decomposing the physical second fundamental form, c.f.~\eqref{eq:physshearexp}, into the physical expansion and shear via
$\pextcurv_{I}{}^{J}=\sigmab_I{}^J+\bar{\Htt}\bar{\gtt}_I{}^J$,
we see from \eqref{eq:pexp}--\eqref{eq:pshear} that
\begin{equation}
  t\beta\pexp=e^{-\Phi}\biggl(t\beta\cexp
  +\frac{s_c^2}{(n-1)s^2-s_c^2}\biggr),\quad 
  t\beta\pshear_{I}{}^J=e^{-\Phi}\cshear_{I}{}^J.
\end{equation}
From these expressions, \eqref{eq:porthonframe}
and the bounds \eqref{Htt-sigma-bnds}, we obtain the estimates
\begin{align}
  \biggl\| t\beta e^\Phi\pexp-\frac{s_c^2}{(n-1)s^2-s_c^2}-\frac 1{n-1}t^\Rtt\kf_{M}{}^{M}\biggr\|_{H^{k}(\mathbb T^{n-1})} 
&\lesssim t^{\kappat+\Rtt}
+t^{1-\ep+\kappat},\label{eq:pexpbound}\\
  \biggl\| t^{1-\Rtt}\beta e^\Phi\bar\sigma_{\Lb}{}{}^{\Jb}-\Bigl(\kf_{L}{}^{J}-\frac{1}{n-1}\kf_{M}{}^{M}\gtt_{L}{}^{J}\Bigr)\delta^L_{\Lb}\delta_J^{\Jb}\biggr\|_{H^{k}(\mathbb T^{n-1})} 
&\lesssim t^{\kappat}
+t^{1-\ep-\Rtt+\kappat},
\end{align}
for the physical expansion and shear relative to the physical orthonormal frame \eqref{eq:porthonframe}. These estimates together with the formula \eqref{eq:conffactorekpy.thm} for the conformal factor, \eqref{Ntt=betat} and \eqref{lapseconv.thm} and the assumption that $\Rtt<1-\ep$  yield the bounds \eqref{eq:pexpbound.thm} and \eqref{eq:pshearbound.thm}.

Next, we note that
\begin{equation*}
\frac{1}{t\beta e^\Phi\pexp}-\frac{(n-1)s^2-s_c^2}{s_c^2} = \chi\biggl(t\beta e^\Phi\pexp-\frac{s_c^2}{(n-1)s^2-s_c^2}\biggr)
\end{equation*}
where 
\[\chi(x)=\frac{1}{\frac{s_c^2}{(n-1)s^2-s_c^2}+x}-\frac{(n-1)s^2-s_c^2}{s_c^2}.\]
From the above identity and the smoothness of $\chi(x)$ for $|x|<\frac{s_c^2}{(n-1)s^2-s_c^2}$, we deduce, with the help of
the Sobolev and Moser inequalities, e.g.,~ Propositions 2.4., 3.7.~and 3.9.~from   \cite[Ch.~13]{taylor2011}, and the bound \eqref{eq:pexpbound}, that 
\begin{equation}
  \label{eq:inversepexpbound}
  \Bnorm{ \frac{1}{t\beta e^\Phi\pexp}-\frac{(n-1)s^2-s_c^2}{s_c^2}}_{H^{k}(\mathbb T^{n-1})} 
\lesssim t^{\Rtt}.
\end{equation}
In particular, this implies that the physical orthonormal frame components of the \emph{scale invariant} physical shear tensor $\bar\Sigma_{\Lb}{}^{\Jb}$, see \eqref{eq:physshearresc}, is bounded by 
\begin{equation}
  \label{eq:pshearrescbound}
  \biggl\|t^{-\Rtt}\bar\Sigma_{\Lb}{}^{\Jb}-\Bigl(\kf_{\Lb}{}^{\Jb}-\frac{1}{n-1}\kf_{\Mb}{}^{\Mb}\gtt_{\Lb}{}^{\Jb}\Bigr)\frac{(n-1)s^2-s_c^2}{s_c^2}\biggr\|_{H^{k}(\mathbb T^{n-1})}
  \lesssim {t^{\kappat}
+t^{1-\ep-\Rtt+\kappat}+t^{\Rtt} },
\end{equation}
which establishes \eqref{eq:pshearrescbound.thm}.

In order to establish geodesic incompleteness, we require a lower bound on the physical expansion. We obtain this by 
employing Sobolev's inequality together with the estimate \eqref{eq:pexpbound} as follows:
\begin{align*}
  \label{eq:pexplowerbound}
  \inf_{x\in\Tbb^{n-1}} t\beta e^\Phi\pexp
  &=\frac{s_c^2}{(n-1)s^2-s_c^2}+\inf_{x\in\Tbb^{n-1}} \left(t\beta e^\Phi\pexp-\frac{s_c^2}{(n-1)s^2-s_c^2}\right)\\
  &\ge\frac{s_c^2}{(n-1)s^2-s_c^2}-\sup_{x\in\Tbb^{n-1}} \left|t\beta e^\Phi\pexp-\frac{s_c^2}{(n-1)s^2-s_c^2}\right|\\
  &\ge\frac{s_c^2}{(n-1)s^2-s_c^2}-
  C\biggl\| t\beta e^\Phi\pexp-\frac{s_c^2}{(n-1)s^2-s_c^2}\biggr\|_{H^k(\mathbb T^{n-1})} \\ 
&\ge \frac{s_c^2}{(n-1)s^2-s_c^2}-C
t^{\Rtt}. 
\end{align*}
This lower bound together with \eqref{eq:conffactorekpy.thm} and  \eqref{betaconv} shows that the physical expansion scalar $\pexp$ blow-up uniformly as $t\searrow 0$. Consequently, $t=0$ is a \emph{crushing singularity} \cite{eardley1979}, and
Hawking's singularity theorem \cite[Chapter~14, Theorem~55A]{oneill1983a} guarantees that timelike geodesics are past incomplete.

\subsubsection{Spatial curvature decay}  The conformal spatial Riemann curvature tensor, when expressed with respect to the conformal orthonormal frame, is given by
\begin{align} %\label{Rtt-formula}
  \cspatcurv_{IJK}{}^M
    = e_J(\Gamma_I{}^M{}_K)
    - e_I(\Gamma_J{}^M{}_K)
    -\Gamma_I{}^M{}_L\Gamma_J{}^L{}_K
    +\Gamma_J{}^M{}_L\Gamma_I{}^L{}_K
    +\Gamma_L{}^M{}_K \Gamma_I{}^L{}_J
    -\Gamma_L{}^M{}_K \Gamma_J{}^L{}_I,
\end{align}
where $\Gamma_M{}^J{}_K$ are the spatial components of the connection coefficients of the conformal spatial metric $\cspatmetr_{IJ}=\delta_{IJ}$. From this and \eqref{Ccdef},  \eqref{p-fields}, \eqref{ellt-def}, \eqref{psit-def},   \eqref{psi-def.NN} and \eqref{l-def.NN}, we observe that $\cspatcurv_{IJK}{}^M$ can be expressed using $*$-notation from Section \ref{*-not} as 
\begin{equation}
  \label{eq:spatialcurvaturerepr}
  \cspatcurv_{IJK}{}^M=t^{-\ep} e* \del{}\psi+t^{-\ep} e* \del{}\ell+t^{-2\ep}(\psi+\ell)*(\psi+\ell).
\end{equation}

As the conformal factor $e^\Phi$ (see \eqref{eq:conffactorekpy.thm}) is constant along each $t=const$-hypersurface, it is straightforward to verify, using the above expression, that the physical spatial Riemann curvature tensor, expressed in terms of the physical orthonormal frame $\bar e_{\Ib}^\mu$ (see \eqref{eq:porthonframe}), is
\begin{equation*}
  \pspatcurv_{\Ib\Jb\bar K}{}^{\Mb}=e^{-2\Phi}\bigl(t^{-\ep} e* \del{}\psi+t^{-\ep} e* \del{}\ell+t^{-2\ep}(\psi+\ell)*(\psi+\ell)\bigr).
\end{equation*}
A short calculation using \eqref{betaconv}, \eqref{frameconv}, and \eqref{eq:inversepexpbound} then shows that the physical orthonormal frame components of the physical spatial Riemann curvature tensor, rescaled by the square of the physical expansion scalar, satisfy
\begin{equation}
  \label{eq:spatcurvaturedecay}
  \Bnorm{\frac{\pspatcurv_{\Ib\Jb\bar K}{}^{\Mb}}{\pexp^2}}_{H^{k-1}(\mathbb T^{n-1})}\lesssim{t^{2-2\ep}}+t^{\Rtt},
\end{equation}
which establishes \eqref{eq:spatcurvaturedecay.thm}.

\subsubsection{Spacetime Ricci tensor blow up} By \eqref{eq:pstRicci},
we have
\begin{align*}
  (t\beta e^{\Phi})^{2}\Rb_{\ib \bar k}\pnormal^{\ib}\pnormal^{\bar k}&=\frac{2}{\alpha^2}+\frac{4}{n-2}t^{s/\alpha}(t\beta e^{\Phi})^2,\\
  (t\beta e^{\Phi})^{2}\Rb_{\ib \bar k}\pspatmetr_{\jb}{}^{\ib}\pspatmetr_{\lb}{}^{\bar k}&=
  -\frac{4}{n-2}t^{s/\alpha}(t\beta e^{\Phi})^{2}\pspatmetr_{\jb \lb},
\end{align*}
where the $\Rb_{\ib\jb}$ are the conformal orthonormal frame components of the physical Ricci tensor.
Noting from \eqref{sc-def}, \eqref{eq:ekpyroticparameterconstants.N}, \eqref{eq:conffactorekpy.thm}, \eqref{tau=t}, and \eqref{Ntt=betat} that
\begin{equation}
  \frac{2}{\alpha^2}=\frac{8(n-1)^2s^2}{((n-1)s^2-s_c^2)^2}\AND
  \frac{4}{n-2}t^{s/\alpha}t^2 \beta^2 e^{2\Phi}
  = 
  \frac{(n-1)s_c^2(s^2-s_c^2)}{((n-1)s^2-s_c^2)^2}\beta^2,\label{eq:aux2}
\end{equation}
we can, using \eqref{sc-def}, rewrite the above formulas for the physical Ricci tensor as
\begin{align*}
  (t\beta e^{\Phi})^{2}\Rb_{\ib \bar k}\pnormal^{\ib}\pnormal^{\bar k}&=
  \frac{(n-1)s_c^2}{(n-1)s^2-s_c^2}+\frac{(n-1)s_c^2(s^2-s_c^2)}{((n-1)s^2-s_c^2)^2} (\beta+1)(\beta-1),\\
  (t\beta e^{\Phi})^{2}\Rb_{\ib \bar k}\pspatmetr_{\jb}{}^{\ib}\pspatmetr_{\lb}{}^{\bar k}&=
  -\left(\frac{(n-1)s_c^2(s^2-s_c^2)}{((n-1)s^2-s_c^2)^2}
  +\frac{(n-1)s_c^2(s^2-s_c^2)}{((n-1)s^2-s_c^2)^2}(\beta+1)(\beta-1)\right)\pspatmetr_{\jb \lb}.
\end{align*}
Combining these formulas with the estimates \eqref{betaconv} and \eqref{eq:inversepexpbound}, it follows that
\begin{align*}
  \biggl\|\frac{\Rb_{\ib \bar k}\pnormal^{\ib}\pnormal^{\bar k}}{\pexp^2}-(n-1)\frac{(n-1)s^2-s_c^2}{s_c^2}\biggr\|_{H^{k}(\mathbb T^{n-1})}
  &\lesssim t^{\Rtt-\ep+\kappat} +t^{\Rtt},\\
  \biggl\|\frac{\Rb_{\ib \bar k}\pspatmetr_{\jb}{}^{\ib}\pspatmetr_{\lb}{}^{\bar k}}{\pexp^2}
  +\frac{(n-1)(s^2-s_c^2)}{s_c^2}\pspatmetr_{\jb\lb}\biggr\|_{H^{k}(\mathbb T^{n-1})}
  &\lesssim
  t^{\Rtt-\ep+\kappat} +t^{\Rtt},
\end{align*}
which establishes the bounds \eqref{eq:Ricciblowup.thm.1} and \eqref{eq:Ricciblowup.thm.2}.
These estimates, in conjunction with \eqref{eq:pexpbound}, imply that the physical Ricci tensor blows up at the same rate as the square of the physical expansion scalar. In particular, this shows that the perturbed positive ekpyrotic-FLRW solutions are $C^2$-inextendible at $t=0$.

\subsubsection{Physical Weyl tensor decay} 
According to \eqref{eq:Weylrepresentation}, the conformal orthonormal frame components of the physical Weyl tensor $\Wb_{ijkl}$ can be decomposed in terms of the tensor fields $\pWeyl_{mnpq}$ and $\pWeyl_{mnp}$, where $\pWeyl_{mnpq}$ and $\pWeyl_{mnp}$ are defined by \eqref{eq:relevantWeylprojections} and can be computed using  \eqref{Weyl1}-\eqref{Weyl2}. 

To analyse $\pWeyl_{mnp}$, we introduce the function 
\[\xi(x)=\log\biggl(\frac{s_c^2}{(n-1)s^2-s_c^2}+x\biggr)-\log\biggl(\frac{s_c^2}{(n-1)s^2-s_c^2}\biggr),\]
which allows us to write
\begin{equation*}
\log(t\beta e^\Phi\pexp)-\log\biggl(\frac{s_c^2}{(n-1)s^2-s_c^2}\biggr)=  \xi\biggl(t\beta e^\Phi\pexp-\frac{s_c^2}{(n-1)s^2-s_c^2}\biggr).
\end{equation*}
The smoothness of $\xi(x)$ for $|x|\leq\frac{s_c^2}{(n-1)s^2-s_c^2}$,
the Sobolev and Moser inequalities, e.g.,~ Propositions 2.4., 3.7.~and 3.9.~from   \cite[Ch.~13]{taylor2011}, and the bound \eqref{eq:pexpbound} imply
\begin{equation}
  \label{eq:logpexpbound}
  \biggl\| \log(t\beta e^\Phi\pexp)-\log\left(\frac{s_c^2}{(n-1)s^2-s_c^2}\right)\biggr\|_{H^{k}(\mathbb T^{n-1})} 
\lesssim t^{\Rtt}.
\end{equation}

Next, we express \eqref{Weyl2} in terms of the physical orthonormal frame and employ  \eqref{eq:conffactorekpy.thm} and \eqref{e0i-fix} to obtain
\begin{align*}
  %\label{Weyl2}
  \frac{\pWeyl_{\Mb\Nb\Pb}}{\pexp}
  =&(\pshearresc_{\Nb \Pb}+\pspatmetr_{\Nb\Pb})\pspnabla_{\Mb} \log (t\beta e^\Phi \pexp)
  -(\pshearresc_{\Nb \Pb}+\pspatmetr_{\Nb\Pb})\beta^{-1}\pspnabla_{\Mb} \beta\\
  &-(\pshearresc_{\Mb\Pb}+\pspatmetr_{\Mb\Pb})\pspnabla_{\Nb} \log(t\beta e^\Phi \pexp)
  +(\pshearresc_{\Mb \Pb}+\pspatmetr_{\Mb\Pb})\beta^{-1}\pspnabla_{\Nb} \beta
  + \pspnabla_{\Mb} \pshearresc_{\Nb\Pb}
  -\pspnabla_{\Nb} \pshearresc_{\Mb\Pb}.
\end{align*}
Using this expression together with \eqref{Ntt=betat} and $\pspatmetr_{\Nb\Pb}=\delta_{\Nb\Pb}$, we obtain the bound 
\begin{equation*}
  \biggl\|\frac{\pWeyl_{\Mb\Nb\Pb}}{\pexp}\biggr\|_{H^{k-1}(\mathbb T^{n-1})}\lesssim t^{\Rtt}+t^{\Rtt-\ep+\kappat}
\end{equation*}
as an immediate consequence of the estimates \eqref{eq:pshearrescbound} and \eqref{eq:logpexpbound}. This establishes \eqref{eq:Weylestimate.thm.1}.

Next, we obtain bounds on $\pWeyl_{mnpq}$ analogous to the ones obtained above for $\pWeyl_{mnp}$. Because the derivation is very similar, we omit most details and instead focus on the differences. The starting point is to express \eqref{Weyl1} in terms of the physical orthonormal frame and divide the resulting expression by $\pexp{}^2$. The terms coming from the first line of \eqref{Weyl1} can immediately be estimated in the $H^{k-1}$-norm using \eqref{eq:pshearrescbound} and \eqref{eq:spatcurvaturedecay}. The most interesting part of the analysis occurs in the second line of \eqref{Weyl1}, where there is a surprising cancellation in the following expression:    
\begin{align*}
  \frac{2}{\alpha^2}
  -4\tau^{s/\alpha}& t^2 e^{2\Phi}\beta^2
  -(n-1)(n-2) (t \beta e^{\Phi}\pexp)^2 \\
  =&\frac{8(n-1)^2s^2}{((n-1)s^2-s_c^2)^2}
  -\frac{(n-1)(n-2)s_c^2(s^2-s_c^2)}{((n-1)s^2-s_c^2)^2}
  -(n-1)(n-2) (\frac{s_c^2}{(n-1)s^2-s_c^2})^2\\
  &-\frac{(n-1)(n-2)s_c^2(s^2-s_c^2)}{((n-1)s^2-s_c^2)^2}(\beta^2-1)
  -(n-1)(n-2)\biggl( (t \beta e^{\Phi}\pexp)^2-\frac{s_c^4}{((n-1)s^2-s_c^2)^2}\biggr)\\
  =&\frac{8(n-1)^2s^2}{((n-1)s^2-s_c^2)^2}
  -\frac{8(n-1)^2(s^2-s_c^2)}{((n-1)s^2-s_c^2)^2}
  - \frac{8(n-1)^2s_c^2}{((n-1)s^2-s_c^2)^2}\\
  &-\frac{(n-1)(n-2)s_c^2(s^2-s_c^2)}{((n-1)s^2-s_c^2)^2}(\beta^2-1)
  -(n-1)(n-2)\biggl( (t \beta e^{\Phi}\pexp)^2-\frac{s_c^4}{((n-1)s^2-s_c^2)^2}\biggr)\\
  =&-\frac{(n-1)(n-2)s_c^2(s^2-s_c^2)}{((n-1)s^2-s_c^2)^2}(\beta^2-1)
  -(n-1)(n-2)\biggl( (t \beta e^{\Phi}\pexp)^2-\frac{s_c^4}{((n-1)s^2-s_c^2)^2}\biggr),
\end{align*}
where in the above calculation we employed \eqref{sc-def}, \eqref{Ntt=betat}, \eqref{tau=t} and \eqref{eq:aux2}.
Putting everything together, we conclude, with the help of \eqref{Weyl1} and \eqref{eq:inversepexpbound}, that 
\begin{equation*}
  \Bnorm{\frac{\pWeyl_{\Mb\Nb\Pb\Qb}}{\pexp^2}}_{H^{k-1}(\mathbb T^{n-1})}\lesssim t^{2-2\ep}+t^{\Rtt}+t^{\Rtt-\ep+\kappat},
\end{equation*}
which establishes \eqref{eq:Weylestimate.thm.2}.

\subsubsection{Scalar field limits}

Recall from Section \ref{sec:FLRWEinsteinSF} the quantities $\phi_0$ and $\phi_1$ defined by \eqref{eq:Ringstromlimits} for homogeneous solutions to Einstein-scalar field equations, whose limits as $t\searrow 0$ characterize the asymptotic properties of the scalar field. Since we are concerned with perturbations of the positive ekpyrotic-FLRW solutions, we set $\xtt_i=s/s_c$ and define analogous quantities in the inhomogeneous setting via
\begin{align}
  \phi_1&:=\frac{\pnormal^i\nablab_i \phi}{(n-1)\pexp}
  =-\frac{1}{(n-1)\alpha}\frac{1}{e^{\Phi}\beta\tau \pexp}
  =\frac{2s}{(n-1)s^2-s_c^2}\frac{1}{e^{\Phi}\beta t \pexp}, \label{eq:phi1.proof}\\
  \phi_0&:=\phi+\frac{s_c^2}{s^2}\phi_1\log((n-1){\bar\Htt})\label{eq:phi0.proof}\\
  &=-\frac 1\alpha\log t-\frac{s_c^2}{s^2}\phi_1\log(t e^\Phi)
  +\frac{s_c^2}{s^2}\phi_1\log(t e^\Phi\beta\bar\Htt)
  +\frac{s_c^2}{s^2}\phi_1\log((n-1))
  -\frac{s_c^2}{s^2}\phi_1\log(\beta),\notag
\end{align}
where, in obtaining the further identities, we have used  
\eqref{eq:ekpyroticparameterconstants.N}, \eqref{grad-tau}, \eqref{eq:conf2phys.thm},  \eqref{Ntt=betat}, \eqref{eq:confnormal} and \eqref{eq:trafounitnormal}.
From \eqref{eq:inversepexpbound} and \eqref{eq:phi1.proof}, we immediately see that
\begin{equation*}
%\label{eq:inversepexpbound}
  \Bnorm{ \phi_1-\frac{2s}{s_c^2}}_{H^{k}(\mathbb T^{n-1})} 
  \lesssim t^{\Rtt},
\end{equation*}
which shows that $\phi_1$ has the same limit as $t\searrow 0$ as in the homogeneous setting, c.f.~\eqref{eq:FLRWPhiOekpy}. For $\phi_0$, we first consider the first two terms on the right-hand side of \eqref{eq:phi0.proof}: 
\begin{align*}
  &-\frac 1\alpha\log t-\frac{s_c^2}{s^2}\phi_1\log(t e^\Phi)\\
  =&
  -\frac{2s}{s^2}\log\left(\frac{(n-1)\sqrt{2(s^2-s_c^2)}}{(n-1)s^2-s_c^2}\right)\\
  &+\frac{(n-1)s_c^2}{(n-1)s^2-s_c^2}\left(\frac{2s}{s_c^2}-\phi_1\right)\log t
  -\frac{s_c^2}{s^2}\left(\phi_1-\frac{2s}{s_c^2}\right)\log\left(\frac{(n-1)\sqrt{2(s^2-s_c^2)}}{(n-1)s^2-s_c^2}\right)
\end{align*}
owing to \eqref{eq:ekpyroticparameterconstants.N} and \eqref{eq:conffactorekpy.thm}. Combining this with \eqref{eq:logpexpbound}, we conclude from an application of the Sobolev and Moser inequalities to \eqref{betaconv} that 
\begin{equation*}
  \Bnorm{\phi_0
  -\frac{1}{s}\log\left(\frac{s_c^4}{2(s^2-s_c^2)}\right)
  }_{H^{k}(\mathbb T^{n-1})} 
  \lesssim t^\Rtt(1+\log t) +t^{\Rtt-\ep+\kappat}.
\end{equation*}
From this estimate, we see, after a short calculation, that $\phi_0$ also has the same limit as $t\searrow$ as in the homogeneous setting, c.f.~ \eqref{eq:FLRWphi0} with $\xtt_i=s/s_c$ and $\phi_*$ given by \eqref{eq:FLRWFPConstr.ekp.2}. This yields \eqref{eq:phi_0limit.thm}.

\subsubsection{Physical acceleration $\pacc$}
From \eqref{eq:porthonframe}, \eqref{eq:trafounitnormal} and \eqref{defphysextr}, we know that the physical acceleration vector $\pacc$ is purely spatial and that its non-vanishing components relative to the physical orthonormal frame $\eb_i$ are given by
\begin{equation*}
\pacc_{\Jb} = e^{-\Phi}\cacc_{I}\delta_{\Jb}^I = t^{-\ep} e^{-\Phi}\delta^{L}_{\Jb}\Bigl( -\frac{1}{4}\bigl(2m_L-\delta^{KM}(2\ell_{KML}-\ell_{LKM})\bigr)-\beta\xi_{0L}\Bigr), 
\end{equation*}
where in deriving the second equality we used \eqref{m-def}, \eqref{xi-def.N}, \eqref{l-def.NN} and \eqref{ndot-exp}. Normalising the physical acceleration by  $\bar\Htt$ gives 
\begin{equation*}
\frac{\pacc_{\Jb}}{\bar\Htt} =  \frac{t^{1-\ep}}{t\beta \bar\Htt e^{\Phi}} \beta\Bigl( -\frac{1}{4}\bigl(2m_L-\delta^{KM}(2\ell_{KML}-\ell_{LKM})\bigr)-\beta\xi_{0L}\Bigr)\delta^{L}_{\Jb}.
\end{equation*}
With the help of the Sobolev and Moser inequalities, e.g.,~ Propositions 2.4., 3.7.~and 3.9.~from   \cite[Ch.~13]{taylor2011}, along with the decay estimates \eqref{eq:PbbuDecay.Cor} and the definitions \eqref{W-def} and \eqref{Pbb-def}, we can then bound $\frac{\pacc_{\Jb}}{\bar\Htt}$ by
\begin{equation*}
\biggl\|\frac{\pacc_{\Jb}}{\bar\Htt}\biggr\|_{H^k(\Tbb^{n-1})}
\lesssim \Bigl\|  \frac{1}{t\beta \bar\Htt e^{\Phi}}\Bigr\|_{H^k(\Tbb^{n-1})} \norm{\beta}_{H^k(\Tbb^{n-1})}(1+\norm{\beta}_{H^k(\Tbb^{n-1})}) t^{1-\ep+\kappat}.
\end{equation*}
This estimate together with \eqref{betaconv} and \eqref{eq:inversepexpbound} establishes \eqref{ndot-bar-est} and completes the proof.

\appendix

\section{$(n-1)+1$-decomposition}
\label{spacetimedecomp}

Let $\cnormal=\cnormal^\mu\del{\mu}$ denote the future pointing
timelike normal to the foliation of\footnote{Recall $\tau$ is defined by \eqref{taudef}.} $\tau=\textrm{constant}$ hypersurfaces that is unit with respect to the conformal metric $g_{\mu\nu}$. By \eqref{e0-def} and the orthonomality of the frame\footnote{Recall that lower case Latin indices, e.g.,~$i,j,k$, run from $0$ to $n-1$ and enumerate frame indices.}  $e_i=e_i^\mu\del{\mu}$ with respect to that conformal metric, i.e.,~ $g_{ij}:=g(e_i,e_j)=\eta_{ij}$, we have that $e_0^\mu=\cnormal^\mu$,
\begin{equation}
  \label{eq:confnormal}
  \cnormal^i=\delta^i_0 \AND \cnormal_i=-\delta_i^0.
\end{equation}
From~\eqref{e0-def} and  \eqref{e0-mu}, we also note that the \emph{conformal lapse} is given by
\begin{equation}
  \label{eq:conflapsebetat}
  \Ntt=\beta,
\end{equation}
while the \emph{shift vector} vanishes.

The \emph{conformal spatial metric}, which is induced on the $\tau=\textrm{constant}$ hypersurfaces by the conformal metric, is given by 
\begin{equation}
  \label{eq:cspatmetrdef}
  \cspatmetr_{ij}=g_{ij}+\cnormal_i\cnormal_j,
\end{equation}
where\footnote{Recall that upper case Latin indices, e.g.,~$I,J,K$, run from $1$ to $n-1$ and enumerate indices associated to the spatial orthonormal frame $e_I=e_I^\mu\del{\mu}$. } 
\begin{equation}
  \label{eq:confspatmetr}
  \cspatmetr_{i0}=0 \AND \cspatmetr_{IJ}=\delta_{IJ}.
\end{equation}
Letting $\nabla_i$ denote the Levi-Civita connection of the conformal metric $g_{ij}=\eta_{ij}$, the \emph{conformal acceleration} $\cacc_i$ and \emph{conformal extrinsic curvature} $\cextcurv_{ij}$  are uniquely defined via the relations:
\begin{equation}
  \label{defconfextr}
  \nabla_i\cnormal_j=-\cacc_j\cnormal_i+\cextcurv_{ij},\quad \cacc_j n^j=0 \AND \cextcurv_{ij} n^j=\cextcurv_{ji} n^j=0.
\end{equation}
Because the normal $n$ is hypersurface orthogonal, the conformal extrinsic curvature is symmetric, i.e.,~ $\cextcurv_{ij}=\cextcurv_{ji}$, and we observe from \eqref{Ccdef}, \eqref{Gamma-def}, \eqref{gamma-0K0}, \eqref{p-fields}, \eqref{gi00}, \eqref{ellt-def}-\eqref{mt-def} and \eqref{eq:confnormal} that
\begin{equation} \label{ndot-exp}
\cacc^K=\cnormal^j\nabla_j\cnormal^K=\Gamma_0{}^K{}_0
  =-\frac{1}{4}\delta^{K L}(2\mt_{L}-\delta^{JM}(2\ellt_{JML}-\ellt_{LJM}))-\beta \delta^{KL}\tau_{0L}.
\end{equation}
Moreover,
\begin{equation}
  \label{eq:confsecff}
  \cextcurv_{I}{}^{J}=\Gamma_I{}^J{}_0=\frac{1}{2}\kt_{I}{}^{J} + \delta^{JK}\gamma_{I}{}^0{}_K \AND \Ktt_{0}{}^i=\Ktt_i{}^0=0
\end{equation}
are a direct consequence of \eqref{Ccdef}, \eqref{Gamma-def},  \eqref{p-fields}, \eqref{gamma-IJ0}, \eqref{kt-def}, \eqref{eq:confnormal}  and \eqref{defconfextr}.
We further decompose the conformal extrinsic curvature $\cextcurv_{ij}$ into its \emph{conformal expansion} $\cexp$ and \emph{conformal shear} (trace free) $\cshear_{ij}$ as follows
\begin{equation}
  \label{eq:confexpshear}
  \cextcurv_{ij}=\cshear_{ij}+\cexp\cspatmetr_{ij},\quad \cshear_{ij} \cspatmetr^{ij}=0,\quad \cexp=\frac 1{n-1}\cextcurv_{ij}\cspatmetr^{ij},
\end{equation}
and introduce the \emph{conformal scale invariant shear} $\cshearresc_{ij}$ defined by
\begin{equation}
  \label{eq:confshearresc}
  \cshearresc_{ij}=\frac{\cshear_{ij}}{\cexp}.
\end{equation}

From \eqref{confmet}, it is clear that 
\begin{equation}
\label{eq:trafounitnormal.a}
\pnormal_{i}=e^\Phi \cnormal_i=-e^\Phi\delta_{i}^0
\end{equation}
defines the future pointing \emph{physical timelike normal} to the foliation of $\tau=\textrm{constant}$ hypersurfaces  that is unit with respect to the physical metric $\gb_{\mu\nu}$. On the $\tau=\textrm{constant}$ hypersurfaces, the induced physical spatial metric is given by
\begin{equation}
  \label{eq:pspatmetrdef}
  \pspatmetr_{i j}=\gb_{i j}+\pnormal_{i}\pnormal_{j}
\end{equation}
or equivalently, see \eqref{confmet} and \eqref{eq:trafounitnormal}, by
\begin{equation}
  \label{eq:trafospatmetric}
  \pspatmetr_{i j}=e^{2\Phi}\cspatmetr_{ij}.
\end{equation}
Letting  $\nablab_i$ denote the Levi-Civita connection of the physical metric $\gb_{ij}$, the physical \emph{acceleration} $\pacc_i$ and \emph{extrinsic curvature} $\pextcurv_{ij}$ are uniquely defined via the relations
\begin{equation}
  \label{defphysextr}
  \nablab_i\pnormal_j=-\pacc_j\pnormal_i+\pextcurv_{ij}, \quad \pacc_j \nb^j=0,\quad \pextcurv_{ij} \nb^j=\pextcurv_{ji} \nb^j=0.
\end{equation}
The symmetry $\pextcurv_{ij}=\pextcurv_{ji}$ of the physical extrinsic curvature is a consequence of hypersurface orthogonality of the normal $\nb_i$. 
We see also follows from \eqref{confChrist}, \eqref{defconfextr}, \eqref{eq:trafounitnormal.a} and \eqref{eq:trafospatmetric} that
\begin{align*}  \nablab_i\pnormal_j&=\nabla_i\pnormal_j-g^{kl}(g_{il}\nabla_j \Phi + g_{jl}\nabla_i\Phi -g_{ij}\nabla_l\Phi)\pnormal_k\\
  &= 
  e^\Phi\nabla_i \cnormal_j
  -e^\Phi \cnormal_i\nabla_j \Phi  +e^\Phi g_{ij}\cnormal^l\nabla_l\Phi\\
  &= 
  -\cacc_j\pnormal_i+e^\Phi\cextcurv_{ij}
  - \pnormal_i\cspnabla_j \Phi  
  +e^\Phi \cspatmetr_{ij}\cnormal^l\nabla_l\Phi,
\end{align*}
where $\cspnabla_j$ is the Levi-Civita connection of the conformal spatial metric $\cspatmetr_{ij}$.
Comparing
\eqref{defconfextr} and \eqref{defphysextr}, we obtain 
\begin{align} \label{ndot-bar}
  \pacc_j=\cacc_j+\cspnabla_j \Phi \overset{\eqref{exp(2Phi)}}{=} \cacc_j \AND
  \pextcurv_{ij}=e^\Phi\bigl(\cextcurv_{ij}+\cnormal^l\nabla_l\Phi \cspatmetr_{ij}\bigr).
\end{align}
Raising\footnote{In the following, indices related to conformal tensors are raised and lowered using the conformal metric, e.g.,~ $\cnormal^i=g^{ij}\cnormal_j$, while indices related to physical tensors are raised and lowered using the physical metric, e.g.,~$\pnormal^i=\gb^{ij}\pnormal_j$. Similarly, we raise and lower indices related to conformal spatial and physical spatial tensors using the conformal and physical spatial metrics, respectively, e.g.,~
$\pextcurv_{i}{}^j=\pextcurv_{il}\pspatmetr^{lj}$
and $\cextcurv_{i}{}^j=\cextcurv_{il}\cspatmetr^{lj}$.
} one index of $\pextcurv_{ij}$ yields
\begin{equation}
  \label{eq:pextcurv}
  \pextcurv_{i}{}^j\overset{\eqref{confmet}}{=}e^{-\Phi}\bigl(\cextcurv_{i}{}^j+\cnormal^l\nabla_l\Phi \cspatmetr_{i}{}^j\bigr).
\end{equation}
Decomposing the physical extrinsic curvature $\pextcurv_{ij}$ into its \emph{physical expansion} $\pexp$ and \emph{physical shear} (trace free) $\pshear_{ij}$ as 
\begin{equation}
 \label{eq:physshearexp}
 \pextcurv_{i}{}^j=\pshear_{i}{}^j+\pexp\pspatmetr_{i}{}^j,
\end{equation}
we find, with the help of \eqref{confmet}, \eqref{eq:confexpshear} and \eqref{eq:pextcurv}, that
\begin{equation}
  \label{eq:pexp}
  \pexp=\frac 1{n-1}\pextcurv_{i}{}^i
  =e^{-\Phi}\bigl(\cexp+\cnormal^l\nabla_l\Phi\bigr),
\end{equation}
and
\begin{equation}
  \label{eq:pshear}
  \pshear_{i}{}^j=e^{-\Phi}\bigl(\cextcurv_{i}{}^j+\cnormal^l\nabla_l\Phi \cspatmetr_{i}{}^j\bigr)-e^{-\Phi}\bigl(\cexp+\cnormal^l\nabla_l\Phi\bigr) \pspatmetr_{i}{}^{j}
  =e^{-\Phi}\cshear_{i}{}^j.
\end{equation}
Analogous to \eqref{eq:confshearresc}, it is useful to introduce the \emph{physical scale invariant shear}  $\pshearresc_{ij}$ defined  by
\begin{equation}
  \label{eq:physshearresc}
  \pshearresc_{i}{}^j=\frac{\pshear_{i}{}^j}{\pexp}=\frac{\cshear_{i}{}^j}{\cexp+\cnormal^l\nabla_l\Phi}
  =\frac{\cexp}{\cexp+\cnormal^l\nabla_l\Phi}\cshearresc_{i}{}^j.
\end{equation}

\section{Riemann, Ricci and Weyl curvature tensors}
\label{sec:curvature_tensors}

Let  $\Rb_{ijkl}$ and $\Rb_{ij}$ denote the \emph{physical Riemann and Ricci tensors}\footnote{See Section~\ref{sec:curv_conv} for our curvature conventions.}, respectively, determined by the physical metric $\gb_{ij}$. Then, with the help of \eqref{ESF.1}, \eqref{confmet}, \eqref{taudef}, \eqref{eq:DefVtt} with $V_0=-1$, \eqref{grad-tau}, \eqref{eq:cspatmetrdef}, \eqref{eq:trafounitnormal} and \eqref{eq:pspatmetrdef}, we can express the physical Ricci tensor as
\begin{align}
  \Rb_{ij}=&2\nablab_i\phi\nablab_j \phi+\frac{4}{n-2}V(\phi)\gb_{ij}
  =\frac{2}{\alpha^2}\tau^{-2}\nablab_i\tau\nablab_j \tau-\frac{4}{n-2}\tau^{s/\alpha}\gb_{ij} \notag \\
  =&\frac{2}{\alpha^2}\tau^{-2}\beta^{-2} e^{-2\Phi}\pnormal_{i}\pnormal_{j}
  -\frac{4}{n-2}\tau^{s/\alpha}\gb_{ij}\notag \\
  &=\left(\frac{2}{\alpha^2}\tau^{-2}\beta^{-2}+\frac{4}{n-2}\tau^{s/\alpha}e^{2\Phi}\right)\cnormal_{i}\cnormal_{j}
  -\frac{4}{n-2}\tau^{s/\alpha}e^{2\Phi}\cspatmetr_{ij},   \label{eq:pstRicci}
\end{align}
or equivalently
\begin{equation*}
  \Rb_{i}{}^k=(\tau\beta e^{\Phi})^{-2}\left(\frac{2}{\alpha^2}+\frac{4}{n-2}\tau^{s/\alpha}(\tau\beta e^{\Phi})^2\right)\cnormal_{i} \cnormal^{k}
  -\frac{4}{n-2}\tau^{s/\alpha}\cspatmetr_{i}{}^k.
\end{equation*}
Contracting the $i,k$ indices yields the \textit{physical Ricci scalar}
\begin{equation}
  \label{eq:pstRiccisc}
  \Rb
  =-(\tau\beta e^{\Phi})^{-2}\left(\frac{2}{\alpha^2}+\frac{4n}{n-2}\tau^{s/\alpha}(\tau\beta e^{\Phi})^2\right).
\end{equation}

Following \cite[Ex.~3.2.24]{petersen2016}, we define the \emph{physical Schouton tensor} $\Pb_{ij}$ by
\begin{equation*}
\Pb_{ij}=\frac{2}{n-2}\Bigl(\Rb_{ij}-\frac{\Rb}{2(n-1)}\gb_{ij}\Bigr).
\end{equation*}
The \emph{physical Weyl tensor} $\Wb_{ijkl}$, see 
\cite[Ex.~3.2.25]{petersen2016}, is then defined by\footnote{Note the sign differences in the second and third term on the right side between our definition of the Weyl tensor and the one from \cite{petersen2016}. These sign changes are needed to account for the different Ricci tensor sign conventions employed here and in \cite{petersen2016}.}
\begin{align}
  \Wb_{ijkl}&=\Rb_{ijkl}-\Pb_{i[k} \gb_{l]j}+\Pb_{j[k} \gb_{l]i} \notag \\
  &=\Rb_{ijkl}-\frac{2}{n-2}\bigl(\Rb_{i[k} \gb_{l]j}-\Rb_{j[k} \gb_{l]i}\bigr)
  +\frac{\Rb}{(n-1)(n-2)}\bigl(\gb_{i[k} \gb_{l]j}-\gb_{j[k} \gb_{l]i}\bigr).\label{Weyldef}
\end{align}
From the above expression, it is clear that the physical Riemann tensor $\Rb_{ijkl}$ is completely determined by the Weyl tensor $\Wb_{ijkl}$ and the Ricci tensor $\Rb_{ij}$.  
It is also straightforward to verify that the Weyl tensor is trace free, i.e.,~vanishes upon contracting any two indices, and transforms under the conformal transformation \eqref{confmet} as 
\begin{equation}
  \label{eq:Weylconfinv}
  W_{ijkl}=e^{-2\Phi} \Wb_{ijkl}
\end{equation}
where $W_{ijkl}$ is the conformal Weyl tensor, i.e., the Weyl tensor determined by the conformal metric $g_{ij}$.

In spacetime dimension $n=4$, it is common to decompose the Weyl tensor into its electric and magnetic parts, and express these parts in terms of the spatial curvature, the second fundamental form and matter, e.g.,~see \cite[Eqns.~(1.80) \& (1.81)]{wainwright1997}. While,  this decomposition is not valid in spacetime dimension $n\ge 5$, it is straightforward to derive a related decomposition that holds in all spacetime dimensions $n\geq 3$.  Indeed, from the tracefree property of the Weyl tensor and the symmetries the Weyl tensor inherits from the Riemann curvature tensor, it is not difficult to show that the physical Weyl tensor is fully determined by the projections 
\begin{equation}
  \label{eq:relevantWeylprojections}
  \pWeyl_{mnpq}=\Wb_{ijkl}\pspatmetr^i{}_m\pspatmetr^j{}_n\pspatmetr^k{}_p\pspatmetr^l{}_q
  \quad\text{and}\quad \pWeyl_{mnp}=\Wb_{ijkl}\pspatmetr^i{}_m\pspatmetr^j{}_n\pspatmetr^k{}_p\pnormal^l,
\end{equation}
where $\pnormal^i$ and $\pspatmetr_{ij}$ are as defined by \eqref{eq:trafounitnormal.a} and \eqref{eq:pspatmetrdef}, respectively.
In terms of these projections, the physical Weyl tensor is given by
\begin{align}
  \label{eq:Weylrepresentation}
  \Wb_{ijkl}=&\nb_i\pWeyl_{klj}
  -\nb_j\pWeyl_{kli}
  +\nb_k\pWeyl_{ijl}
  -\nb_l\pWeyl_{ijk}\\
  &+\nb_i\nb_k \pWeyl_{mjpl}\pspatmetr^{mp} 
  +\nb_i\nb_l \pWeyl_{mjkp}\pspatmetr^{mp} 
  +\nb_j\nb_k \pWeyl_{impl}\pspatmetr^{mp}
  +\nb_j\nb_l \pWeyl_{imkp}\pspatmetr^{mp}
  +\pWeyl_{ijkl}\notag
\end{align}
as can be verified using the following identities:
\begin{gather*}
  \Wb_{ijkl}\pnormal^i\pnormal^j\pnormal^k\pnormal^l=0,\\
  \Wb_{ijkl}\pnormal^i\pnormal^j\pnormal^k\pspatmetr^l{}_m=0,
  \quad \Wb_{ijkl}\pnormal^i\pnormal^j\pspatmetr^k{}_m\pnormal^l=0,
  \quad \Wb_{ijkl}\pnormal^i\pspatmetr^j{}_m\pnormal^k\pnormal^l=0,
  \quad \Wb_{ijkl}\pspatmetr^i{}_m\pnormal^j\pnormal^k\pnormal^l=0,\\
  \Wb_{ijkl}\pnormal^i\pnormal^j\pspatmetr^k{}_m\pspatmetr^l{}_n=0,
  \quad \Wb_{ijkl}\pspatmetr^i{}_m\pspatmetr^j{}_n \pnormal^k\pnormal^l=0,
\end{gather*}
\begin{gather*}
  \pWeyl_{mnp}=\Wb_{ijkl}\pspatmetr^i{}_m\pspatmetr^j{}_n\pspatmetr^k{}_p\pnormal^l
  =-\Wb_{ijkl}\pspatmetr^i{}_m\pspatmetr^j{}_n\pnormal^k\pspatmetr^l{}_p
  =\Wb_{ijkl}\pspatmetr^i{}_p\pnormal^j\pspatmetr^k{}_m\pspatmetr^l{}_n
  =-\Wb_{ijkl}\pnormal^i\pspatmetr^j{}_p\pspatmetr^k{}_m\pspatmetr^l{}_n,
\end{gather*}
and
\begin{gather}
  \Wb_{ijkl}\pnormal^i\pspatmetr^j{}_m\pnormal^k\pspatmetr^l{}_n
  =\Wb_{ijkl}\pspatmetr^{ik}\pspatmetr^j{}_m\pspatmetr^l{}_n=\pWeyl_{mnpq}\pspatmetr^{mp},\label{eq:electric}\\
  \Wb_{ijkl}\pnormal^i\pspatmetr^j{}_m\pnormal^k\pspatmetr^l{}_n
  =-\Wb_{ijkl}\pspatmetr^i{}_m \pnormal^j\pnormal^k\pspatmetr^l{}_n
  =-\Wb_{ijkl}\pnormal^i\pspatmetr^j{}_m\pspatmetr^k{}_n \pnormal^l
  =\Wb_{ijkl}\pspatmetr^i{}_m\pnormal^j\pspatmetr^k{}_n \pnormal^l.\notag
\end{gather}

Both tensor fields $\pWeyl_{mnpq}$ and $\pWeyl_{mnp}$ in \eqref{eq:relevantWeylprojections} are purely spatial, and hence, they are completely determined by their spatial frame components, i.e.,~ $\pWeyl_{MNPQ}$ and $\pWeyl_{MNP}$. As a side remark, we can, using \eqref{eq:electric},  define, in all spacetime dimensions $n\geq 3$, the \emph{electric part} of the physical Weyl tensor by
\begin{equation}
  \label{eq:electricweyl}
  \bar E_{mn}=\Wb_{ijkl}\pnormal^i\pspatmetr^j{}_m\pnormal^k\pspatmetr^l{}_n=\pWeyl_{mnpq}\pspatmetr^{mp}.
\end{equation}
On the other hand, there is no counterpart to the definition of the magnetic part of the physical Weyl tensor in spacetime dimensions $n\geq 5$  because it relies on dualisation, which is dimension dependent.

The Gauss and Codazzi relations for the physical Riemann tensor, see \cite{Gourgoulhon:Book}, take the form
\begin{align*}
  \Rb_{ijkl}\pspatmetr^{i}{}_m\pspatmetr^{j}{}_n\pspatmetr^{k}{}_p\pspatmetr^{l}{}_q
  &=\pspatcurv_{mnpq}-\pextcurv_{np}\pextcurv_{mq}+\pextcurv_{mp}\pextcurv_{nq},%\label{Gauss}
  \\
  \Rb_{ijkl}\pspatmetr^{i}{}_m\pspatmetr^{j}{}_n\pspatmetr^{k}{}_p\pnormal^l
  &=\pspnabla_m\pextcurv_{np}-\pspnabla_n\pextcurv_{mp},%\label{Codazzi}
\end{align*}
where $\pspatcurv_{mnpq}$ and $\Db_m$ are the Riemann tensor and Levi-Civita connection determined by the the physical spatial metric $\pspatmetr_{ij}$, and $\Ktt_{mp}$ is the extrinsic curvature determined by determined by the $\tau=\text{constant}$ hypersurfaces and the physical metric $\gb_{ij}$.
Using these relations along with \eqref{Weyldef} allows us to express the spatial tensors $\pWeyl_{mnpq}$ and $\pWeyl_{mnp}$ in the following form:
\begin{align*}
  \pWeyl_{mnpq}
  =&\pspatcurv_{mnpq}-\pextcurv_{np}\pextcurv_{mq}+\pextcurv_{mp}\pextcurv_{nq}
  -\frac{2}{n-2}\bigl(\Rb_{i[k} \gb_{l]j}-\Rb_{j[k} \gb_{l]i}\bigr)\pspatmetr^{i}{}_m\pspatmetr^{j}{}_n\pspatmetr^{k}{}_p\pspatmetr^{l}{}_q\\
  &+\frac{\Rb}{(n-1)(n-2)}\bigl(\gb_{i[k} \gb_{l]j}-\gb_{j[k} \gb_{l]i}\bigr)\pspatmetr^{i}{}_m\pspatmetr^{j}{}_n\pspatmetr^{k}{}_p\pspatmetr^{l}{}_q\\
  =&\pspatcurv_{mnpq}-\pextcurv_{np}\pextcurv_{mq}+\pextcurv_{mp}\pextcurv_{nq}
  -\frac{1}{n-2}\bigl(\Rb_{ik} \pspatmetr_{lj}-\Rb_{jk} \pspatmetr_{li}-\Rb_{il} \pspatmetr_{kj}+\Rb_{jl} \pspatmetr_{ki}\bigr)\pspatmetr^{i}{}_m\pspatmetr^{j}{}_n\pspatmetr^{k}{}_p\pspatmetr^{l}{}_q\\
  &+\frac{\Rb}{(n-1)(n-2)}\bigl(\pspatmetr_{mp} \pspatmetr_{qn}-\pspatmetr_{np} \pspatmetr_{qm}
  \bigr),
\end{align*}
and
\begin{align*}
  \pWeyl_{mnp}
  =&\pspnabla_m\pextcurv_{np}-\pspnabla_n\pextcurv_{mp}
  -\frac{1}{n-2}\bigl(\Rb_{ik} \gb_{lj}-\Rb_{jk} \gb_{li}-\Rb_{il} \gb_{kj}+\Rb_{jl} \gb_{ki}\bigr)\pspatmetr^{i}{}_m\pspatmetr^{j}{}_n\pspatmetr^{k}{}_p\pnormal^l\\
  =&\pspnabla_m\pextcurv_{np}-\pspnabla_n\pextcurv_{mp}
  -\frac{1}{n-2}\bigl(-\Rb_{il} \pspatmetr_{kj}+\Rb_{jl} \pspatmetr_{ki}\bigr)\pspatmetr^{i}{}_m\pspatmetr^{j}{}_n\pspatmetr^{k}{}_p\pnormal^l.
\end{align*}
From these expressions, we then, with the help of \eqref{eq:pstRicci} and \eqref{eq:pstRiccisc}, obtain
\begin{align*}
  \pWeyl_{mnpq}=&\pspatcurv_{mnpq}-\pextcurv_{np}\pextcurv_{mq}+\pextcurv_{mp}\pextcurv_{nq}
  -(\tau e^{\Phi}\beta)^{-2}\frac{2}{(n-1)(n-2)}\Bigl(\frac{1}{\alpha^2}
  -2\tau^{s/\alpha}(\tau e^{\Phi}\beta)^{2}\Bigr)
  \bigl(\pspatmetr_{mp} \pspatmetr_{qn}-\pspatmetr_{np} \pspatmetr_{qm}\bigr)
  \end{align*}
and
\begin{equation*}
  \pWeyl_{mnp}
  =\pspnabla_m\pextcurv_{np}-\pspnabla_n\pextcurv_{mp}.
\end{equation*}
On the other hand, we have, by \eqref{eq:physshearexp} and \eqref{eq:physshearresc}, that   
\begin{align*}
  -\pextcurv_{np}\pextcurv_{mq}+&\pextcurv_{mp}\pextcurv_{nq}
  =-(\pshear_{np}+\pexp\pspatmetr_{np})(\pshear_{mq}+\pexp\pspatmetr_{mq})
  +(\pshear_{mp}+\pexp\pspatmetr_{mp})(\pshear_{nq}+\pexp\pspatmetr_{nq})\\
  &=\pexp^2(\pspatmetr_{mp}\pspatmetr_{nq}-\pspatmetr_{np}\pspatmetr_{mq})
  -(\pshear_{np}+\pexp\pspatmetr_{np})\pshear_{mq}
  -\pshear_{np}\pexp\pspatmetr_{mq}
  +(\pshear_{mp}+\pexp\pspatmetr_{mp})\pshear_{nq}
  +\pshear_{mp}\pexp\pspatmetr_{nq}\\
  &=\pexp^2(\pspatmetr_{mp}\pspatmetr_{nq}-\pspatmetr_{np}\pspatmetr_{mq})
  -\pexp^2\Bigl((\pshearresc_{np}+\pspatmetr_{np})\pshearresc_{mq}
  +\pshearresc_{np}\pspatmetr_{mq}
  -(\pshearresc_{mp}+\pspatmetr_{mp})\pshearresc_{nq}
  -\pshearresc_{mp}\pspatmetr_{nq}\Bigr),
\end{align*}
and
\begin{align*}
  \pspnabla_m\pextcurv_{np}-\pspnabla_n\pextcurv_{mp}
  &=\pspnabla_m (\pexp(\pshearresc_{np}+\pspatmetr_{np}))-\pspnabla_n(\pexp(\pshearresc_{mp}+\pspatmetr_{mp}))\\
  &=\pspnabla_m \pexp(\pshearresc_{np}+\pspatmetr_{np})+\pexp \pspnabla_m \pshearresc_{np}
  -\pspnabla_n \pexp(\pshearresc_{mp}+\pspatmetr_{mp})-\pexp \pspnabla_n \pshearresc_{mp}\\
  &=\pexp\Bigl((\pshearresc_{np}+\pspatmetr_{np})\pspnabla_m \log\pexp+ \pspnabla_m \pshearresc_{np}
  -(\pshearresc_{mp}+\pspatmetr_{mp})\pspnabla_n \log\pexp-\pspnabla_n \pshearresc_{mp}\Bigr).
\end{align*}
Putting everything together yields
\begin{align} 
  \pWeyl_{mnpq}=\pspatcurv_{mnpq}
  &-\pexp^2\Bigl((\pshearresc_{np}+\pspatmetr_{np})\pshearresc_{mq}
  -\pshearresc_{np}\pspatmetr_{mq}
  +(\pshearresc_{mp}+\pspatmetr_{mp})\pshearresc_{nq}
  +\pshearresc_{mp}\pspatmetr_{nq}\Bigr)\notag \\
  &-(\tau^{}e^{\Phi}\beta)^{-2}\frac{1}{(n-1)(n-2)}\biggl(\frac{2}{\alpha^2}
  -4\tau^{s/\alpha}(\tau e^{\Phi}\beta)^2
  \notag \\
  &-(n-1)(n-2)(\tau\beta e^\Phi\pexp)^2
  \biggr)
  \bigl(\pspatmetr_{mp} \pspatmetr_{qn}-\pspatmetr_{np} \pspatmetr_{qm}\bigr) \label{Weyl1}
\end{align}
and
\begin{equation}
  \label{Weyl2}
  \pWeyl_{mnp}
  =\pexp\Bigl((\pshearresc_{np}+\pspatmetr_{np})\pspnabla_m \log\pexp+ \pspnabla_m \pshearresc_{np}
  -(\pshearresc_{mp}+\pspatmetr_{mp})\pspnabla_n \log\pexp-\pspnabla_n \pshearresc_{mp}\Bigr).
\end{equation}

\bigskip

%\noindent \textbf{Acknowledgements}

%\section{Statements and declarations}

%\subsection*{Funding} 
%This research was partially supported by the Marsden Fund Council from Government
%funding, managed by Royal Society Te Apārangi (Award Number: MFP-UOO2322).
%This work was partially supported by the Marsden Fund grant MFP-UOO2322.

%\subsection*{Data Availability}  This article has no associated data.

%\subsection*{Conflict of interest} The authors declare no conflict of interest.

\bibliographystyle{amsplain}
%\bibliography{bibliography, refs}
\bibliography{big-bang_scalarfield_potential}

\end{document}